\documentclass[11pt, letterpaper]{article}

\usepackage{graphicx}
\usepackage[margin=1in]{geometry}
\usepackage{enumitem}
\usepackage{mleftright}
\usepackage{dsfont}
\usepackage{mathtools}
\usepackage{amsmath, amssymb, amsthm}
\usepackage{thmtools}
\usepackage{mathabx}
\usepackage{bm}
\usepackage{nicefrac}       
\usepackage{microtype}      
\usepackage{comment}
\usepackage{bbm}
\usepackage{aligned-overset}
\usepackage{ifthen}
\usepackage[colorinlistoftodos,textsize=scriptsize]{todonotes}

  \newcounter{constant} 
  \newcommand{\newconstant}[1]{\refstepcounter{constant}\label{#1}} 
  \newcommand{\useconstant}[1]{C_{\ref{#1}}}

\makeatletter
\def\namedlabel#1#2{\begingroup
    #2%
    \def\@currentlabel{#2}%
    \phantomsection\label{#1}\endgroup
}
\makeatother

\PassOptionsToPackage{hyphens}{url}
\usepackage[backref=page]{hyperref}
\hypersetup{
    colorlinks,
    citecolor=blue,
    filecolor=blue,
    linkcolor=blue,
    urlcolor=blue
}
\renewcommand*{\backref}[1]{}
\renewcommand*{\backrefalt}[4]{%
    \ifcase #1 %
    \or     #2%
    \else   #2%
    \fi
}
\usepackage{listings}
\usepackage{fancyhdr}

\usepackage[T1]{fontenc}
\usepackage{framed, color}
\usepackage{lmodern}
\usepackage{bigints}
\usepackage{cancel}
\allowdisplaybreaks
\DeclareSymbolFont{rmlargesymbols}{OMX}{mdbch}{m}{n}
\DeclareMathSymbol{\rmintop}{\mathop}{rmlargesymbols}{82}
\newcommand{\rmint}{\rmintop\nolimits}

\AtBeginDocument{%
  \DeclareFontShape{T1}{lmr}{m}{scit}{<->ssub*lmr/m/scsl}{}%
}

\usepackage{algorithm}
\usepackage{algpseudocode}

\let\originalleft\left
\let\originalright\right
\renewcommand{\left}{\mathopen{}\mathclose\bgroup\originalleft}
\renewcommand{\right}{\aftergroup\egroup\originalright}

\definecolor{codegreen}{rgb}{0,0.6,0}
\definecolor{codegray}{rgb}{0.5,0.5,0.5}
\definecolor{codepurple}{rgb}{0.58,0,0.82}
\definecolor{backcolour}{rgb}{0.95,0.95,0.92}

\lstdefinestyle{mystyle}{
    backgroundcolor=\color{backcolour},   
    commentstyle=\color{codegreen},
    keywordstyle=\color{magenta},
    numberstyle=\tiny\color{codegray},
    stringstyle=\color{codepurple},
    basicstyle=\footnotesize,
    breakatwhitespace=false,         
    breaklines=true,                 
    captionpos=b,                    
    keepspaces=true,                 
    numbers=left,                    
    numbersep=5pt,                  
    showspaces=false,                
    showstringspaces=false,
    showtabs=false,                  
    tabsize=2
}
 
\makeatletter
\newcommand{\AlgoResetCount}{\renewcommand{\@ResetCounterIfNeeded}{\setcounter{AlgoLine}{0}}}
\newcommand{\AlgoNoResetCount}{\renewcommand{\@ResetCounterIfNeeded}{}}
\newcounter{AlgoSavedLineCount}

\makeatother

\usepackage{tikz}
\tikzset{every picture/.style={line width=0.75pt}} 

\newcommand{\N}{{\rm I\!N}}

\newcommand{\bR}{\mathbb{R}}
\newcommand{\bS}{\mathbb{S}}

\newcommand{\id}{\mathbb{I}}

\newcommand{\cB}{\mathcal{B}}

\newcommand{\cD}{\mathcal{D}}

\newcommand{\cM}{\mathcal{M}}
\newcommand{\cN}{\mathcal{N}}
\newcommand{\cO}{\mathcal{O}}

\newcommand{\cU}{\mathcal{U}}

\newcommand{\cX}{\mathcal{X}}
\newcommand{\cY}{\mathcal{Y}}

\newcommand{\Beta}{\mathrm{Beta}}

\newcommand{\diag}{\mathrm{diag}}

\newcommand{\Score}{\mathrm{Score}}

\newcommand{\supp}{\mathrm{supp}}

\newcommand{\tr}{\mathrm{tr}}

\newcommand{\TV}{d_{\mathrm{TV}}}

\newcommand{\CovAwMeanEst}{\textsc{CovAwMeanEst}}

\newcommand{\UnbndGaussSamp}{\textsc{UnbndGaussSampl}}

\newcommand{\StableCov}{\textsc{StableCov}}
\newcommand{\StableMean}{\textsc{StableMean}}
\newcommand{\LargestGS}{\textsc{LargestGoodSubset}}
\newcommand{\LargestCore}{\textsc{LargestCore}}

\newcommand{\Pass}{\texttt{Pass}}
\newcommand{\Fail}{\texttt{Fail}}
\newcommand{\Out}{\texttt{OUT}}

\newcommand{\iprod}[1]{\left\langle #1 \right\rangle}
\newcommand{\paren}[1]{\left( #1 \right)}
\newcommand{\brk}[1]{\left[ #1 \right]}
\newcommand{\brc}[1]{\left\{ #1 \right\}}
\newcommand{\cl}[1]{\left\lceil #1 \right\rceil}
\newcommand{\fl}[1]{\left\lfloor #1 \right\rfloor}
\newcommand{\pr}[2]{
    \ifthenelse{\equal{#1}{}}{
        \mathbb{P}\left[ #2 \right]
    }
    {
        \underset{#1}{\mathbb{P}}\left[ #2 \right]
    }
}
\newcommand{\ex}[2]{
    \ifthenelse{\equal{#1}{}}{
        \mathbb{E}\left[ #2 \right]
    }
    {
        \underset{#1}{\mathbb{E}}\left[ #2 \right]
    }
}
\newcommand{\var}[2]{
    \ifthenelse{\equal{#1}{}}{
        \mathrm{Var}\left( #2 \right)
    }
    {
        \underset{#1}{\mathrm{Var}}\left( #2 \right)
    }
}
\newcommand{\llnorm}[1]{\left\| #1 \right\|_2}

\newcommand{\fnorm}[1]{\left\| #1 \right\|_F}
\newcommand{\mnorm}[2]{\left\| #2 \right\|_{#1}}

\newcommand{\abs}[1]{\left| #1 \right|}

\newcommand{\eps}{\varepsilon}

\definecolor{shadecolor}{rgb}{0.83, 0.83, 0.83}
\theoremstyle{plain}
\newtheorem{thm}{Theorem}[section]
\newtheorem{theorem}[thm]{Theorem}
\newtheorem{fact}[thm]{Fact}
\newtheorem{proposition}[thm]{Proposition}
\newtheorem{corollary}[thm]{Corollary}
\newtheorem{lemma}[thm]{Lemma}

\theoremstyle{definition}
\newtheorem{definition}[thm]{Definition}

\newtheorem{remark}[thm]{Remark}

\title{Optimal Differentially Private Sampling of Unbounded Gaussians\thanks{Authors are listed in alphabetical order.}}
\author{
    Valentio Iverson\thanks{Cheriton School of Computer Science, University of Waterloo. {\tt viverson@uwaterloo.ca}.}
\and
    Gautam Kamath\thanks{Cheriton School of Computer Science, University of Waterloo and Vector Institute. {\tt g@csail.mit.edu}. Supported by a Canada CIFAR AI Grant, an NSERC Discovery Grant, and an unrestricted gift from Google.}
\and
    Argyris Mouzakis\thanks{Cheriton School of Computer Science, University of Waterloo and Vector Institute. {\tt amouzaki@uwaterloo.ca}. Supported by a Canada CIFAR AI Grant, an NSERC Discovery Grant, an unrestricted gift from Google, and the Onassis Foundation (Scholarship ID: F ZT 053-1/2023-2024).}
}

\begin{document}

\maketitle

\begin{abstract}
We provide the first $\widetilde{\cO}\paren{d}$-sample algorithm for sampling from unbounded Gaussian distributions under the constraint of $\paren{\eps, \delta}$-differential privacy. This is a quadratic improvement over previous results for the same problem, settling an open question of~\cite{GhaziHKM23}.
\end{abstract}

\section{Introduction}
\label{sec:intro}

Given a sample from a probability distribution, output a sample from the same probability distribution. 
As stated, this problem is trivial: one can simply output the sample received.
However, when \emph{privacy} is a concern, this is clearly inappropriate, as the sample may contain sensitive personal information.
To address this issue, one may turn to the concept of \emph{differential privacy} (DP)~\cite{DworkMNS06}, a rigorous notion of privacy that ensures that the output of a procedure does not reveal too much information about any individual's data.

One well-studied option is to privately \emph{learn} the underlying distribution.
That is, given samples from a distribution, output a (differentially private) estimate of it.
After this estimated distribution is released, one may freely output any desired number of samples from the distribution without any further privacy concerns (appealing to the post-processing property of differential privacy).

However, this strategy may be costly.
Indeed, in both the private and non-private setting, learning a distribution is clearly a much harder problem than simply outputting a single sample from it. 
In an influential paper, Raskhodnikova, Sivakumar, Smith, and Swanberg~\cite{RaskhodnikovaSSS21} introduced the problem of privately \emph{sampling} from a distribution. 
Given a collection of samples from a distribution, privately output a single sample from (approximately) the same underlying distribution.
The authors provided upper and lower bounds for private sampling on discrete distributions and binary product distributions, showing that, in various cases, private sampling is either no easier or dramatically easier than private learning.

We note that this problem has immediate implications for machine learning in the context of private generative models.
Rather than privately training a generative model (which may be data-intensive), one could instead use private sampling if only a limited number of samples (e.g., prompt completions) are needed.

Subsequently, Ghazi, Hu, Kumar, and Manurangsi~\cite{GhaziHKM23} studied private sampling of Gaussian distributions.
Technically speaking, the most interesting case is when we are sampling from distributions with unknown covariance matrix $\Sigma$.
The authors provide two algorithms for privately sampling from $d$-dimensional Gaussians subject to $\paren{\eps, \delta}$-DP:
\begin{enumerate}
    \item If $\id \preceq \Sigma \preceq K \id$ (i.e., the covariance is bounded), they give an algorithm with sample complexity $\widetilde{\cO}\paren{\frac{d K^2}{\eps}}$;
    \item Without any assumptions (i.e., a potentially unbounded covariance), they give an algorithm with sample complexity $\widetilde{\cO}\paren{\frac{d^2}{\eps}}$.
\end{enumerate}
They also prove a sample complexity lower bound of $\widetilde{\Omega}\paren{\frac{d}{\eps}}$ (for the case $\id \preceq \Sigma \preceq 2 \id$), thus establishing that, when the covariance is bounded up to constant factors, their algorithm is near-optimal.
However, this leaves open the case where $K = \omega\paren{1}$.
For example, if $K = \Omega\paren{\sqrt{d}}$, then the best known upper bound of $\widetilde{\cO}\paren{d^2}$ is off from the lower bound of $\widetilde{\Omega}\paren{d}$ by a quadratic factor.
Is it always a choice between an algorithm whose sample complexity depends linearly on the dimension $d$ (but also suffers from a dependence on the ``range'' of the covariance $K$), and one whose sample complexity is quadratic in $d$?
Or is there an $\widetilde{\cO}\paren{d}$-sample algorithm, regardless of how large $K$ may be?

\subsection{Results and Techniques}
\label{subsec:results_techniques}

Our main result is an algorithm with $\widetilde{\cO}\paren{d}$ sample complexity, for private sampling from Gaussians with arbitrary (i.e., potentially unbounded) covariance.

\begin{theorem}
\label{thm:main_thm_informal}
There exists an $\paren{\eps, \delta}$-differentially private algorithm such that, given:
\[
    n = \cO\paren{\frac{\log\paren{\frac{1}{\delta}}}{\eps} \paren{d + \log\paren{\frac{\log\paren{\frac{1}{\delta}}}{\alpha \eps}}}},
\]
samples drawn i.i.d.\ from a distribution $\cN\paren{\mu, \Sigma}$, outputs a sample from a Gaussian distribution which is $\alpha$-close to $\cN\paren{\mu, \Sigma}$.
\end{theorem}

Observe that, in terms of its dependence on $d$ in the sample complexity, this nearly matches the $\Omega\paren{\frac{d}{\eps \sqrt{\log d}}}$ lower bound of~\cite{GhaziHKM23} for the same problem.
Additionally, it is quadratically less than the sample complexity of privately \emph{learning} Gaussians in the same setting~\cite{KamathLSU19,AdenAliAK21,KamathMSSU22,AshtianiL22,KothariMV22,KamathMS22}, which is $\widetilde{\Theta}\paren{\frac{d^2}{\alpha^2} + \frac{d^2}{\alpha \eps}}$.

We give a brief technical overview of our approach.
As a starting point, we use the algorithm of~\cite{GhaziHKM23} which privately samples from a Gaussian with covariance $\id \preceq \Sigma \preceq K \id$ using $\widetilde{\cO}\paren{d K^2}$ samples.
Recall the following basic property of Gaussians: given $n$ samples $X_1, \dots, X_n$ from $\cN\paren{\mu, \Sigma}$, their average $\frac{1}{n} \sum\limits_{i \in \brk{n}} X_i$ is distributed as $\cN\paren{\mu, \frac{1}{n} \Sigma}$.
Simply outputting this average is insufficient for two reasons: i) the variance of the resulting distribution is too small by a factor of $n$, and ii) the algorithm is not differentially private. 
Ideally, one could solve both problems at once by adding $\cN\paren{0, \frac{n-1}{n} \Sigma}$ noise, but recall that $\Sigma$ itself is unknown.
As a proxy for this,~\cite{GhaziHKM23} uses the \emph{samples themselves} to sample from this distribution.
Another basic property of Gaussians is $2$-stability: it implies that, if $z \in \bR^n$ is a random unit vector and $X_1, \dots, X_n$ are samples from a Gaussian, then $\sum\limits_{i \in \brk{n}} z_i X_i$ is a sample from the same Gaussian.\footnote{See Lemma~\ref{lem:stability_gaussian} for a proof of this property.}
Thus, letting $Y_i \coloneqq \frac{1}{\sqrt{2}} \paren{X_{n + i} - X_{n + m + i}} \sim \cN\paren{0, \Sigma}$, the random variable $\frac{1}{n} \sum\limits_{i \in \brk{n}} X_i + \sqrt{1 - \frac{1}{n}} \sum\limits_{i \in \brk{m}} z_i Y_i$ is distributed as $\cN\paren{\mu, \Sigma}$.
To complete the privacy analysis, two additional components are needed: a) ``clipping'' the datapoints to a predetermined range (based on the upper bound $\Sigma \preceq K \id$) and b) a ``Propose-Test-Release'' (PTR) step~\cite{DworkL09} to ensure that the collection of $Y_i$'s have large (empirical) variance in all directions (which is implied with high probability given the lower bound $\id \preceq \Sigma$). 
Unfortunately, due to these two steps, the resulting sample complexity is $\Omega\paren{d K^2}$.

We modify this general framework, turning to recent results of Brown, Hopkins, and Smith~\cite{BrownHS23}.\footnote{A concurrent work to theirs of Kuditipudi, Duchi, and Haque~\cite{KuditipudiDH23} obtains a similar result for covariance-aware mean estimation via similar techniques. We focus on~\cite{BrownHS23} since we directly appeal to their primitives; it is conceivable that methods in~\cite{KuditipudiDH23} could have been used instead.}
Their main result is an $\widetilde{\cO}\paren{d}$ sample algorithm for covariance-aware private mean estimation -- that is, mean estimation where the error in the estimate may scale proportional to the underlying distribution's variance in that direction. 
To do this, they design algorithms \StableMean\ and \StableCov.
Focusing on \StableMean\ (\StableCov\ is similar), this algorithm takes in a dataset $X_1, \dots, X_n$, and outputs a vector of weights $v_1, \dots, v_n$ and a score $s$.\footnote{As in given in~\cite{BrownHS23}, \StableMean\ outputs a non-private estimate for the mean of the distribution which is a weighted version of the empirical mean that uses the weight vector $v$.
We have slightly changed the format of the output, since it facilitates aspects of our exposition.
For more details on this, we point readers to Section~\ref{sec:bhs23}.}
This algorithm has two key properties: a) if the data is truly Gaussian, then the weights will be uniform (i.e., $v_i = \frac{1}{n}$ for all $i$) with high probability, and b) if the score $s$ is sufficiently small (which can be verified using a PTR step), then the weighted empirical mean $\sum\limits_{i \in \brk{n}} v_i X_i$ will have bounded sensitivity when modifying one datapoint.
We highlight that bounded sensitivity is a property of the overall sum, and individual summands may change dramatically if a single datapoint is changed -- this is in contrast to clipping-based strategies where individual summands also have bounded sensitivity.
These elegant primitives allow one to perform private mean estimation without any explicit clipping.

We employ these primitives in the framework of~\cite{GhaziHKM23}: we output the quantity $\sum\limits_{i \in \brk{n}} v_i X_i + \sqrt{m \paren{1 - \frac{1}{n}}} \sum\limits_{i \in \brk{m}} z_i \sqrt{w_i} Y_i$, where the $v_i$'s and $w_i$'s are the weights outputtted by \StableMean\ and \StableCov, respectively. 
That is, we incorporate these weights to bound the sensitivity, as an alternative to clipping. 
We also require a preceding PTR step in order to ensure that the scores output by \StableMean\ and \StableCov\ are both sufficiently small. 
To perform the utility analysis, it suffices to show that the weights are uniform with high probability, and the analysis is similar to before.
The privacy analysis however requires newfound care.

We illustrate in the univariate case (i.e., $Y_i \in \bR$), focusing exclusively on the term $\sum\limits_{i \in \brk{m}} z_i U_i$, where we use $U_i$ to denote $\sqrt{w_i} Y_i$. 
We accordingly define $U_i'$, which is defined with respect to a neighboring dataset $Y_1', \dots, Y_n'$.
~\cite{GhaziHKM23} previously used a coupling-based argument to reduce to examining only the term involving the $Y_i$ that differs (i.e., the datapoint such that $Y_i \neq Y_i'$).
But, as mentioned before, changing a single $Y_i$ may result in changes to \emph{all} $U_i$'s.
Since the stability property of \StableCov\ tells us that the empirical variance with these weights has bounded sensitivity, all we know is that $\sum\limits_{i \in \brk{m}} U_i^2 \in \paren{1 \pm \frac{1}{m}} \sum\limits_{i \in \brk{m}} U_i'^2$.
Our goal is to argue that, for fixed $U$ and $U'$, $\iprod{z, U}$ is similar in distribution to $\iprod{z, U'}$, where all the randomness is in the choice of the random unit vector $z \in \bR^m$. 
Perhaps surprisingly, if $U$ and $U'$ satisfy the aforementioned bounded sensitivity property, then this suffices to prove the desired statement.
To this end, it is convenient to think of $U \in \bR^m$ as a vector.
By rotational symmetry, $\iprod{z, U}$ is equal in distribution to $\llnorm{U} \iprod{z, e_1}$, where $e_1$ is the first standard basis vector, and similarly, $\iprod{z, U'}$ is equal to $\llnorm{U'} \iprod{z, e_1}$.
It is not too hard to argue that the two transformed distributions are close to each other (as required in Hockey-Stick divergence), as $\iprod{z, e_1}$ can be seen to be a Beta distribution, and the two distributions are scalings of this distribution by similar constants. 
This sketch exposes some of the key ideas in the univariate case for the ``variance'' term of the output, handling the general case requires additional care.

\subsection{Related Work}
\label{subsec:related_work}

Private sampling is still a relatively nascent topic. 
In addition to the aforementioned works~\cite{RaskhodnikovaSSS21, GhaziHKM23}, Cheu and Nayak~\cite{CheuN25} study private \emph{multi-sampling}, which aims to privately draw \emph{multiple} samples from a distribution, rather than just one.

The related problem of private distribution learning is a much more mature topic of study. 
Specifically, a long line of works investigates private learning of Gaussian distributions, e.g.,~\cite{KarwaV18, BunKSW19, BunS19, KamathLSU19, BiswasDKU20, DuFMBG20, CaiWZ21, HuangLY21, HopkinsKMN23, BieKS22, BenDavidBCKS23, AumullerLNP23, AlabiKTVZ23, AsiUZ23, LiuKKO21, LiuKO22, KamathMRSSU23, AdenAliAK21, KamathMSSU22, AshtianiL22, KothariMV22, KamathMS22, Narayanan23, PortellaH24}.
Some works focus on the case where the covariance matrix is potentially unbounded.
While packing lower bounds imply that learning is not possible for $\paren{\eps, 0}$-DP~\cite{HardtT10, BunKSW19, HopkinsKM22}, several works successfully show that learning is possible (despite unbounded covariance) under $\paren{\eps, \delta}$-DP, including~\cite{KarwaV18} in the univariate case, and~\cite{AdenAliAK21, KamathMSSU22, KothariMV22, AshtianiL22, AlabiKTVZ23} for the multivariate case. 
Our work is most related to the aforementioned works on covariance-aware mean estimation~\cite{BrownHS23,KuditipudiDH23}, which can accurately estimate the mean with $\widetilde{\cO}\paren{d}$ samples even when the covariance is unknown.
That is, we rely heavily on the tools introduced by~\cite{BrownHS23}, hinting at the potential of deeper technical connections between private estimation and private sampling.

As described above, private sampling can be seen as a less cumbersome alternative to private distribution learning, when only a limited number of samples are required.
There is an analogous line of work on private prediction (see, e.g.,~\cite{PapernotAEGT17, DworkF18, BassilyTT18}), which can be seen as a more lightweight alternative to private PAC-learning, when only a limited number of predictions are required.

A long line of empirical work studies private learning of deep generative models, see, e.g.,~\cite{XieLWWZ18, BeaulieuJonesWWLBBG19, CaoBVFK21, BieKZ23, DockhornCVK23, HarderJSP23}.
Some recent empirical works study private sampling from generative models~\cite{LinGKNY24, XieLBGYINJZLLY24}, a morally similar problem to ours, though these works generally require having enough data to train an effective non-private model. 

An interesting related question in the non-private setting is that of \emph{sample amplification}~\cite{AxelrodGSV20}.
Given $n$ samples from a distribution, can one generate $m \gg n$ samples from (approximately) the same distribution? 
This line of work served as an inspiration for private sampling~\cite{RaskhodnikovaSSS21}.

\subsection{Organization of the Paper}
\label{subsec:organization}

We start by covering the necessary background and definitions in Section~\ref{sec:prelim}, with particular emphasis on linear algebra, statistics, and privacy preliminaries.
Then, in Section~\ref{sec:bhs23}, we present the covariance-aware mean estimator of~\cite{BrownHS23} (Algorithm~\ref{alg:bhs23}).
Finally, in Section~\ref{sec:sampling_algo}, we present our near-optimal algorithm for privately sampling from unbounded Gaussian distributions (Algorithm~\ref{alg:sampling}), and establish its guarantees.

\section{Preliminaries}
\label{sec:prelim}

This section is split into two parts.
Section~\ref{subsec:math_prelim} introduces some basic notation, while also including statements of basic facts from linear algebra and probability \& statistics to which we will appeal later, while Section~\ref{subsec:priv_prelim} introduces the necessary background on privacy-related notions.

\subsection{General Notation and Basic Facts}
\label{subsec:math_prelim}

\textbf{Linear Algebra Preliminaries.}
We write $\brk{n} \coloneqq \brc{1, \dots, n}$ and $\brk{\alpha \pm R} \coloneqq \brk{\alpha - R, \alpha + R}$.
Also, given a vector $v \in \bR^d$, $v_i$ refers to its $i$-th component, while $v_{- i} \in \bR^{d - 1}$ refers to the vector one gets by removing the $i$-th component, and $v_{\le i} \in \bR^{i}$ ($v_{> i} \in \bR^{d - i}$) describes the vector one gets by keeping (removing) the first $i$ components.
The set of all $d$-dimensional unit vectors is denoted by $\bS^{d - 1}$, while the set of the standard basis vectors of $\bR^d$ is denoted by $\brc{e_i}_{i \in \brk{d}}$.
For a matrix $M \in \bR^{n \times m}$, $M_{i j}$ describes the element at the $i$-th row and $j$-th column, and by \emph{singular values} of $M$ we refer to the square roots of the eigenvalues of the matrix $M^{\top} M \in \bR^{m \times m}$.
Furthermore, $\llnorm{M}$ is the \emph{spectral norm} (the largest singular value of the matrix), $\fnorm{M}$ is the \emph{Frobenius norm} (the square root of the sum of the squares of the singular values), and $\mnorm{\tr}{M}$ is the \emph{trace norm} (the sum of the singular values).
Additionally, for $M \in \bR^{d \times d}$, given a positive-definite matrix $\Sigma \in \bR^{d \times d}$, the \emph{Mahalanobis norm} of $M$ with respect to $\Sigma$ is $\mnorm{\Sigma}{M} \coloneqq \fnorm{\Sigma^{- \frac{1}{2}} M \Sigma^{- \frac{1}{2}}}$.
Given a pair of symmetric matrices $A, B \in \bR^{d \times d}$, we write $A \succeq B$ if and only if $x^{\top} \paren{A - B} x \geq 0, \forall x \in \bR^d$.
Finally, by $\id_{n \times n}$ we will denote the identity matrix of size $n \times n$, and by $0_{n \times m}$ we will denote the all $0$s matrix of size $n \times m$, where the subscripts will be dropped if the size is clear from the context.
For additional background on linear algebra, we refer readers to Appendix~\ref{sec:lin_alg_facts}.

\medskip

\textbf{Probability \& Statistics Preliminaries.}
Given a sample space $\cX$, we denote the set of all distributions over $\cX$ by $\Delta\paren{\cX}$.
By $\supp\paren{\cD}$ we denote the \emph{support} of $\cD$, i.e., the set over which $\cD$ assigns positive density.
For any distribution $\cD \in \Delta\paren{\cX}$, $\cD^{\otimes n}$ denotes the product distribution over $\cX^n$, where each marginal distribution is $\cD$.
Thus, given a dataset $X \coloneqq \paren{X_1, \dots, X_n}$ drawn i.i.d.\ from $\cD$, we write $X \sim \cD^{\otimes n}$.
Whenever using the symbol of probability, we may use a subscript what the randomness is over in cases where it might not be clear from the context, e.g., $\pr{X}{\cdot}$.
Some of the distributions that we will encounter are the Gaussian distribution with mean $\mu$ and covariance $\Sigma$ (denoted by $\cN\paren{\mu, \Sigma}$), the uniform distribution over a set $\cX$ (denoted by $\cU\paren{\cX}$), and, given a vector $v \sim \cU\paren{\bS^{n - 1}}$, the distribution of its first $i$ components (denoted by $v_{\le i} \sim \cU\paren{\bS^{n - 1}}_{\le i}$).
We will use the either $X \overset{d}{=} Y$ or $X \sim Y$ when the random variables $X$ and $Y$ are \emph{equal in distribution}.

A fundamental concept that will feature heavily in this work is that of \emph{$f$-divergences}, which are used to quantify the distance between distributions.
We now give the formal definition of this notion.

\begin{definition}
\label{def:f_div}
Let $P, Q \in \Delta\paren{\cX}$ such that $\supp\paren{P} \subseteq \supp\paren{Q}$.\footnote{If $\supp\paren{P} \not\subseteq \supp\paren{Q}$ the definition is still valid, but can result in $D_f\paren{P \middle\| Q} = \infty$.
However, for the $f$-divergences considered in present work, this will not be the case, and we will be able to calculate them using simplified versions of the general definition without concern about the distributions' supports.}
Also, let $f \colon \bR_{\geq 0} \to \left(- \infty, \infty\right]$ be a convex function such that $\abs{f\paren{x}} < \infty, \forall x > 0$, $f\paren{1} = 0$, and $\lim\limits_{x \to 0^+} f\paren{x} = 0$.
Then, the $f$-divergence from $P$ to $Q$ is defined as:
\[
    D_f\paren{P \middle\| Q} \coloneqq \rmint\limits_{\cX} f\paren{\frac{dP}{dQ}} dQ. 
\]
We call $f$ the \emph{generator} of $D_f$.
\end{definition}

An $f$-divergence that will be of particular importance to this work is the \emph{Hockey-Stick divergence}.
We define it below.

\begin{definition}
\label{def:hs_div}
Let $P, Q \in \Delta\paren{\cX}$ with density functions $p, q$, respectively.
For $\eps \geq 0$, we define the Hockey-Stick divergence from $P$ to $Q$ is defined as:
\[
    D_{e^{\eps}}\paren{P \middle\| Q} \coloneqq \max\limits_{S \subseteq \cX} \brc{P\paren{S} - e^{\eps} Q\paren{S}} = \frac{1}{2} \rmint\limits_{\cX} \abs{p\paren{x} - e^{\eps} q\paren{x}} \, dx - \frac{1}{2} \paren{e^{\eps} - 1}.
\]
\end{definition}

\begin{remark}
\label{rem:hs_div_equiv}
The two equivalent expressions given above for the Hockey-Stick divergence correspond to different generators $f$.
Specifically, the former definition corresponds to $f\paren{x} \coloneqq \max\brc{x - e^{\eps}, 0}, \forall x \in \bR_{\geq 0}$, whereas the latter corresponds to $f\paren{x} \coloneqq \frac{1}{2} \abs{x - e^{\eps}} - \frac{1}{2} \paren{e^{\eps} - 1}, \forall x \in \bR$.
Verifying the equivalence of the two definitions is a simple exercise, where it suffices to consider the second definition and partition the domain depending on whether $p\paren{x} > e^{\eps} q\paren{x}$ or otherwise.
Additionally, both definitions reduce to \emph{TV-distance} when $\eps = 0$.
\end{remark}

For additional facts from probability and statistics, we point readers to Appendix~\ref{sec:prob_facts}.

\subsection{Privacy Preliminaries}
\label{subsec:priv_prelim}

We define differential privacy here and introduce some of its most fundamental properties.

\begin{definition}[Approximate Differential Privacy (DP)~\cite{DworkMNS06}]
\label{def:dp}
A randomized algorithm $M \colon \cX^n \rightarrow \cY$ satisfies $\paren{\eps, \delta}$-DP if for every pair of neighboring datasets $X, X' \in \cX^n$ (i.e., datasets that differ in at most one entry), we have:
\[
    \pr{M}{M\paren{X} \in Y} \le e^{\eps} \pr{M}{M\paren{X'} \in Y} + \delta, \forall Y \subseteq \cY.
\]
\end{definition}

\begin{remark}
\label{rem:pure_dp}
When $\delta = 0$, we say that $M$ satisfies \emph{pure-DP}.
\end{remark}

\begin{remark}
\label{rem:dp_hs}
We note that both the definition of DP can be captured using the Hockey-Stick divergence introduced earlier.
Indeed, the definition of $\paren{\eps, \delta}$-DP is equivalent to saying that, for every pair of adjacent datasets $X, X'$, we have $D_{e^{\eps}}\paren{M\paren{X} \middle\| M\paren{X'}} \le \delta$.
\end{remark}

\begin{definition}[DP under condition~\cite{KothariMV22}]
\label{def:dp_under_condition}
Let $\Psi \colon \cX^n \to \brc{0, 1}$ be a predicate.
An algorithm $M \colon \cX^n \to \cY$ is \emph{$\paren{\eps, \delta}$-DP under condition $\Psi$} for $\eps, \delta > 0$ if, for every neighboring datasets $X, X' \in \cX^n$ both satisfying $\Psi$, we have:
\[
    \pr{M}{M\paren{X} \in Y} \le e^{\eps} \pr{M}{M\paren{X'} \in Y} + \delta, \forall Y \subseteq \cY.
\]
\end{definition}

\begin{lemma}[Composition for Algorithm with Halting~\cite{KothariMV22}]
\label{lem:composition_halting}
Let $M_1: \cX^n \to \cY_1 \cup \brc{\Fail}, M_2: \cY_1 \times \cX^n \to \cY_2$ be algorithms.
Furthermore, let $M$ denote the following algorithm:  let $y_1 \coloneqq M_1\paren{X}$ and, if $y_1 = \Fail$, then halt and output \Fail\ or else, output $y_2 \coloneqq M_2\paren{y_1, X}$.

Let $\Psi$ be any condition such that, if $X$ does not satisfy $\Psi$, then $M_1\paren{X}$ always returns $\Fail$.
Suppose that $M_1$ is $\paren{\eps_1, \delta_1}$-DP, and $M_2$ is $\paren{\eps_2, \delta_2}$-DP under condition $\Psi$.
Then, $M$ is $\paren{\eps_1 + \eps_2, \delta_1 + \delta_2}$-DP.
\end{lemma}

\section{The Algorithm of~\texorpdfstring{\cite{BrownHS23}}{} and its Properties}
\label{sec:bhs23}

We now present the main algorithm given in~\cite{BrownHS23} for covariance-aware mean estimation.
The algorithm's goal is to privately obtain an estimate $\widetilde{\mu}$ for the mean of a Gaussian (or sub-Gaussian) distribution with mean vector $\mu$ and covariance matrix $\Sigma$ that satisfies $\mnorm{\Sigma}{\widetilde{\mu} - \mu} \le \alpha$ with high probability.
Roughly speaking, the algorithm involves producing a pair of non-private estimates $\paren{\widehat{\mu}, \widehat{\Sigma}}$ for the mean and the covariance of the distribution, respectively.
These estimates come from \emph{stable} estimators.
This is a class of algorithms which, given well-behaved input datasets (i.e., datasets that do not contain significant outliers), produce outputs which are resilient to small perturbations of the input dataset, despite having no formal robustness and/or privacy guarantees.
A PTR step follows, which privately determines whether the results produced by the algorithms correspond to a ``good'' dataset.
If the estimates fail the test, the algorithm aborts.
Conversely, if the test is successful, a point is sampled from the distribution $\cN\paren{\widehat{\mu}, c^2 \widehat{\Sigma}}$,\footnote{$c$ is a multiplicative factor that exists so that the magnitude of the noise is large enough to ensure privacy.}
which comes with the desired utility and privacy guarantees when the dataset is appropriately large.\footnote{Sampling a point from $\cN\paren{\widehat{\mu}, c^2 \widehat{\Sigma}}$ essentially amounts to an implementation of the Gaussian mechanism with noise which, instead of being spherical, depends on the shape of $\widehat{\Sigma}$, i.e., $\cN\paren{\widehat{\mu}, c^2 \widehat{\Sigma}} = \widehat{\mu} + \widehat{\Sigma}^{\frac{1}{2}} \cN\paren{0, c^2 \id}$.
This first appeared in~\cite{BrownGSUZ21} as a means of performing covariance-aware mean estimation (termed the \emph{Emprically Rescaled Gaussian Mechanism}), and then in~\cite{AlabiKTVZ23} for covariance estimation as part of the \emph{Gaussian Sampling Mechanism}.}

\begin{algorithm}[H]
    \caption{Covariance-Aware Mean Estimator}\label{alg:bhs23}
    \hspace*{\algorithmicindent} \textbf{Input:} Dataset $X \coloneqq \paren{X_1, \dots, X_n}^{\top} \in \bR^{n \times d}$; privacy parameters $\eps, \delta > 0$; outlier threshold $\lambda_0 \geq 1$. \\
    \hspace*{\algorithmicindent} \textbf{Output:} $\widetilde{\mu} \in \bR^d \cup \brc{\Fail}$.
    \begin{algorithmic}[1]
    \Procedure{\CovAwMeanEst$_{\paren{\eps, \delta}, \lambda_0}$}{$X$}
        \State \Comment{Initialize.}
        \State Let $k \gets \cl{\frac{6 \log\paren{\frac{6}{\delta}}}{\eps}} + 4$; $M \gets 6 k + \cl{18 \log\paren{\frac{16 n}{\delta}}}$; $c^2 \gets \frac{720 e^2 \lambda_0 \log\paren{\frac{12}{\delta}}}{\eps^2 n^2}$.
        \State \Comment{Compute stable estimates.}
        \State Let $\paren{W, \Score_1} \gets \StableCov_{\lambda_0, k}\paren{X}$.\Comment{See Appendix~\ref{sec:bhs23_details} for details.}
        \State Let $\widehat{\Sigma} \gets W W^{\top}$.
        \State Let $R \sim \cU\paren{\brc{R' \subseteq \brk{n} \colon \abs{R'} = M}}$.
        \State Let $\paren{v, \Score_2} \gets \StableMean_{\widehat{\Sigma}, \lambda_0, k, R}\paren{X}$.\Comment{See Appendix~\ref{sec:bhs23_details} for details.}
        \State Let $\widehat{\mu} \gets \sum\limits_{i \in \brk{n}} v_i X_i$.
        \State \Comment{PTR step.}
        \If {$\cM_{\mathrm{PTR}}^{\paren{\frac{\eps}{3}, \frac{\delta}{6}}}\paren{\max\brc{\Score_1, \Score_2}} = \Pass$}
            \State Let $\widetilde{\mu} \sim \cN\paren{\widehat{\mu}, c^2 \widehat{\Sigma}}$.
            \State \Return $\widetilde{\mu}$.
        \Else
            \State Let $\widetilde{\mu} \gets \Fail$.
            \State \Return $\widetilde{\mu}$.
        \EndIf
    \EndProcedure
    \end{algorithmic}
\end{algorithm}

Algorithm~\ref{alg:bhs23} uses a pair of subroutines (\StableCov\ and \StableMean), which correspond to the stable estimators mentioned earlier.
These algorithms function by assigning weights to points depending on how far they are from other points in the dataset (essentially performing a \emph{soft outlier removal} operation).
For brevity of exposition, we will not give the pseudocode of these algorithms here, instead deferring the full presentation to Appendix~\ref{sec:bhs23_details}.
However, we will note that we have made a slight modification to the format of the output of these algorithms.
Rather than outputting the weighted sums that correspond to the stable estimates, \StableCov\ outputs a matrix $W \in \bR^{d \times n}$ whose columns are of the form $\sqrt{w_i} Y_i$ where $Y_i \coloneqq \frac{1}{\sqrt{2}} \paren{X_i - X_{i + \fl{\frac{n}{2}}}}$, while \StableMean\ outputs a vector of weights, and the estimates $\paren{\widehat{\mu}, \widehat{\Sigma}}$ are constructed separately in a post-processing step.
The goal of this modification is to facilitate the use of these algorithms in subsequent sections.

For the rest of this section, we will focus on some lemmas which will prove to be particularly important for our analysis in Section~\ref{sec:sampling_algo}.
These are related to the \emph{stability} guarantees of \StableCov\ and \StableMean, as well as the way they interact with the PTR step.
They are crucial to the analysis of~\cite{BrownHS23} for establishing the privacy guarantees in their main result.
We note that lemmas related to the accuracy guarantees given in~\cite{BrownHS23} will not be discussed at any point in this manuscript, the reason being that the utility guarantee in the present work (closeness of output distribution to input distribution in TV-distance) differs significantly from the target guarantee in~\cite{BrownHS23}, and we will not be applying their results.

We stress that the versions of the statements we will give here are slightly different from the original ones.
Indeed, the original formulations assume that \StableMean\ and \StableCov\ are fed the same dataset.
However, in Section~\ref{sec:sampling_algo}, this will not be the case.
For that reason, we have made the smallest possible number of changes to the statements to ensure that they are compatible with how they are used in subsequent sections.

We first state a result of~\cite{BrownHS23} that captures the effect of changing one input datapoint has on the output of \StableCov, under the assumption that both the original dataset and the new one are well-concentrated (which is quantified by the score outputted by \StableCov).

\begin{lemma}[Lemma $11$ from~\cite{BrownHS23}]
\label{lem:stable_cov}
Fix dataset size $2 n_2$, dimension $d$, outlier threshold $\lambda_0 \geq 1$, and privacy parameters $\eps, \delta \in \brk{0, 1}$.
Set the discretization parameter $k$ to be $\cl{\frac{6 \log\paren{\frac{6}{\delta}}}{\eps}} + 4$ (as in Algorithm~\ref{alg:bhs23}).
Let $X, X'$ be adjacent $d$-dimensional datasets of size $2 n_2$.
Assume $\gamma \coloneqq \frac{8 e^2 \lambda_0}{n_2} \le \frac{1}{2 k}$ (i.e., $n_2 \geq 16 e^2 \lambda_0 k$).
Let:
\begin{align*}
    \paren{W_1, \Score}
    &\coloneqq \StableCov\paren{X, \lambda_0, k}, \\
    \paren{W_2, \Score'}
    &\coloneqq \StableCov\paren{X', \lambda_0, k}.
\end{align*}
Assume $\Score, \Score' < k$.
Then $\Sigma_1 \coloneqq W_1 W_1^{\top}, \Sigma_2 \coloneqq W_2 W_2^{\top} \succ 0$ and $\paren{1 - \gamma} \Sigma_1 \preceq \Sigma_2 \preceq \frac{1}{1 - \gamma} \Sigma_1$.
Additionally, we have:
\[
    \mnorm{\tr}{\Sigma_1^{- \frac{1}{2}} \Sigma_2 \Sigma_1^{- \frac{1}{2}} - \id}, \mnorm{\tr}{\Sigma_2^{- \frac{1}{2}} \Sigma_1 \Sigma_2^{- \frac{1}{2}} - \id} \le \paren{1 + 2 \gamma} \gamma.
\]
\end{lemma}

The previous is complemented by an analogous guarantee for \StableMean.
However, before stating the guarantee for \StableMean, we first need to introduce the notion of a \emph{degree-representative} subset of a dataset that was given in~\cite{BrownHS23}.
The motivation behind this definition has to do with the fact that Algorithm~\ref{alg:bhs23} uses a \emph{reference set $R$} in \StableMean.
This is a subset of the dataset that is used to efficiently implement the soft outlier removal process described earlier.

\begin{definition}[Definition $32$ from~\cite{BrownHS23}]
\label{def:deg_rep}
Fix $d$-dimensional dataset $X$ of size $n_1$, covariance $\Sigma \succ 0$, outlier threshold $\lambda_0 \geq 1$, and reference set $R \subseteq \brk{n_1}$.
For all $i$, let:
\begin{align*}
    N_i
    &\coloneqq
    \brc{j \in \brk{n_1} \colon \mnorm{\Sigma}{X_i - X_j}^2\le e^2 \lambda_0}, \\
    \widetilde{N}_i
    &\coloneqq \brc{j \in R \colon \mnorm{\Sigma}{X_i - X_j}^2 \le e^2 \lambda_0}.
\end{align*}
Let $z_i \coloneqq \abs{N_i}$ and $\widetilde{z}_i \coloneqq \abs{\widetilde{N}_i}$ be the sets' sizes.
We say that $R$ is \emph{degree-representative for $X$} if for every index $i$ we have $\abs{\frac{1}{\abs{R}} \widetilde{z}_i - \frac{1}{n_1} z_i} \le \frac{1}{6}$.
\end{definition}

Having given the above definition, we now proceed to state the lemma capturing the stability guarantees of \StableMean.

\begin{lemma}[Lemma $13$ from~\cite{BrownHS23}]
\label{lem:stable_mean}
Fix dataset sizes $n_1$ and $2 n_2$, dimension $d$, outlier threshold $\lambda_0 \geq 1$, and privacy parameters $\eps, \delta \in \brk{0, 1}$.
Set the discretization parameter $k$ to be $\cl{\frac{6 \log\paren{\frac{6}{\delta}}}{\eps}} + 4$ (as in Algorithm~\ref{alg:bhs23}).
Use reference set $R \subseteq \brk{n_1}$ with $\abs{R} > 6 k$.
Let $X$ and $X'$ be adjacent $d$-dimensional datasets of size $n_1$.
Let $\Sigma_1, \Sigma_2 \in \bR^{d \times d}$ be positive definite matrices satisfying $\paren{1 - \gamma} \Sigma_1 \preceq \Sigma_2 \preceq \frac{1}{1 - \gamma} \Sigma_1$ for $\gamma \coloneqq \frac{8 e^2 \lambda_0}{n_2}$.
Assume $n_1 \geq 32 e^2 k$ and $\gamma \le \frac{1}{2 k}$ (i.e., $n_2 \geq 16 e^2 \lambda_0 k$).
Let:
\begin{align*}
    \paren{v, \Score}
    &\coloneqq \StableMean\paren{X, \Sigma_1, \lambda_0, k, R}, \\
    \paren{v', \Score'}
    &\coloneqq \StableMean\paren{X', \Sigma_2, \lambda_0, k, R}.
\end{align*}
If $\Score, \Score' < k$ and $R$ is degree-representative for both $X$ and $X'$, then:
\[
    \mnorm{\Sigma_1}{\widehat{\mu} - \widehat{\mu}'}^2 \le \frac{\paren{1 + 2 \gamma} 38 e^2 \lambda_0}{n_1^2},
\]
where $\widehat{\mu} \coloneqq \sum\limits_{i \in \brk{n}} v_i X_i$ and $\widehat{\mu}' \coloneqq \sum\limits_{i \in \brk{n}} v_i' X_i'$.
\end{lemma}

We conclude this section with two results that capture the function of the PTR step of Algorithm~\ref{alg:bhs23}.
The first result describes the conditions under which the PTR step works.

\begin{lemma}[Claim $10$ from~\cite{BrownHS23}]
\label{lem:bhs_ptr}
Fix $0 < \eps \le 1$ and $0< \delta \le \frac{\eps}{10}$.
There is an algorithm $\cM_{\mathrm{PTR}}^{\paren{\eps, \delta}} \colon \bR \to \brc{\Pass, \Fail}$ that satisfies the following conditions:
\begin{enumerate}
    \item Let $\cU$ be a set and $g \colon \cU^n \to \bR_{\ge 0}$ a function.
    If, for all $X, X' \in \cU^n$ that differ in one entry,  $\abs{g\paren{X} - g\paren{X'}} \le 2$, then $\cM_{\mathrm{PTR}}^{\paren{\eps, \delta}}\paren{g\paren{\cdot}}$ is $\paren{\eps, \delta}$-DP.
    \item $\cM_{\mathrm{PTR}}^{\paren{\eps, \delta}}\paren{0} = \Pass$.
    \item For all $z \geq \frac{2 \log\paren{\frac{1}{\delta}}}{\eps} + 4$, $\cM_{\mathrm{PTR}}^{\paren{\eps, \delta}}\paren{z} = \Fail$.
\end{enumerate}
\end{lemma}

The second result describes how the guarantees of Lemmas~\ref{lem:stable_cov} and~\ref{lem:stable_mean} suffice to satisfy the conditions of Lemma~\ref{lem:bhs_ptr}.

\begin{lemma}[Lemma $9$ from~\cite{BrownHS23}]
\label{lem:stable_conditions_ptr}
Fix dataset size $n$, dimension $d$, outlier threshold $\lambda_0 \geq 1$, reference set $R \subseteq \brk{n_1}$, and privacy parameters $\eps, \delta \in \brk{0, 1}$.
Set $k \coloneqq \cl{\frac{6 \log\paren{\frac{6}{\delta}}}{\eps}} + 4$ (as in Algorithm~\ref{alg:bhs23}).
Assume $n \coloneqq n_1 + 2 n_2$ with $n_1 \geq 32 e^2 k$ and $n_2 \geq 16 e^2 \lambda_0 k$.
Let $X$ and $X'$ be adjacent $d$-dimensional datasets of size $n$.
Let:
\begin{align*}
    \paren{W, \Score_1}
    &\coloneqq \StableCov\paren{X_{> n_1}, \lambda_0, k}, \\
    \paren{W', \Score_1'}
    &\coloneqq \StableCov\paren{X_{> n_1}', \lambda_0, k}, \\
    \paren{v, \Score_2}
    &\coloneqq \StableMean\paren{X_{\le n_1}, W W^{\top}, \lambda_0, k, R}, \\
    \paren{v', \Score_2'}
    &\coloneqq \StableMean\paren{X_{\le n_1}', W' W'^{\top}, \lambda_0, k, R}.
\end{align*}
Then $\abs{\max\brc{\Score_1, \Score_2} - \max\brc{\Score_1', \Score_2'}} \le 2$.
\end{lemma}

\section{Near-Optimal Sampling from Unbounded Gaussian Distributions}
\label{sec:sampling_algo}

In this section, we present our sampler, and establish its properties.
The section is split into three parts.
In Section~\ref{subsec:sampler} we present the algorithm and state its guarantees.
Then, in Section~\ref{subsec:uti_proof}, we establish the utility guarantee.
Subsequently, in Section~\ref{subsec:priv_proof}, we show the privacy guarantee.
Finally, in Section~\ref{subsec:wrap}, we tie everything together and identify the final sample complexity.

\subsection{The Sampler and its Guarantees}
\label{subsec:sampler}

Here, we give our sampler for Gaussian distributions with unbounded parameters.
Our algorithm combines elements of Algorithm~\ref{alg:bhs23} with Algorithm $2$ of~\cite{GhaziHKM23}, as was highlighted in Section~\ref{subsec:results_techniques}.
We start by giving the pseudocode here.
\newconstant{C1}
\newconstant{C2}

\begin{algorithm}[H]
    \caption{Unbounded Gaussian Sampler}\label{alg:sampling}
    \hspace*{\algorithmicindent} \textbf{Input:} Dataset $X \coloneqq \paren{X_1, \dots, X_n}^{\top} \in \bR^{n \times d}$; error $\alpha \in \brk{0, 1}$; privacy parameters $\eps, \delta > 0$. \\
    \hspace*{\algorithmicindent} \textbf{Output:} $Z \in \bR^d \cup \brc{\Fail}$.
    \begin{algorithmic}[1]
    \Procedure{\UnbndGaussSamp$_{\paren{\eps, \delta}, \alpha}$}{$X$}
        \State \Comment{Initialize.}
        \State Let $\lambda_0 \gets 4 d + 8 \sqrt{d \ln\paren{\frac{3 n}{\alpha}}} + 8 \ln\paren{\frac{3 n}{\alpha}}$.
        \State Let $n_1 \gets \useconstant{C1} \frac{\sqrt{\lambda_0} \log\paren{\frac{1}{\delta}}}{\eps}$ for a sufficiently large constant $\useconstant{C1} \geq 1$.\Comment{Require $n = n_1 + 2 n_2$.}
        \State Let $n_2 \gets \useconstant{C2} \frac{\lambda_0 \log\paren{\frac{1}{\delta}}}{\eps}$ for a sufficiently large constant $\useconstant{C2} \geq 1$.
        \State Let $k \gets \cl{\frac{6 \log\paren{\frac{6}{\delta}}}{\eps}} + 4$; $M \gets 6 k + \cl{18 \log\paren{\frac{16 n}{\delta}}}$.\label{ln:init}
        \State \Comment{Calculate stable estimates.}
        \State Let $\paren{W, \Score_1} \gets \StableCov_{\lambda_0, k}\paren{X_{> n_1}}$.
        \State Let $\widehat{\Sigma} \gets W W^{\top}$.
        \State Let $R \sim \cU\paren{\brc{R' \subseteq \brk{n} \colon \abs{R'} = M}}$.\label{ln:repr_set}
        \State Let $\paren{v, \Score_2} \gets \StableMean_{\widehat{\Sigma}, \lambda_0, k, R}\paren{X_{\le n_1}}$.
        \State Let $\widehat{\mu} \gets \sum\limits_{i \in \brk{n_1}} v_i X_i$.
        \State \Comment{PTR step.}
        \If {$\cM_{\mathrm{PTR}}^{\paren{\frac{\eps}{3}, \frac{\delta}{6}}}\paren{\max\brc{\Score_1, \Score_2}} = \Pass$}\label{ln:ptr}
            \State Let $z \sim \cU\paren{\bS^{n_2 - 1}}$.\label{ln:post_ptr1}
            \State Let $Z \gets \widehat{\mu} + \sqrt{\paren{1 - \frac{1}{n_1}} n_2} W z$.\Comment{$W z = \sum\limits_{i \in \brk{n_2}} z_i \sqrt{w_i} Y_i$ for $Y_i \coloneqq \frac{1}{\sqrt{2}} \paren{X_{n_1 + i} - X_{n_1 + n_2 + i}}$.}
            \State \Return $Z$.\label{ln:post_ptr2}
        \Else
            \State Let $Z \gets \Fail$.
            \State \Return $Z$.
        \EndIf
    \EndProcedure
    \end{algorithmic}
\end{algorithm}

We now formally state the guarantees of Algorithm~\ref{alg:sampling}.

\begin{theorem}
\label{thm:main_thm_formal}
For any $\alpha \in \brk{0, 1}, \eps \in \brk{0, 1}, \delta \in \brk{0, \frac{\eps}{10}}$, when given:
\[
    n = \cO\paren{\frac{\log\paren{\frac{1}{\delta}}}{\eps} \paren{d + \log\paren{\frac{\log\paren{\frac{1}{\delta}}}{\alpha \eps}}}},
\]
input samples, Algorithm~\ref{alg:sampling} has the following guarantees:
\begin{itemize}
    \item \textbf{Privacy.}
    Algorithm~\ref{alg:sampling} is $(\eps, \delta)$-DP.
    
    \item \textbf{Utility.}
    Let $\cN(\mu, \Sigma)$ be any Gaussian distribution in $d$-dimensions with full-rank covariance.
    Then, given a dataset $X \sim \cN(\mu, \Sigma)^{\otimes n}$, Algorithm~\ref{alg:sampling} has runtime polynomial with respect to the size of the input, and produces output $Z \in \bR^d \cup \brc{\Fail}$ such that $\TV(Z, \cN(\mu, \Sigma)) \le \alpha$.
\end{itemize}
\end{theorem}

\begin{remark}
\label{rem:main_thm_full_rank}
Theorem~\ref{thm:main_thm_formal} is stated in terms of Gaussians with full-rank covariances $\Sigma$.
However, the result holds even for degenerate ones, with a slight modification of the algorithm.
Specifically, we would have to add a pre-processing step that employs a subspace-finding algorithm (e.g.,~\cite{AshtianiL22}) to identify the subspace corresponding to the directions of non-zero variance, and then project everything to that subspace.
These extra steps require $\cO\paren{\frac{d \log\paren{\frac{1}{\delta}}}{\eps}}$ samples, and run in polynomial time (see Theorem $4.1$ of~\cite{AshtianiL22}), while all subsequent steps will require the same number of samples that Theorem~\ref{thm:main_thm_formal} predicts, but $d$ will be replaced by the rank of $\Sigma$.
Thus, the sample and time complexity guarantees of the resulting algorithm are within the same order of magnitude as the ones given in Theorem~\ref{thm:main_thm_formal}.
\end{remark}

\begin{remark}
\label{rem:main_thm_time}
We do not perform an explicit analysis of the time complexity of our algorithm, but the claim of polynomial runtime should be clear.
Indeed, Algorithm~\ref{alg:sampling} can essentially be viewed as a direct combination of the operations performed by the algorithms given in~\cite{BrownHS23} and~\cite{GhaziHKM23}, both of which have polynomial runtime.
\end{remark}

The rest of this work will be devoted to proving Theorem~\ref{thm:main_thm_formal}.
The proof of the utility guarantee will be given in Section~\ref{subsec:uti_proof} (Theorem~\ref{thm:uti_guarantee}).
Then, the proof of the privacy guarantee will be given in Section~\ref{subsec:priv_proof} (Theorem~\ref{thm:priv_guarantee}).
Finally, in Section~\ref{subsec:wrap}, we will put things together to identify the final sample complexity, and complete the proof of Theorem~\ref{thm:main_thm_formal}.

\subsection{Establishing the Utility Guarantee}
\label{subsec:uti_proof}

In this section, we will show that, if the input dataset is drawn from $\cN\paren{\mu, \Sigma}$, and $\lambda_0 \geq 4 d + 8 \sqrt{d \ln\paren{\frac{3 n}{\alpha}}} + 8 \ln\paren{\frac{3 n}{\alpha}}$, the output $Z$ of Algorithm~\ref{alg:sampling} will satisfy $\TV\paren{Z, \cN\paren{\mu, \Sigma}} \le \alpha$.
Showing this will require us to argue that, with high probability, there will be no significant outliers in the dataset.
This follows from the strong concentration properties of the Gaussian distribution.
In that case, the weights assigned to datapoints by \StableCov\ and \StableMean\ will be uniform.
Intuitively, this yields:
\[
    \widehat{\mu} = \frac{1}{n_1} \sum\limits_{i \in \brk{n_1}} X_i \sim \cN\paren{\mu, \frac{1}{n_1} \Sigma} \text{ and } W z = \frac{1}{\sqrt{n_2}} \sum\limits_{i \in \brk{n_2}} z_i Y_i \sim \cN\paren{0, \frac{1}{n_2} \Sigma} \implies Z \sim \cN\paren{\mu, \Sigma}.
\]
We stress that the above argument is not formal by any means.
Indeed, as mentioned above, for the weights to be uniform, we conditioned on the dataset being well-concentrated.
This conditioning affects the input distribution, making it a \emph{truncated Gaussian} instead of a regular Gaussian.
Despite this, as we will see, the intuition is correct, and will lead us to a formal result that goes through an application of the \emph{joint convexity} property of $f$-divergences (Fact~\ref{fact:joint_convexity}).

We start working now towards making the above intuition more formal.
Let $\widebar{\mu}$ and $\widebar{\Sigma}$ denote the standard empirical mean and empirical covariance of $X$ and $Y$ respectively, i.e., $\widebar{\mu} \coloneqq \frac{1}{n_1} \sum\limits_{i \in \brk{n_1}} X_i$ and $\widebar{\Sigma} \coloneqq \frac{1}{n_2} \sum\limits_{i \in \brk{n_2}} Y_i Y_i^{\top}$.
In a series of lemmas, we will identify sufficient conditions which establish that, with probability at least $1 - \alpha$, we have $\widehat{\mu} = \widebar{\mu}$ and $\widehat{\Sigma} = \widebar{\Sigma}$, i.e., $v_i = \frac{1}{n_1}, \forall i \in \brk{n_1}$ and $w_i = \frac{1}{n_2}, \forall i \in \brk{n_2}$, and the PTR step of Line~\ref{ln:ptr} succeeds.
To help us in the proof, we define the following three events of interest:
\begin{align*}
    E_1
    &\coloneqq \brc{\mnorm{\Sigma}{X_i - \mu}^2 \le \frac{\lambda_0}{4}, \forall i \in \brk{n_1}}, \\
    E_2
    &\coloneqq \brc{\mnorm{\Sigma}{Y_i}^2 \le \frac{\lambda_0}{4}, \forall i \in \brk{n_2}}, \\  
    E_3
    &\coloneqq \brc{\llnorm{\Sigma^{\frac{1}{2}} \widebar{\Sigma}^{-1} \Sigma^{\frac{1}{2}} - \id} \le 2}, \\
    E
    &\coloneqq E_1 \cap E_2 \cap E_3.
\end{align*}
Our goal is to show that $E$ is a sufficient condition for the weights outputted by both \StableCov\ and \StableMean\ to be uniform, and to have $\Score_1 = \Score_2 = 0$.
However, before doing so, we identify sufficient conditions for each of the events $E_1, E_2$, and $E_3$ to hold with probability at least $1 - \frac{\alpha}{3}$.
A union bound then implies that $E$ must hold with probability at least $1 - \alpha$.

We start by identifying a bound on the outlier threshold $\lambda_0$ that is used by \StableCov\ and \StableMean\ that ensures that each of the events $E_1$ and $E_2$ occurs with probability at least $1 - \frac{\alpha}{3}$.

\begin{lemma}
\label{lem:utility1}
Let $X \coloneqq \paren{X_1, \dots, X_{n_1}} \sim \cN\paren{\mu, \Sigma}^{\otimes n_1}$ and $Y \coloneqq \paren{Y_1, \dots, Y_{n_2}} \sim \cN\paren{0, \Sigma}^{\otimes n_2}$.
Then, for $n \coloneqq n_1 + 2 n_2$, given $\lambda_0 \geq 4 d + 8 \sqrt{d \ln\paren{\frac{3 n}{\alpha}}} + 8 \ln\paren{\frac{3 n}{\alpha}}$, we have $\pr{}{E_1} \geq 1 - \frac{\alpha}{3}$ and $\pr{}{E_2} \geq 1 - \frac{\alpha}{3}$.
\end{lemma}

\begin{proof}
We will show the result for $E_1$, and the proof for $E_2$ is entirely analogous.
We note that, for all $i \in \brk{n_1}$, we have $X_i - \mu \sim \cN\paren{0, \Sigma} \implies \Sigma^{- \frac{1}{2}} \paren{X_i - \mu} \sim \cN\paren{0, \id}$.
By a union bound and Fact~\ref{fact:hd_gaussian_tail}, we get:
\[
    \pr{}{E_1^c} \le n_1 \pr{Z \sim \cN\paren{0, \id}}{\llnorm{Z}^2 \geq \frac{\lambda_0}{4}} \le n_1 e^{-\frac{\paren{\sqrt{\frac{\lambda_0}{2} - d} - \sqrt{d}}^2}{4}} \le \frac{\alpha}{3},
\]
where the last inequality follows from our bound on $\lambda_0$.
\end{proof}

We now proceed to identify a bound on $n_2$ such that $E_3$ holds.

\begin{lemma}
\label{lem:utility2}
Let $Y \coloneqq \paren{Y_1, \dots, Y_{n_2}} \sim \cN\paren{0, \id}^{\otimes n_2}$.
For $n_2 \geq \cO\paren{d + \log\paren{\frac{1}{\alpha}}}$, we have $\pr{}{E_3} \geq 1 - \frac{\alpha}{3}$.
\end{lemma}

\begin{proof}
We note that $\Sigma^{- \frac{1}{2}} \widebar{\Sigma} \Sigma^{- \frac{1}{2}} = \frac{1}{n_2} \sum\limits_{i \in \brk{n_2}} \paren{\Sigma^{- \frac{1}{2}} X_i} \paren{\Sigma^{- \frac{1}{2}} X_i}^{\top}$, where $\Sigma^{- \frac{1}{2}} X_i \sim \cN\paren{0, \id}$.
By Fact~\ref{fact:gaussian_spec_conc} and our sample complexity bound, we have with probability at least $1 - \frac{\alpha}{3}$:
\begin{align*}
    \frac{1}{3} \id \preceq \Sigma^{- \frac{1}{2}} \widebar{\Sigma} \Sigma^{- \frac{1}{2}} \preceq \frac{5}{3} \id \iff \frac{3}{5} \id \preceq \Sigma^{\frac{1}{2}} \widebar{\Sigma}^{- 1} \Sigma^{\frac{1}{2}} \preceq 3 \id
    &\iff - \frac{2}{5} \id \preceq \Sigma^{\frac{1}{2}} \widebar{\Sigma}^{- 1} \Sigma^{\frac{1}{2}} - \id \preceq 2 \id \\
    &\implies \llnorm{\Sigma^{\frac{1}{2}} \widebar{\Sigma}^{- 1} \Sigma^{\frac{1}{2}} - \id} \le 2,
\end{align*}
yielding the desired result.
\end{proof}

We now leverage the previous two lemmas to establish that, with high probability, the output of \StableCov\ will coincide with the standard empirical covariance $\widebar{\Sigma}$ and $\Score_1$ will be $0$.

\begin{lemma}
\label{lem:stable_cov_uti}
Let $X \coloneqq \paren{X_1, \dots, X_n} \sim \cN\paren{\mu, \Sigma}^{\otimes n}$, where $n \coloneqq n_1 + 2 n_2$ and $n_2 \geq \cO\paren{d + \log\paren{\frac{1}{\alpha}}}$.
Also, let $\lambda_0 \geq 4 d + 8 \sqrt{d \ln\paren{\frac{3 n}{\alpha}}} + 8 \ln\paren{\frac{3 n}{\alpha}}$.
Then, except with probability $\frac{2 \alpha}{3}$, $E_2 \cap E_3$ occurs, and the weights $w_i$ assigned to datapoints by $\StableCov_{\lambda_0, k}\paren{X_{> n_1}}$ will satisfy $w_i = \frac{1}{n_2}, \forall i \in \brk{n_2}$, and we will have $\widebar{\Sigma} = \widehat{\Sigma}$, and $\Score_1 = 0$.
\end{lemma}

\begin{proof}
The assumptions of Lemmas~\ref{lem:utility1} and~\ref{lem:utility2} are satisfied, thus a union bound implies that $E_2 \cap E_3$ is realized, except with probability $\frac{2 \alpha}{3}$.
For the following, we condition over this event.
Under this assumption, it suffices to prove that \LargestGS\ (Algorithm~\ref{alg:largest_gs} in Appendix~\ref{sec:bhs23_details}) will return the entire set $\brk{n_2}$ every time it is called by \StableCov\ (Algorithm~\ref{alg:stablecov} in Appendix~\ref{sec:bhs23_details}).
For this to happen, it suffices to have $\llnorm{\widebar{\Sigma}^{- \frac{1}{2}} Y_i}^2 \le \lambda_0, \forall i \in \brk{n_2}$, where $Y_i \coloneqq \frac{1}{\sqrt{2}} \paren{X_{n_1 + i} - X_{n_1 + n_2 + i}} \sim \cN\paren{0, \Sigma}$.
We note that the definition of the spectral norm implies the upper bound:
\[
    \llnorm{\widebar{\Sigma}^{- \frac{1}{2}} Y_i}^2 = \llnorm{\paren{\widebar{\Sigma}^{- \frac{1}{2}} \Sigma^{\frac{1}{2}}} \paren{\Sigma^{- \frac{1}{2}} Y_i}}^2 \le \llnorm{\widebar{\Sigma}^{- \frac{1}{2}} \Sigma^{\frac{1}{2}}}^2 \llnorm{\Sigma^{- \frac{1}{2}} Y_i}^2 = \llnorm{\widebar{\Sigma}^{- \frac{1}{2}} \Sigma^{\frac{1}{2}}}^2 \mnorm{\Sigma}{Y_i}^2.
\]
Thus, to establish the desired result, we upper-bound each term in the above separately.
The fact that $E_2$ is realized implies that the second term is upper-bounded by $\frac{\lambda_0}{4}$.
As for the first term, we have again by the definition of the spectral norm:
\begin{equation}
    \llnorm{\widebar{\Sigma}^{- \frac{1}{2}} \Sigma^{\frac{1}{2}}}^2 = \llnorm{\Sigma^{\frac{1}{2}} \widebar{\Sigma}^{- 1} \Sigma^{\frac{1}{2}}} \le \llnorm{\Sigma^{\frac{1}{2}} \widebar{\Sigma}^{- 1} \Sigma^{\frac{1}{2}} - \id} + \llnorm{\id} \le 3, \label{eq:matrix_prod_spec_ub}
\end{equation}
where we appealed to the triangle inequality and the fact that $E_3$ is assumed to occur.

The desired results follow directly by combining the previous upper bounds.
\end{proof}

The previous is complemented by a lemma which establishes that, with high probability, the output of \StableMean\ will coincide with the standard empirical mean $\widebar{\mu}$ and $\Score_2$ will be $0$.

\begin{lemma}
\label{lem:stable_mean_uti}
Let $X \coloneqq \paren{X_1, \dots, X_n} \sim \cN\paren{\mu, \Sigma}^{\otimes n}$, where $n \coloneqq n_1 + 2 n_2$ and $n_2 \geq \cO\paren{d + \log\paren{\frac{1}{\alpha}}}$.
Also, let $\lambda_0 \geq 4 d + 8 \sqrt{d \ln\paren{\frac{3 n}{\alpha}}} + 8 \ln\paren{\frac{3 n}{\alpha}}$.
Then, except with probability $\alpha$, $E$ occurs, implying that the weight vector $v$ outputted by $\StableMean_{\widehat{\Sigma}, \lambda_0, k, R}\paren{X_{\le n_1}}$ will satisfy $v_i = \frac{1}{n_1}, \forall i \in \brk{n_1}$, and we will have $\widebar{\mu} = \widehat{\mu}$, and $\Score_2 = 0$.
\end{lemma}

\begin{proof}
The assumptions of Lemmas~\ref{lem:utility1} and~\ref{lem:utility2} are satisfied, thus a union bound implies that $E$ is realized, except with probability $\alpha$.
For the following, we condition over this event.
It suffices to prove that \LargestCore\ (Algorithm~\ref{alg:largest_core} Appendix~\ref{sec:bhs23_details}) will return the entire set $\brk{n_1}$ every time it is called by \StableMean\ (Algorithm~\ref{alg:stablemean} in Appendix~\ref{sec:bhs23_details}).
For this to hold, it suffices to have $\mnorm{\widehat{\Sigma}}{X_i - \mu}^2 \le \frac{\lambda_0}{4}, \forall i \in \brk{n_1}$.
If the previous is satisfied, the choice of the subset $R$ chosen for the execution of \StableMean\ in Line~\ref{ln:repr_set} will not matter as all the points are close to one another.
Indeed, for any $i \in \brk{n_1}$ and $j \in R$, the triangle inequality yields: 
\[
    \mnorm{\widehat{\Sigma}}{X_i - X_j}^2 \le \paren{\mnorm{\widehat{\Sigma}}{X_i - \mu} + \mnorm{\widehat{\Sigma}}{X_j - \mu}}^2 \le \lambda_0,
\]
and therefore $N_i = R$ for all $i \in [n_1]$, and thus \LargestCore\ will always return $\brk{n_1}$, and \StableMean\ will assign uniform weights to every $X_i$.
Thus, the rest of the proof is devoted to showing that $E$ implies $\mnorm{\widehat{\Sigma}}{X_i - \mu}^2 \le \frac{\lambda_0}{4}, \forall i \in \brk{n_1}$.
By the definition of the spectral norm, we have:
\[
    \mnorm{\widehat{\Sigma}}{X_i - \mu}^2 = \llnorm{\paren{\widehat{\Sigma}^{- \frac{1}{2}} \Sigma^{\frac{1}{2}}} \brk{\Sigma^{- \frac{1}{2}} \paren{X_i - \mu}}}^2 \le \llnorm{\widehat{\Sigma}^{- \frac{1}{2}} \Sigma^{\frac{1}{2}}}^2 \llnorm{\Sigma^{- \frac{1}{2}} \paren{X_i - \mu}}^2 = \llnorm{\widehat{\Sigma}^{- \frac{1}{2}} \Sigma^{\frac{1}{2}}}^2 \mnorm{\Sigma}{X_i - \mu}^2.
\]
Since we have assumed that $E_1$ is realized, we immediately get that the second term is upper-bounded by $\frac{\lambda_0}{4}$.
For the first term, since $E_2 \cap E_3$ is also assumed to occur, Lemma~\ref{lem:stable_cov_uti} implies that $\widehat{\Sigma} = \widebar{\Sigma}$.
This allows us to work as in Lemma~\ref{lem:stable_cov_uti} (namely (\ref{eq:matrix_prod_spec_ub})), and get $\llnorm{\widehat{\Sigma}^{- \frac{1}{2}} \Sigma^{\frac{1}{2}}}^2 \le 3$.
Then, the desired results follow directly.
\end{proof}

We are now ready to complete the proof of the utility guarantee.
Our proof will proceed by writing the output $Z$ of Algorithm~\ref{alg:sampling} as a mixture of two components, with one corresponding to when $E$ is realized, and the other corresponding to when $E^c$ is realized.
Then, the desired result will follow directly from Fact~\ref{fact:joint_convexity}.

\begin{theorem}
\label{thm:uti_guarantee}
Let $X \coloneqq \paren{X_1, \dots, X_n} \sim \cN\paren{\mu, \Sigma}^{\otimes n}$, where $n \coloneqq n_1 + 2 n_2$ and $n_2 \geq \cO\paren{d + \log\paren{\frac{1}{\alpha}}}$.
Also, let $\lambda_0 \geq 4 d + 8 \sqrt{d \ln\paren{\frac{3 n}{\alpha}}} + 8 \ln\paren{\frac{3 n}{\alpha}}$.
Then, we have $\TV\paren{Z, \cN\paren{\mu, \Sigma}} \le \alpha$.
\end{theorem}

\begin{proof}
First, we note that Lemmas~\ref{lem:stable_cov_uti} and~\ref{lem:stable_mean_uti} imply:
\[
    \paren{Z \middle| E} = \paren{\frac{1}{n_1} \sum\limits_{i \in \brk{n_1}} X_i + \sqrt{\paren{1 - \frac{1}{n_1}} n_2} \frac{1}{\sqrt{n_2}} \sum\limits_{i \in \brk{n_2}} z_i Y_i \middle| E} \text{ and } \pr{}{E^c} \le \alpha.
\]
Additionally, observe that Lemma~\ref{lem:stability_gaussian} implies that:
\[
    \frac{1}{n_1} \sum\limits_{i \in \brk{n_1}} X_i + \sqrt{\paren{1 - \frac{1}{n_1}} n_2} \frac{1}{\sqrt{n_2}} \sum\limits_{i \in \brk{n_2}} z_i Y_i \overset{d}{=} \cN\paren{\mu, \Sigma}.
\]
Combining the previous two remarks and applying Fact~\ref{fact:joint_convexity}, we get:
\begin{align*}
    \TV\paren{Z, \cN\paren{\mu, \Sigma}}
    &= \TV\paren{Z, \frac{1}{n_1} \sum\limits_{i \in \brk{n_1}} X_i + \sqrt{\paren{1 - \frac{1}{n_1}} n_2} \frac{1}{\sqrt{n_2}} \sum\limits_{i \in \brk{n_2}} z_i Y_i} \\
    &= \TV\paren{
        \mathds{1}\brc{E} Z + \mathds{1}\brc{E^c} Z, 
        \begin{array}{c}
            \mathds{1}\brc{E} \frac{1}{n_1} \sum\limits_{i \in \brk{n_1}} X_i + \sqrt{\paren{1 - \frac{1}{n_1}} n_2} \frac{1}{\sqrt{n_2}} \sum\limits_{i \in \brk{n_2}} z_i Y_i \\
            + \mathds{1}\brc{E^c} \frac{1}{n_1} \sum\limits_{i \in \brk{n_1}} X_i + \sqrt{\paren{1 - \frac{1}{n_1}} n_2} \frac{1}{\sqrt{n_2}} \sum\limits_{i \in \brk{n_2}} z_i Y_i
        \end{array}
    } \\
    &\le \pr{}{E} \cancelto{0}{\TV\paren{\paren{Z \middle| E}, \paren{\frac{1}{n_1} \sum\limits_{i \in \brk{n_1}} X_i + \sqrt{\paren{1 - \frac{1}{n_1}} n_2} \frac{1}{\sqrt{n_2}} \sum\limits_{i \in \brk{n_2}} z_i Y_i \middle| E}}} \\
    & \\
    &\quad\ + \pr{}{E^c} \TV\paren{\paren{Z \middle| E^c}, \paren{\frac{1}{n_1} \sum\limits_{i \in \brk{n_1}} X_i + \sqrt{\paren{1 - \frac{1}{n_1}} n_2} \frac{1}{\sqrt{n_2}} \sum\limits_{i \in \brk{n_2}} z_i Y_i \middle| E^c}} \\
    &\le \pr{}{E^c} \\
    &\le \alpha.
\end{align*}
\end{proof}

\subsection{Establishing the Privacy Guarantee}
\label{subsec:priv_proof}

Let $Z$ and $Z'$ be the outputs of Algorithm~\ref{alg:sampling} on adjacent input datasets $X$ and $X'$, respectively.
In this section, we will establish that, when $n = \cO\paren{\frac{\lambda_0 \log\paren{\frac{1}{\delta}}}{\eps}}$, we get $D_{e^{\eps}}\paren{Z \middle\| Z'} \le \delta$ (essentially relying on Remark~\ref{rem:dp_hs} to establish privacy).

We note that the structure of Algorithm~\ref{alg:sampling} falls into the framework of Definition~\ref{def:dp_under_condition} and Lemma~\ref{lem:composition_halting} due to the use of a PTR argument.
Given that the PTR step is the same as in~\cite{BrownHS23}, and its analysis is essentially covered by Lemmas~\ref{lem:bhs_ptr} and~\ref{lem:stable_conditions_ptr}, we will mainly focus on establishing the privacy guarantee of the subsequent steps under the condition that the PTR step has succeeded.
Furthermore, we note that Lemma~\ref{lem:stable_mean} relies on the set $R$ that is sampled in Line~\ref{ln:repr_set} being degree-representative (Definition~\ref{def:deg_rep}).
Thus, in addition to the PTR step succeeding, we also need to condition on $R$ being degree-representative.

Taking into account the above discussion, let us fix an input dataset $X$, and a representation set $R$.
For brevity, we write $\Psi\paren{X, R} = 1$ when both the outcome of the PTR step is \Pass\ and $R$ is degree-representative, and $\Psi\paren{X, R} = 0$ if either of the previous fails.\footnote{The assumption about $R$ being degree-representative will only matter in Section~\ref{subsubsec:priv_anal2} really, since that is the only section that relies on Lemma~\ref{lem:stable_mean}.}
Thus, we want to study the quantity $D_{e^{\eps}}\paren{\paren{Z \middle| \Psi\paren{X, R} = 1} \middle\| \paren{Z' \middle| \Psi\paren{X', R} = 1}}$.
This can be written as:
\[
    D_{e^{\eps}}\paren{\widehat{\mu} + \sqrt{\paren{1 - \frac{1}{n_1}} n_2} W z \middle\| \widehat{\mu}' + \sqrt{\paren{1 - \frac{1}{n_1}} n_2} W' z}.
\]
Observe that, in the above, we assumed without loss of generality that the random bits of Algorithm~\ref{alg:sampling} related to the sampling of $R$ and $z$ are coupled in the two executions of the algorithm under input $X$ and $X'$.
This allows us to take $R' = R$ and $z' = z$.

Observe that, by Lemma~\ref{lem:hs_div_triangle_ineq}, we get:
\begin{align}
    &\quad\ D_{e^{\eps}}\paren{\widehat{\mu} + \sqrt{\paren{1 - \frac{1}{n_1}} n_2} W z \middle\| \widehat{\mu}' + \sqrt{\paren{1 - \frac{1}{n_1}} n_2} W' z} \nonumber \\
    &\le D_{e^{\frac{\eps}{2}}}\paren{\widehat{\mu} + \sqrt{\paren{1 - \frac{1}{n_1}} n_2} W z \middle\| \widehat{\mu}' + \sqrt{\paren{1 - \frac{1}{n_1}} n_2} W z} \nonumber \\
    &\quad\ + e^{\frac{\eps}{2}} D_{e^{\frac{\eps}{2}}}\paren{\widehat{\mu}' + \sqrt{\paren{1 - \frac{1}{n_1}} n_2} W z \middle\| \widehat{\mu}' + \sqrt{\paren{1 - \frac{1}{n_1}} n_2} W' z} \nonumber \\
    &\le 2
    \paren{
        \begin{array}{c}
            D_{e^{\frac{\eps}{2}}}\paren{\widehat{\mu} + \sqrt{\paren{1 - \frac{1}{n_1}} n_2} W z \middle\| \widehat{\mu}' + \sqrt{\paren{1 - \frac{1}{n_1}} n_2} W z} \\
            + D_{e^{\frac{\eps}{2}}}\paren{\widehat{\mu}' + \sqrt{\paren{1 - \frac{1}{n_1}} n_2} W z \middle\| \widehat{\mu}' + \sqrt{\paren{1 - \frac{1}{n_1}} n_2} W' z}
        \end{array}
    }, \label{eq:weak_triangle_ineq}
\end{align}
where the last inequality used the assumption that $\eps \le 1$.

We stress that the application of Lemma~\ref{lem:hs_div_triangle_ineq} represents a clear difference from the privacy analysis of~\cite{GhaziHKM23}.
The reason for that has to do with the fact that $\widehat{\mu}$ and $\widehat{\mu}'$ are weighted sums whose weights depend on the output of \StableCov.
Thus, although the subset of the dataset that is fed to \StableMean\ is disjoint from the one that is fed to \StableCov, changing one of the points that are given to \StableCov\ will affect the execution of \StableMean\ due to $\widehat{\Sigma}$ and $\widehat{\Sigma}'$ being different.
Conversely, $\widehat{\mu}$ and $\widehat{\mu}'$ in~\cite{GhaziHKM23} are just the truncated empirical means.

The analysis will now be done in four steps, with a separate section devoted to each:
\begin{enumerate}
    \item as a warm-up, we will analyze the second (and more interesting term) of (\ref{eq:weak_triangle_ineq}) in the univariate setting (Section~\ref{subsubsec:priv_anal_uni1}),
    \item we will generalize the previous analysis to the multivariate setting (Section~\ref{subsubsec:priv_anal_multi1}),
    \item we will analyze the first term of (\ref{eq:weak_triangle_ineq}) in the multivariate setting (Section~\ref{subsubsec:priv_anal2}),
    \item we will tie everything together to complete the proof of the privacy guarantee (Section~\ref{subsubsec:priv_guarantee}).
\end{enumerate}

\subsubsection{Warm-Up: Analyzing the Second Term of (\ref{eq:weak_triangle_ineq}) for Univariate Data}
\label{subsubsec:priv_anal_uni1}

We start by analyzing the second term of (\ref{eq:weak_triangle_ineq}) in the univariate setting.
The proof makes crucial use of Lemma~\ref{lem:stable_cov} to argue that $\frac{\llnorm{W'}}{\llnorm{W}}$ will be close to $1$, which in turn implies that the Hockey-Stick divergence must be small.
We will give the proof in the three steps, with one statement devoted to each step of the proof:
\begin{enumerate}
    \item we will first show that the Hockey-Stick divergence term that we are considering can be equivalently expressed as the divergence between $z_1$ and $T \coloneqq \frac{\llnorm{W'}}{\llnorm{W}} z_1$ (Lemma~\ref{lem:priv_anal_uni1_dpi}),
    \item subsequently, we will identify a sufficient condition for points $t \in \supp\paren{z_1}$ that implies that the log-density ratio $\ln\paren{\frac{f_{z_1}\paren{t}}{f_T\paren{t}}}$ is at most $\frac{\eps}{2}$ (Lemma~\ref{lem:priv_anal_uni1_suff_cond}),
    \item finally, we will show that our sample complexity suffices for the second term of (\ref{eq:weak_triangle_ineq}) to be at most $\delta$ (Proposition~\ref{prop:priv_anal_uni1}).
\end{enumerate}
We start with the first of the three lemmas.
The proof combines the \emph{Data-Processing Inequality} (DPI - Fact~\ref{fact:dpi}) with the guarantees implied by Lemma~\ref{lem:stable_cov}.

\begin{lemma}
\label{lem:priv_anal_uni1_dpi}
Let $\eps \in \brk{0, 1}$ and $\delta \in \brk{0, \frac{\eps}{10}}, \lambda_0 \geq 1$.
Also, assume that $n_2 \geq \useconstant{C2} \frac{\lambda_0 \log\paren{\frac{1}{\delta}}}{\eps}$ for some appropriately large absolute constant $\useconstant{C2} \geq 1$.
Finally, for $n \coloneqq n_1 + 2 n_2$, let $X, X' \in \bR^n$ be adjacent datasets and $R \subseteq \brk{n_1}$ be a representation set such that $\Psi\paren{X, R} = \Psi\paren{X', R} = 1$.
Then, for $z \sim \cU\paren{\bS^{n_2 - 1}}$, we have:
\[
    D_{e^{\frac{\eps}{2}}}\paren{\widehat{\mu}' + \sqrt{\paren{1 - \frac{1}{n_1}} n_2} W z \middle\| \widehat{\mu}' + \sqrt{\paren{1 - \frac{1}{n_1}} n_2} W' z} = D_{e^{\frac{\eps}{2}}}\paren{z_1 \middle\| T},
\]
where $T \coloneqq \frac{\llnorm{W'}}{\llnorm{W}} z_1$.
\end{lemma}

\begin{proof}
First, by Fact~\ref{fact:dpi} (which is satisfied as an equality in this case), we get:
\begin{equation}
    D_{e^{\frac{\eps}{2}}}\paren{\widehat{\mu}' + \sqrt{\paren{1 - \frac{1}{n_1}} n_2} W z \middle\| \widehat{\mu}' + \sqrt{\paren{1 - \frac{1}{n_1}} n_2} W' z} = D_{e^{\frac{\eps}{2}}}\paren{W z \middle\| W' z}. \label{eq:uni_dpi_app1}
\end{equation}
We assume that $X$ and $X'$ differ at one of the points that are fed to \StableCov, since otherwise (\ref{eq:uni_dpi_app1}) is trivially $0$.
We note that, since we are dealing with univariate data, the shape of the matrix $W$ will be $1 \times n_2$, i.e., it is a row vector whose elements are of the form $\sqrt{w_i} Y_i$ for $Y_i \coloneqq \frac{1}{\sqrt{2}} \paren{X_{n_1 + i} - X_{n_1 + n_2 + i}}, \forall i \in \brk{n_2}$.
By the second part of Fact~\ref{fact:matrix_fact} and Remark~\ref{rem:svd_comment}, we get that $W$ and $W'$ can be written in the form $W = \llnorm{W} e_1^{\top} V^{\top}$ and $W' = \llnorm{W'} e_1^{\top} V'^{\top}$, where $V, V' \in \bR^{n_2 \times n_2}$ are rotation matrices.
However, Fact~\ref{fact:rotational_invariance} implies that $V^{\top} z \overset{d}{=} z$, yielding:
\begin{equation}
    D_{e^{\frac{\eps}{2}}}\paren{W z \middle\| W' z} = D_{e^{\frac{\eps}{2}}}\paren{\llnorm{W} e_1^{\top} z \middle\| \llnorm{W'} e_1^{\top} z} = D_{e^{\frac{\eps}{2}}}\paren{\llnorm{W} z_1 \middle\| \llnorm{W'} z_1}. \label{eq:uni_hs_div_simp1}
\end{equation}
We note now that $\llnorm{W} = \widehat{\sigma}, \llnorm{W'} = \widehat{\sigma}'$.
Given that, in Algorithm~\ref{alg:sampling}, $k$ is set to be equal to $\cl{\frac{6 \log\paren{\frac{6}{\delta}}}{\eps}} + 4$, our bound on $n_2$ implies that the condition $n_2 \geq 16 e^2 \lambda_0 k$ from Lemma~\ref{lem:stable_cov} is satisfied.
Since we assumed that $\Psi\paren{X, R} = \Psi\paren{X', R} = 1$, Lemma~\ref{lem:stable_cov} implies that $\llnorm{W}, \llnorm{W'} > 0$.
Thus, applying Fact~\ref{fact:dpi} (again as an equality) yields the desired result.
\end{proof}

We note now that, by Definition~\ref{def:hs_div}, we have the upper bound:
\begin{equation}
    D_{e^{\frac{\eps}{2}}}\paren{z_1 \middle\| T} \le \pr{t \sim z_1}{f_{z_1}\paren{t} > e^{\frac{\eps}{2}} f_T\paren{t}}, \label{eq:uni_hs_div_ub1}
\end{equation}
where by $f_{z_1}$ and $f_T$ we denote the densities of $z_1$ and $T$.
Thus, to reason about the above, it suffices to identify a set of points $t$ in the support of $z_1$ can lead to the condition $\ln\paren{\frac{f_{z_1}\paren{t}}{f_T\paren{t}}} \le \frac{\eps}{2}$ being satisfied for the log-density ratio.
Then, the probability of the event $f_{z_1}\paren{t} > e^{\frac{\eps}{2}} f_T\paren{t}$ will be upper-bounded by the probability of the aforementioned condition failing.
We identify the desired condition in the following lemma.

\begin{lemma}
\label{lem:priv_anal_uni1_suff_cond}
In the setting of Lemma~\ref{lem:priv_anal_uni1_dpi}, we have that $\abs{t} \le \sqrt{\frac{2}{3} \cdot \frac{\eps}{\eps + 16 e^2 \lambda_0}} \implies \ln\paren{\frac{f_{z_1}\paren{t}}{f_T\paren{t}}} \le \frac{\eps}{2}$.
\end{lemma}

\begin{proof}
We start by observing that Fact~\ref{fact:transformations} implies that:
\[
    f_T\paren{t} = \frac{\llnorm{W}}{\llnorm{W'}} f_{z_1}\paren{\frac{\llnorm{W}}{\llnorm{W'}} t}, \forall t \in \brk{\pm \frac{\llnorm{W'}}{\llnorm{W}}}.
\]
This should be contrasted with the support of $z_1$, which is the interval $\brk{\pm 1}$.
By Lemma~\ref{lem:stable_cov}, we get that $\sqrt{1 - \gamma} \le \frac{\llnorm{W'}}{\llnorm{W}} \le \frac{1}{\sqrt{1 - \gamma}}$ for $\gamma \coloneqq \frac{8 e^2 \lambda_0}{n_2}$.
Our choice of $n_2$ implies that $\gamma \le \frac{\eps}{4} \le \frac{1}{4} \implies \sqrt{1 - \gamma} \geq \sqrt{\frac{3}{4}}$.
The range that we consider for $t$ is $\abs{t} \le \sqrt{\frac{2}{3} \cdot \frac{\eps}{\eps + 16 e^2 \lambda_0}} < \sqrt{\frac{2}{3}} < \sqrt{\frac{3}{4}}$.
This implies that the values of $t$ we consider lie in the support of $T$, so the log-density ratio is finite.
With this knowledge at hand, we work towards establishing the desired upper bound on the log-density ratio.
To that end, Fact~\ref{fact:unit_sphere_proj_dist} and the standard inequality $\ln\paren{x} \le x - 1, \forall x > 0$ yield:
\begin{align}
    \ln\paren{\frac{f_{z_1}\paren{t}}{f_T\paren{t}}}
    &= \ln\paren{\frac{\llnorm{W'}}{\llnorm{W}}} + \frac{n_2 - 3}{2} \ln\paren{\frac{1 - t^2}{1 - \paren{\frac{\llnorm{W}}{\llnorm{W'}} t}^2}} \nonumber \\
    &< \frac{1}{2} \ln\paren{\frac{1}{1 - \gamma}} + \frac{n_2}{2} \ln\paren{\frac{1 - t^2}{1 - \paren{\frac{1}{\sqrt{1 - \gamma}} t}^2}} \nonumber \\
    &\le \frac{1}{2} \paren{\frac{1}{1 - \gamma} - 1} + \frac{n_2}{2} \brk{\frac{1 - t^2}{1 - \paren{\frac{1}{\sqrt{1 - \gamma}} t}^2} - 1} \nonumber \\
    &= \frac{1}{2} \paren{\frac{1}{1 - \gamma} - 1} + \frac{n_2}{2} \cdot \frac{\frac{1}{1 - \gamma} - 1}{1 - \frac{t^2}{1 - \gamma}} t^2 \nonumber \\
    &= \frac{1}{2} \paren{\frac{1}{1 - \gamma} - 1} + \frac{n_2}{2} \cdot \frac{\frac{\gamma}{1 - \gamma}}{1 - \frac{t^2}{1 - \gamma}} t^2 \nonumber \\
    &= \frac{1}{2} \paren{\frac{1}{1 - \gamma} - 1} + \frac{1}{2} \cdot \frac{n_2 \gamma}{\paren{1 - \gamma} - t^2} t^2 \nonumber \\
    &= \frac{1}{2} \paren{\frac{1}{1 - \gamma} - 1} + \frac{4 e^2 \lambda_0}{\frac{1 - \gamma}{t^2} - 1}. \label{eq:uni_log_density_ratio_cov}
\end{align}
To complete the proof, we will now show that each of the two terms of (\ref{eq:uni_log_density_ratio_cov}) is upper-bounded by $\frac{\eps}{4}$.
For the first term, we have:
\[
    \frac{1}{2} \paren{\frac{1}{1 - \gamma} - 1} \le \frac{\eps}{4} \iff \gamma \le \frac{\eps}{2 + \eps},
\]
which holds since we have $\gamma \le \frac{\eps}{4}$ and $\eps \in \brk{0, 1}$.

For the second term, we note that it is an increasing function of $t$, so we show that:
\[
    \frac{4 e^2 \lambda_0}{\frac{1 - \gamma}{\frac{2}{3} \cdot \frac{\eps}{\eps + 16 e^2 \lambda_0}} - 1} \le \frac{\eps}{4} \iff \frac{16 e^2 \lambda_0}{\eps} \le \frac{3}{2} \cdot \frac{\eps + 16 e^2 \lambda_0}{\eps} \paren{1 - \gamma} - 1 \iff \gamma \le \frac{1}{3},
\]
which again is implied by our assumptions.
\end{proof}

We can show now the main result of this section.

\begin{proposition}
\label{prop:priv_anal_uni1}
Let $\eps \in \brk{0, 1}$ and $\delta \in \brk{0, \frac{\eps}{10}}, \lambda_0 \geq 1$.
Also, assume that $n_2 \geq \useconstant{C2} \frac{\lambda_0 \log\paren{\frac{1}{\delta}}}{\eps}$ for some appropriately large absolute constant $\useconstant{C2} \geq 1$.
Finally, for $n \coloneqq n_1 + 2 n_2$, let $X, X' \in \bR^n$ be adjacent datasets and $R \subseteq \brk{n_1}$ be a representation set such that $\Psi\paren{X, R} = \Psi\paren{X', R} = 1$.
Then, for $z \sim \cU\paren{\bS^{n_2 - 1}}$, we have:
\[
    D_{e^{\frac{\eps}{2}}}\paren{\widehat{\mu}' + \sqrt{\paren{1 - \frac{1}{n_1}} n_2} W z \middle\| \widehat{\mu}' + \sqrt{\paren{1 - \frac{1}{n_1}} n_2} W' z} \le \delta.
\]
\end{proposition}

\begin{proof}
By Lemma~\ref{lem:priv_anal_uni1_dpi}, (\ref{eq:uni_hs_div_simp1}), and Lemma~\ref{lem:priv_anal_uni1_suff_cond}, we get that:
\[
    D_{e^{\frac{\eps}{2}}}\paren{\widehat{\mu}' + \sqrt{\paren{1 - \frac{1}{n_1}} n_2} W z \middle\| \widehat{\mu}' + \sqrt{\paren{1 - \frac{1}{n_1}} n_2} W' z} \le \pr{t \sim z_1}{\abs{t} > \sqrt{\frac{2}{3} \cdot \frac{\eps}{\eps + 16 e^2 \lambda_0}}}.
\]
By Fact~\ref{fact:unit_sphere_proj_dist}, we have that $z_1^2 \sim \mathrm{\Beta}\paren{\frac{1}{2}, \frac{n_2 - 1}{2}}$.
By our bound on $n_2$, we get that $n_2 > 2 \iff \frac{n_2 - 1}{2} > \frac{1}{2}$ and $n_2 > \frac{3}{2} \cdot \frac{\eps + 16 e^2 \lambda_0}{\eps}$, so Fact~\ref{fact:beta_concentration}, yields:
\begin{align}
    &\quad\ \pr{t \sim z_1}{t^2 > \frac{2}{3} \cdot \frac{\eps}{\eps + 16 e^2 \lambda_0}} \nonumber \\
    &\le 2 \exp\paren{- \useconstant{cbeta} \min\brc{\frac{\paren{n_2 - 1}^2}{2} \paren{\frac{2}{3} \cdot \frac{\eps}{\eps + 16 e^2 \lambda_0} - \frac{1}{n_2}}^2, \frac{n_2 - 1}{2} \paren{\frac{2}{3} \cdot \frac{\eps}{\eps + 16 e^2 \lambda_0} - \frac{1}{n_2}}}}. \label{eq:beta_concentration1}
\end{align}
Again by our bound on $n_2$, we have that $n_2 > 1 + \frac{9}{2} \cdot \frac{\eps + 16 e^2 \lambda_0}{\eps}$, which implies:
\begin{align*}
    \frac{\paren{n_2 - 1}^2}{2} \paren{\frac{2}{3} \cdot \frac{\eps}{\eps + 16 e^2 \lambda_0} - \frac{1}{n_2}}^2
    &\geq \frac{n_2 - 1}{2} \paren{\frac{2}{3} \cdot \frac{\eps}{\eps + 16 e^2 \lambda_0} - \frac{1}{n_2}} \\
    &\geq \frac{n_2 - 1}{6} \cdot \frac{\eps}{\eps + 16 e^2 \lambda_0} \\
    &\geq \frac{n_2}{12} \cdot \frac{\eps}{\eps + 16 e^2 \lambda_0} \\
    &\geq \frac{\useconstant{C2}}{12} \cdot \frac{\lambda_0 \log\paren{\frac{1}{\delta}}}{\cancel{\eps}} \cdot \frac{\cancel{\eps}}{\eps + 16 e^2 \lambda_0} \\
    &= \frac{\useconstant{C2}}{12} \cdot \frac{\log\paren{\frac{1}{\delta}}}{16 e^2 + \frac{\eps}{\lambda_0}} \\
    &= \frac{\useconstant{C2}}{12} \cdot \frac{\log\paren{\frac{1}{\delta}}}{16 e^2 + 1},
\end{align*}
where we used that $\lambda_0 \geq 1$ and $\eps \le 1$.

For $\useconstant{C2}$ sufficiently large, using the above to upper-bound (\ref{eq:beta_concentration1}) yields the desired result.
\end{proof}

\subsubsection{Analyzing the Second Term of (\ref{eq:weak_triangle_ineq}) for Multivariate Data}
\label{subsubsec:priv_anal_multi1}

We now proceed to show the multivariate analogue of Proposition~\ref{prop:priv_anal_uni1}.
As in the univariate setting, we will break down the proof into a number of lemmas to make the presentation more modular and accessible.
Specifically, the main steps are:
\begin{enumerate}
    \item we first show that the term of interest can be equivalently expressed as the divergence between $z_{\le d}$ and $T \coloneqq \Lambda^{- \frac{1}{2}} U^{\top} U' \Lambda'^{\frac{1}{2}} z_{\le d}$, where $\Lambda, \Lambda' \succ 0$ are diagonal matrices, and $U, U'$ are rotation matrices such that we have $\widehat{\Sigma} \coloneqq W W^{\top} = U \Lambda U^{\top}$ and $\widehat{\Sigma}' \coloneqq W' W'^{\top} = U \Lambda' U^{\top}$ (Lemma~\ref{lem:priv_anal_multi1_dpi}),
    \item then, we identify the density of $T$ (Lemma~\ref{lem:t_density}),\footnote{This requires a non-trivial amount of work compared to the univariate setting, which justifies this step's presentation in the form of an autonomous lemma.}
    \item subsequently, we will identify sufficient conditions for points $t \in \supp\paren{z_1}$ that imply that the log-density ratio $\ln\paren{\frac{f_{z_1}\paren{t}}{f_T\paren{t}}}$ is at most $\frac{\eps}{2}$ (Lemma~\ref{lem:priv_anal_multi1_suff_cond}),
    \item we will consider the matrix $A \coloneqq Q^{\top} \paren{\widehat{\Sigma}^{\frac{1}{2}} \widehat{\Sigma}'^{- 1} \widehat{\Sigma}^{\frac{1}{2}} - \id_{d \times d}} Q - \frac{\eps}{4 n_2} \id_{n_2 \times n_2}$, where $Q \in \bR^{d \times n_2}$ denotes the projection matrix that keeps the first $d$ components of a vector in $\bR^{n_2}$, and bound its trace, as well as its spectral and Frobenius norms (Lemma~\ref{lem:matrix_bounds} and Corollary~\ref{cor:matrix_bounds_supp}).
    \item finally, we will show that our sample complexity suffices for the second term of (\ref{eq:weak_triangle_ineq}) to be at most $\delta$ (Proposition~\ref{prop:priv_anal_multi1}).
\end{enumerate}
Observe that the fourth step in the above outline is completely missing from the univariate setting.
The reason this step is required here is because, as we will see over the course of the proof, Fact~\ref{fact:beta_concentration} will not be strong enough to allow us to complete the proof with optimal sample complexity.
This is in contrast to how things worked out in the univariate setting, and the proof of Proposition~\ref{prop:priv_anal_uni1} specifically.
Instead, our proof will involve a reduction to an argument that uses Gaussian random vectors, which will allow us to appeal to Fact~\ref{fact:hanson_wright}.

We note that the argument given in this section is not our original version of the argument.
Specifically, our initial proof did not rely on the aforementioned reduction, instead using Fact~\ref{fact:beta_concentration}, albeit with a slightly sub-optimal sample complexity by a $\mathrm{polylog}\paren{\lambda_0}$-factor.
We consider that proof to be of significant technical interest, which is why we have included it in Appendix~\ref{sec:priv_proof_alt}.
The appendix also includes a detailed discussion about why the univariate approach cannot be directly applied in the multivariate setting. 

Having concluded the above discussion, we start with the technical content.
We give here the first lemma from the sequence described earlier.
Its proof is analogous that of Lemma~\ref{lem:priv_anal_uni1_dpi}, but requires some extra care due to the fact that we are dealing with multivariate data.

\begin{lemma}
\label{lem:priv_anal_multi1_dpi}
Let $\eps \in \brk{0, 1}$ and $\delta \in \brk{0, \frac{\eps}{10}}, \lambda_0 \geq d$.
Let us assume that $n_2 \geq \useconstant{C2} \frac{\lambda_0 \log\paren{\frac{1}{\delta}}}{\eps}$, for some appropriately large absolute constant $\useconstant{C2} \geq 1$.
Finally, for $n \coloneqq n_1 + 2 n_2$, let $X, X' \in \bR^{n \times d}$ be adjacent datasets and $R \subseteq \brk{n_1}$ be a representation set such that $\Psi\paren{X, R} = \Psi\paren{X', R} = 1$.
Then, for $z \sim \cU\paren{\bS^{n_2 - 1}}$, we have:
\[
    D_{e^{\frac{\eps}{2}}}\paren{\widehat{\mu}' + \sqrt{\paren{1 - \frac{1}{n_1}} n_2} W z \middle\| \widehat{\mu}' + \sqrt{\paren{1 - \frac{1}{n_1}} n_2} W' z} = D_{e^{\frac{\eps}{2}}}\paren{z_{\le d} \middle\| T},
\]
where $T \coloneqq \Lambda^{- \frac{1}{2}} U^{\top} U' \Lambda'^{\frac{1}{2}} z_{\le d}$, while $\Lambda, \Lambda' \succ 0$ are diagonal matrices, and $U, U'$ are rotation matrices such that we have $\widehat{\Sigma} \coloneqq W W^{\top} = U \Lambda U^{\top}$ and $\widehat{\Sigma}' \coloneqq W' W'^{\top} = U \Lambda' U^{\top}$.
\end{lemma}

\begin{proof}
As in Lemma~\ref{lem:priv_anal_uni1_dpi}, the equality case of Fact~\ref{fact:dpi} yields:
\begin{equation}
    D_{e^{\frac{\eps}{2}}}\paren{\widehat{\mu}' + \sqrt{\paren{1 - \frac{1}{n_1}} n_2} W z \middle\| \widehat{\mu}' + \sqrt{\paren{1 - \frac{1}{n_1}} n_2} W' z} = D_{e^{\frac{\eps}{2}}}\paren{W z \middle\| W' z}. \label{eq:multi_dpi_app1}
\end{equation}
We assume that $X$ and $X'$ differ at one of the points that are fed to \StableCov, since otherwise (\ref{eq:multi_dpi_app1}) is trivially $0$.
For multivariate data, the shape of the matrix $W$ will be $d \times n_2$.
By the second part of Fact~\ref{fact:matrix_fact} and Remark~\ref{rem:svd_comment}, we get that $W$ and $W'$ can be written in the form $W = U D V^{\top}$ and $W' = U' D' V'^{\top}$, where $U, U' \in \bR^{d \times d}$ and $V, V' \in \bR^{n_2 \times n_2}$ are rotation matrices, while $D, D' \in \bR^{d \times n_2}$ are matrices of the form $D = \paren{\diag\brc{\sigma_1, \dots, \sigma_d} \middle| 0_{d \times \paren{n_2 - d}}}$ and $D' = \paren{\diag\brc{\sigma_1', \dots, \sigma_d'} \middle| 0_{d \times \paren{n_2 - d}}}$.
However, Fact~\ref{fact:rotational_invariance} implies that $V^{\top} z \overset{d}{=} z$, yielding:
\begin{equation}
    D_{e^{\frac{\eps}{2}}}\paren{W z \middle\| W' z} = D_{e^{\frac{\eps}{2}}}\paren{U D z \middle\| U' D' z} = D_{e^{\frac{\eps}{2}}}\paren{U \Lambda^{\frac{1}{2}} z_{\le d} \middle\| U' \Lambda'^{\frac{1}{2}} z_{\le d}}, \label{eq:multi_hs_div_simp1}
\end{equation}
where $\Lambda \coloneqq D D^{\top} = \diag\brc{\sigma_1^2, \dots, \sigma_d^2}$ and $\Lambda' \coloneqq D' D'^{\top} = \diag\brc{\sigma_1'^2, \dots, \sigma_d'^2}$.

We note now that:
\[
    \widehat{\Sigma} = W W^{\top} = \paren{U D V^{\top}} \paren{U D V^{\top}}^{\top} = U D \paren{V^{\top} V} D^{\top} U^{\top} = U D D^{\top} U^{\top} = U \Lambda U^{\top},
\]
and, analogously, $\widehat{\Sigma}' = U' \Lambda' U'^{\top}$.
Our choice of $n_2$ implies that the condition $n_2 \geq 16 e^2 \lambda_0 k$ of Lemma~\ref{lem:stable_cov} is satisfied, and, since $\Psi\paren{X, R} = \Psi\paren{X', R} = 1$, we have that $\Lambda, \Lambda' \succ 0 \implies \exists \Lambda^{- 1}, \Lambda'^{- 1}$.
Thus, applying Fact~\ref{fact:dpi} (again as an equality) yields the desired result.
\end{proof}

Observe now that the multivariate analogue of (\ref{eq:uni_hs_div_ub1}) implies that:
\begin{equation}
    D_{e^{\frac{\eps}{2}}}\paren{z_{\le d} \middle\| T} \le \pr{t \sim z_{\le d}}{f_{z_{\le d}}\paren{t} > e^{\frac{\eps}{2}} f_T\paren{t}}. \label{eq:multi_hs_div_ub1}
\end{equation}
Thus, we need to identify the density of $T$.
We do so in the following lemma.

\begin{lemma}
\label{lem:t_density}
In the setting of Lemma~\ref{lem:priv_anal_multi1_dpi}, it holds that:
\[
    f_T\paren{t} \ \propto \ \sqrt{\det\paren{\widehat{\Sigma}^{\frac{1}{2}} \widehat{\Sigma}'^{- 1} \widehat{\Sigma}^{\frac{1}{2}}}} \paren{1 - s^{\top} \widehat{\Sigma}^{\frac{1}{2}} \widehat{\Sigma}'^{- 1} \widehat{\Sigma}^{\frac{1}{2}} s}^{\frac{n_2 - d}{2} - 1} \mathds{1} \brc{s^{\top} \widehat{\Sigma}^{\frac{1}{2}} \widehat{\Sigma}'^{- 1} \widehat{\Sigma}^{\frac{1}{2}} s \le 1},
\]
where $s \coloneqq s\paren{t} \coloneqq U t$ with $\llnorm{s} = \llnorm{t}$.
\end{lemma}

\begin{proof}
Facts~\ref{fact:unit_sphere_proj_dist} and~\ref{fact:transformations} yield that:
\begin{align}
    f_T\paren{t}
    &= \det\paren{\Lambda'^{- \frac{1}{2}} U'^{\top} U \Lambda^{\frac{1}{2}}} f_{z_{\le d}}\paren{\Lambda'^{- \frac{1}{2}} U'^{\top} U \Lambda^{\frac{1}{2}} t} \nonumber \\
    &\ \propto \ \det\paren{\Lambda'^{- \frac{1}{2}} U'^{\top} U \Lambda^{\frac{1}{2}}} \paren{1 - \llnorm{\Lambda'^{- \frac{1}{2}} U'^{\top} U \Lambda^{\frac{1}{2}} t}^2}^{\frac{n_2 - d}{2} - 1}, \label{eq:density_complicated}
\end{align}
for all $t \in \bR^d$ satisfying $\llnorm{\Lambda'^{- \frac{1}{2}} U'^{\top} U \Lambda^{\frac{1}{2}} t}^2 \le 1$.

We work towards simplifying the above.
First, observe that:
\begin{equation}
    \llnorm{\Lambda'^{- \frac{1}{2}} U'^{\top} U \Lambda^{\frac{1}{2}} t}^2 = t^{\top} \Lambda^{\frac{1}{2}} U^{\top} U' \Lambda'^{- 1} U'^{\top} U \Lambda^{\frac{1}{2}} t = \paren{U t}^{\top} \widehat{\Sigma}^{\frac{1}{2}} \widehat{\Sigma}'^{- 1} \widehat{\Sigma}^{\frac{1}{2}} \paren{U t} = s^{\top} \widehat{\Sigma}^{\frac{1}{2}} \widehat{\Sigma}'^{- 1} \widehat{\Sigma}^{\frac{1}{2}} s, \label{eq:obs1}
\end{equation}
where we appealed to the guarantees of Lemma~\ref{lem:priv_anal_multi1_dpi} to write $\widehat{\Sigma} = U \Lambda U^{\top}, \widehat{\Sigma}' = U' \Lambda' U'^{\top}$ with $\Lambda, \Lambda' \succ 0$ (which ensures that the inverses and square roots of $\widehat{\Sigma}$ and $\widehat{\Sigma}'$ are well-defined).

Second, we have:
\begin{align}
    \det\paren{\Lambda'^{- \frac{1}{2}} U'^{\top} U \Lambda^{\frac{1}{2}}} \overset{\paren{a}}{=} \det\paren{\Lambda'^{- \frac{1}{2}}} \det\paren{\Lambda^{\frac{1}{2}}}
    \overset{\paren{b}}&{=} \sqrt{\det\paren{\Lambda'^{- 1}} \det\paren{\Lambda}} \nonumber \\
    \overset{\paren{c}}&{=} \sqrt{\det\paren{\Lambda^{\frac{1}{2}}} \det\paren{\Lambda'^{- 1}} \det\paren{\Lambda^{\frac{1}{2}}}} \nonumber \\
    \overset{\paren{d}}&{=} \sqrt{\det\paren{\widehat{\Sigma}^{\frac{1}{2}} \widehat{\Sigma}'^{- 1} \widehat{\Sigma}^{\frac{1}{2}}}}, \label{eq:obs2}
\end{align}
where $\paren{a}$ used that the determinant is multiplicative when considering products of square matrices and that rotation matrices have determinant $1$, $\paren{b}$ used that the determinant of the square root of a diagonal matrix is the square root of the determinant, $\paren{c}$ used again the multiplicative property, and $\paren{d}$ used the same properties as $\paren{a}$.

Substituting to (\ref{eq:density_complicated}) based on (\ref{eq:obs1}) and (\ref{eq:obs2}) yields the desired result about the density.
It remains to argue that $\llnorm{s} = \llnorm{t}$.
This follows directly from Fact~\ref{fact:rotational_invariance}.
\end{proof}

\begin{remark}
\label{rem:t_support}
Lemma~\ref{lem:t_density} establishes that, although $z_{\le d}$ is supported over the origin-centered unit ball, $T$ is supported over an origin-centered ellipsoid, where the orientation and the length of the axes are determined by the eigendecomposition of $\widehat{\Sigma}^{\frac{1}{2}} \widehat{\Sigma}'^{- 1} \widehat{\Sigma}^{\frac{1}{2}}$.
Lemma~\ref{lem:stable_cov} implies that the shape of this ellipsoid will not be too different from that of the unit ball, and this is something that will be leveraged significantly in Lemma~\ref{lem:matrix_bounds}.
\end{remark}

Having identified the density of $T$, and working as we did in Section~\ref{subsubsec:priv_anal_uni1}, we proceed to identify sufficient conditions for the log-density ratio $\ln\paren{\frac{f_{z_{\le d}}\paren{t}}{f_T\paren{t}}}$ to be upper-bounded by $\frac{\eps}{2}$.

\begin{lemma}
\label{lem:priv_anal_multi1_suff_cond}
In the setting of Lemmas~\ref{lem:priv_anal_multi1_dpi} and~\ref{lem:t_density}, we have that:
\[
    s^{\top} \widehat{\Sigma}^{\frac{1}{2}} \widehat{\Sigma}'^{- 1} \widehat{\Sigma}^{\frac{1}{2}} s \le \frac{1}{2} \text{ and } s^{\top} \paren{\widehat{\Sigma}^{\frac{1}{2}} \widehat{\Sigma}'^{- 1} \widehat{\Sigma}^{\frac{1}{2}} - \id} s \le \frac{\eps}{4 n_2} \implies \ln\paren{\frac{f_{z_{\le d}}\paren{t}}{f_T\paren{t}}} \le \frac{\eps}{2}.
\]
\end{lemma}

\begin{proof}
The first of our two conditions ensures that we are considering points in the support of $T$, and thus the log-density ratio will be bounded.
Under that assumption, we work with Fact~\ref{fact:unit_sphere_proj_dist} and Lemma~\ref{lem:t_density}, and obtain:
\begin{align}
    \ln\paren{\frac{f_{z_{\le d}}\paren{t}}{f_T\paren{t}}}
    &= \ln\paren{\frac{1}{\sqrt{\det\paren{\widehat{\Sigma}^{\frac{1}{2}} \widehat{\Sigma}'^{- 1} \widehat{\Sigma}^{\frac{1}{2}}}}}} + \frac{n_2 - d - 2}{2} \ln\paren{\frac{1 - \llnorm{t}^2}{1 - s^{\top} \widehat{\Sigma}^{\frac{1}{2}} \widehat{\Sigma}'^{- 1} \widehat{\Sigma}^{\frac{1}{2}} s}} \nonumber \\
    &= \frac{1}{2} \ln\paren{\det\paren{\widehat{\Sigma}^{- \frac{1}{2}} \widehat{\Sigma}' \widehat{\Sigma}^{- \frac{1}{2}}}} + \frac{n_2 - d - 2}{2} \ln\paren{\frac{1 - \llnorm{s}^2}{1 - s^{\top} \widehat{\Sigma}^{\frac{1}{2}} \widehat{\Sigma}'^{- 1} \widehat{\Sigma}^{\frac{1}{2}} s}}. \label{eq:multi_log_density_ratio_cov1}
\end{align}
We proceed to upper-bound each term of (\ref{eq:multi_log_density_ratio_cov1}) separately.
We will show that, as a consequence of our assumptions, both terms are upper-bounded by $\frac{\eps}{4}$.
For the first term, we have by Fact~\ref{fact:pd_am_gam}:
\begin{equation}
    \frac{1}{2} \ln\paren{\det\paren{\widehat{\Sigma}^{- \frac{1}{2}} \widehat{\Sigma}' \widehat{\Sigma}^{- \frac{1}{2}}}} \le \frac{d}{2} \ln\paren{\frac{\tr\paren{\widehat{\Sigma}^{- \frac{1}{2}} \widehat{\Sigma}' \widehat{\Sigma}^{- \frac{1}{2}}}}{d}}. \label{eq:am_gm_app}
\end{equation}
By triangle inequality and Lemma~\ref{lem:stable_cov}, we get:
\begin{align}
    &\qquad\quad \abs{\tr\paren{\widehat{\Sigma}^{- \frac{1}{2}} \widehat{\Sigma}' \widehat{\Sigma}^{- \frac{1}{2}}} - d} \le \mnorm{\tr}{\widehat{\Sigma}^{- \frac{1}{2}} \widehat{\Sigma}' \widehat{\Sigma}^{- \frac{1}{2}} - \id} \le \paren{1 + 2 \gamma} \gamma \nonumber \\
    &\implies \tr\paren{\widehat{\Sigma}^{- \frac{1}{2}} \widehat{\Sigma}' \widehat{\Sigma}^{- \frac{1}{2}}} \le d + \paren{1 + 2 \gamma} \gamma \le d + \frac{3}{2} \gamma, \label{eq:trace_triangle_ineq}
\end{align}
where the last inequality used the fact that $n_2$ is large enough for us to have $\gamma \le \frac{\eps}{4} \le \frac{1}{4}$.

Upper-bounding (\ref{eq:am_gm_app}) using (\ref{eq:trace_triangle_ineq}) and the inequality $\ln\paren{x} \le x - 1, \forall x > 0$, we get:
\begin{equation}
    \frac{1}{2} \ln\paren{\det\paren{\widehat{\Sigma}^{- \frac{1}{2}} \widehat{\Sigma}' \widehat{\Sigma}^{- \frac{1}{2}}}} \le \frac{d}{2} \ln\paren{1 + \frac{3 \gamma}{2 d}} \le \frac{3}{4} \gamma, \label{eq:first_term_ub}
\end{equation}
We want the upper bound of (\ref{eq:first_term_ub}) to be at most $\frac{\eps}{4}$, which is equivalent to $\gamma \le \frac{\eps}{3}$.
However, our assumptions on $n_2$ and $\gamma$ directly imply that, completing the calculation for the first term of (\ref{eq:multi_log_density_ratio_cov1}).

We now turn to the second term of (\ref{eq:multi_log_density_ratio_cov1}).
Again using the inequality $\ln\paren{x} \le x - 1, \forall x > 0$, we get:
\begin{align*}
    \frac{n_2 - d - 2}{2} \ln\paren{\frac{1 - \llnorm{s}^2}{1 - s^{\top} \widehat{\Sigma}^{\frac{1}{2}} \widehat{\Sigma}'^{- 1} \widehat{\Sigma}^{\frac{1}{2}} s}}
    &< \frac{n_2}{2} \paren{\frac{1 - \llnorm{s}^2}{1 - s^{\top} \widehat{\Sigma}^{\frac{1}{2}} \widehat{\Sigma}'^{- 1} \widehat{\Sigma}^{\frac{1}{2}} s} - 1} \\
    &= \frac{n_2}{2} \cdot \frac{s^{\top} \paren{\widehat{\Sigma}^{\frac{1}{2}} \widehat{\Sigma}'^{- 1} \widehat{\Sigma}^{\frac{1}{2}} - \id} s}{1 - s^{\top} \widehat{\Sigma}^{\frac{1}{2}} \widehat{\Sigma}'^{- 1} \widehat{\Sigma}^{\frac{1}{2}} s} \\
    &\le n_2 s^{\top} \paren{\widehat{\Sigma}^{\frac{1}{2}} \widehat{\Sigma}'^{- 1} \widehat{\Sigma}^{\frac{1}{2}} - \id} s \\
    &\le \frac{\eps}{4},
\end{align*}
where the last two inequalities follow from applications of the assumptions $s^{\top} \widehat{\Sigma}^{\frac{1}{2}} \widehat{\Sigma}'^{- 1} \widehat{\Sigma}^{\frac{1}{2}} s \le \frac{1}{2}$ and $s^{\top} \paren{\widehat{\Sigma}^{\frac{1}{2}} \widehat{\Sigma}'^{- 1} \widehat{\Sigma}^{\frac{1}{2}} - \id} s \le \frac{\eps}{4 n_2}$, respectively.
\end{proof}

The next steps focus on upper-bounding the probability of either of the two conditions of Lemma~\ref{lem:priv_anal_multi1_suff_cond} failing.
For the first of the two conditions, it is possible to derive the bound using Fact~\ref{fact:beta_concentration}, as we did in the univariate case.
However, the second condition cannot be handled that way, if we want to get a result with optimal sample complexity (see Appendix~\ref{sec:priv_proof_alt} for a detailed discuss as to why this is the case).
For that reason, our argument will have to go through Fact~\ref{fact:hanson_wright}.
To facilitate the implementation of the argument, we give the following auxiliary lemma.

\begin{lemma}
\label{lem:matrix_bounds}
Assume we are in the setting of Lemma~\ref{lem:priv_anal_multi1_dpi}.
Let $Q \in \bR^{d \times n_2}$ be the projection matrix which, when acting on a vector in $\bR^{n_2}$, keeps its first $d$ components.
Finally, let $A \coloneqq Q^{\top} \paren{\widehat{\Sigma}^{\frac{1}{2}} \widehat{\Sigma}'^{- 1} \widehat{\Sigma}^{\frac{1}{2}} - \id_{d \times d}} Q - \frac{\eps}{4 n_2} \id_{n_2 \times n_2}$, where we assume that $\widehat{\Sigma}^{\frac{1}{2}} \widehat{\Sigma}'^{- 1} \widehat{\Sigma}^{\frac{1}{2}} \succeq \id$.
Then, we have:
\begin{itemize}
    \item $\tr\paren{A} \le \frac{3}{2} \gamma - \frac{\eps}{4} < 0$,
    \item $\llnorm{A} \le \frac{\eps}{4 n_2} \paren{\frac{128 e^2 \lambda_0}{3} - 1}$,
    \item $\fnorm{A}^2 \le \frac{d \eps^2}{16 n_2^2} \brk{\paren{\frac{48 e^2 \lambda_0}{d \eps} - 1}^2 + \paren{\frac{n_2}{d} - 1}}$.
\end{itemize}
\end{lemma}

\begin{proof}
We start by noting that, by definition, it must be the case that $Q = \paren{\id_{d \times d} \middle| 0_{d \times \paren{n_2 - d}}}$.
This implies that $Q Q^{\top} = \id_{d \times d}$.
Using this observation, we proceed to bound each of the three quantities of interest separately.
Starting with the trace, we have:
\begin{align*}
    \tr\paren{A}
    &= \tr\paren{Q^{\top} \paren{\widehat{\Sigma}^{\frac{1}{2}} \widehat{\Sigma}'^{- 1} \widehat{\Sigma}^{\frac{1}{2}} - \id_{d \times d}} Q - \frac{\eps}{4 n_2} \id_{n_2 \times n_2}} \\
    &= \tr\paren{\paren{\widehat{\Sigma}^{\frac{1}{2}} \widehat{\Sigma}'^{- 1} \widehat{\Sigma}^{\frac{1}{2}} - \id_{d \times d}} Q Q^{\top}} - \frac{\eps}{4 n_2} \tr\paren{\id_{n_2 \times n_2}} \\
    &= \tr\paren{\widehat{\Sigma}^{\frac{1}{2}} \widehat{\Sigma}'^{- 1} \widehat{\Sigma}^{\frac{1}{2}} - \id_{d \times d}} - \frac{\eps}{4} \\
    &\le \frac{3}{2} \gamma - \frac{\eps}{4},
\end{align*}
where, along the way, we appealed to the linearity and cyclic properties of the trace, as well as (\ref{eq:trace_triangle_ineq}).
The fact that the above bound is $< 0$ follows from our assumption about $n_2$ (recall that $\gamma \coloneqq \frac{8 e^2 \lambda_0}{n_2}$).

Moving on to the spectral norm, its definition for symmetric matrices yields:
\begin{align}
    \llnorm{A}
    &= \sup\limits_{v \in \bS^{n_2 - 1}} \abs{v^{\top} \brk{Q^{\top} \paren{\widehat{\Sigma}^{\frac{1}{2}} \widehat{\Sigma}'^{- 1} \widehat{\Sigma}^{\frac{1}{2}} - \id_{d \times d}} Q - \frac{\eps}{4 n_2} \id_{n_2 \times n_2}} v} \nonumber \\
    &= \sup\limits_{v \in \bS^{n_2 - 1}} \abs{\paren{Q v}^{\top} \paren{\widehat{\Sigma}^{\frac{1}{2}} \widehat{\Sigma}'^{- 1} \widehat{\Sigma}^{\frac{1}{2}} - \id_{d \times d}} Q v - \frac{\eps}{4 n_2}} \nonumber \\
    &= \sup\limits_{v \in \bS^{n_2 - 1}} \abs{v_{\le d}^{\top} \paren{\widehat{\Sigma}^{\frac{1}{2}} \widehat{\Sigma}'^{- 1} \widehat{\Sigma}^{\frac{1}{2}} - \id_{d \times d}} v_{\le d} - \frac{\eps}{4 n_2}} \nonumber \\
    &= \max\brc{
        \begin{array}{c}
            \sup\limits_{v \in \bS^{n_2 - 1}} \brc{v_{\le d}^{\top} \paren{\widehat{\Sigma}^{\frac{1}{2}} \widehat{\Sigma}'^{- 1} \widehat{\Sigma}^{\frac{1}{2}} - \id_{d \times d}} v_{\le d}} - \frac{\eps}{4 n_2}, \\
            \frac{\eps}{4 n_2} - \inf\limits_{v \in \bS^{n_2 - 1}} \brc{v_{\le d}^{\top} \paren{\widehat{\Sigma}^{\frac{1}{2}} \widehat{\Sigma}'^{- 1} \widehat{\Sigma}^{\frac{1}{2}} - \id_{d \times d}} v_{\le d}}
        \end{array}
    }. \label{eq:spec_equiv}
\end{align}
We focus on the first term in the above $\max$.
By assumption, we have $\widehat{\Sigma}^{\frac{1}{2}} \widehat{\Sigma}'^{- 1} \widehat{\Sigma}^{\frac{1}{2}} - \id_{d \times d} \succeq 0$.
In this case, the worst-case upper bound is obtained when $\llnorm{v_{\le d}} = 1$ and $v_{> d} = 0$.
Consequently, in the following, we can treat $v$ as a vector in $\bR^d$ instead.
We have:
\begin{equation}
    \sup\limits_{v \in \bS^{d - 1}} \brc{v^{\top} \paren{\widehat{\Sigma}^{\frac{1}{2}} \widehat{\Sigma}'^{- 1} \widehat{\Sigma}^{\frac{1}{2}} - \id_{d \times d}} v} - \frac{\eps}{4 n_2} = \sup\limits_{v \in \bS^{d - 1}} \brc{v^{\top} \widehat{\Sigma}^{\frac{1}{2}} \widehat{\Sigma}'^{- 1} \widehat{\Sigma}^{\frac{1}{2}} v} - \paren{1 + \frac{\eps}{4 n_2}} \le \frac{4}{3} \gamma - \frac{\eps}{4 n_2}, \label{eq:spec_equiv1}
\end{equation}
where relied on Lemma~\ref{lem:stable_cov}, Fact~\ref{fact:sd_ord} and our bounds on $n_2$ and $\eps$ to argue that:
\begin{align}
    \paren{1 - \gamma} \widehat{\Sigma} \preceq \widehat{\Sigma}' \preceq \frac{1}{1 - \gamma} \widehat{\Sigma}
    &\iff \paren{1 - \gamma} \widehat{\Sigma}^{- 1} \preceq \widehat{\Sigma}'^{- 1} \preceq \frac{1}{1 - \gamma} \widehat{\Sigma}^{- 1} \nonumber \\
    &\iff \paren{1 - \gamma} \id \preceq \widehat{\Sigma}^{\frac{1}{2}} \widehat{\Sigma}'^{- 1} \widehat{\Sigma}^{\frac{1}{2}} \preceq \frac{1}{1 - \gamma} \id, \label{eq:sd_order_constraint}
\end{align}
and that $\gamma \le \frac{\eps}{4} \le \frac{1}{4}$.

For the other term in (\ref{eq:spec_equiv}), the assumption $\widehat{\Sigma}^{\frac{1}{2}} \widehat{\Sigma}'^{- 1} \widehat{\Sigma}^{\frac{1}{2}} \succeq \id$ implies that the worst-case upper bound is obtained when $v_{\le d} = 0$ and $\llnorm{v_{> d}} = 1$.
This yields the upper bound:
\begin{equation}
    \frac{\eps}{4 n_2} - \inf\limits_{v \in \bS^{n_2 - 1}} \brc{v_{\le d}^{\top} \paren{\widehat{\Sigma}^{\frac{1}{2}} \widehat{\Sigma}'^{- 1} \widehat{\Sigma}^{\frac{1}{2}} - \id_{d \times d}} v_{\le d}} \le \frac{\eps}{4 n_2}. \label{eq:spec_equiv2}
\end{equation}
Comparing (\ref{eq:spec_equiv1}) with (\ref{eq:spec_equiv2}), we get $\frac{\eps}{4 n_2} \le \frac{4}{3} \gamma - \frac{\eps}{4 n_2} \iff \eps \le \frac{64}{3} e^2 \lambda_0$, which holds by our assumptions on the range of $\eps$ and $\lambda_0$.
Going back to (\ref{eq:spec_equiv}), we get:
\[
    \llnorm{A} \le \frac{4}{3} \gamma - \frac{\eps}{4 n_2} = \frac{\eps}{4 n_2} \paren{\frac{16}{3} \cdot \frac{n_2 \gamma}{\eps} - 1} = \frac{\eps}{4 n_2} \paren{\frac{128 e^2 \lambda_0}{3} - 1}.
\]
Finally, we prove the upper bound on $\fnorm{A}^2$.
Let $\brc{\lambda_i}_{i \in \brk{d}}$ be the spectrum of $\widehat{\Sigma}^{\frac{1}{2}} \widehat{\Sigma}'^{- 1} \widehat{\Sigma}^{\frac{1}{2}}$.
The definition of the Frobenius norm, along with the linearity and cyclic properties of the trace, yield:
\begin{align}
    \fnorm{A}^2
    &= \tr\paren{\brk{Q^{\top} \paren{\widehat{\Sigma}^{\frac{1}{2}} \widehat{\Sigma}'^{- 1} \widehat{\Sigma}^{\frac{1}{2}} - \id_{d \times d}} Q - \frac{\eps}{4 n_2} \id_{n_2 \times n_2}}^2} \nonumber \\
    &= \tr\paren{Q^{\top} \paren{\widehat{\Sigma}^{\frac{1}{2}} \widehat{\Sigma}'^{- 1} \widehat{\Sigma}^{\frac{1}{2}} - \id_{d \times d}} Q Q^{\top} \paren{\widehat{\Sigma}^{\frac{1}{2}} \widehat{\Sigma}'^{- 1} \widehat{\Sigma}^{\frac{1}{2}} - \id_{d \times d}} Q} \nonumber \\
    &\quad\ - \frac{\eps}{2 n_2} \tr\paren{Q^{\top} \paren{\widehat{\Sigma}^{\frac{1}{2}} \widehat{\Sigma}'^{- 1} \widehat{\Sigma}^{\frac{1}{2}} - \id_{d \times d}} Q} + \frac{\eps^2}{16 n_2^2} \tr\paren{\id_{n_2 \times n_2}} \nonumber \\
    &= \tr\paren{\paren{\widehat{\Sigma}^{\frac{1}{2}} \widehat{\Sigma}'^{- 1} \widehat{\Sigma}^{\frac{1}{2}} - \id_{d \times d}}^2} - \frac{\eps}{2 n_2} \tr\paren{\widehat{\Sigma}^{\frac{1}{2}} \widehat{\Sigma}'^{- 1} \widehat{\Sigma}^{\frac{1}{2}} - \id_{d \times d}} + \frac{\eps^2}{16 n_2} \nonumber \\
    &= \sum\limits_{i \in \brk{d}} \paren{\lambda_i - 1}^2 - \frac{\eps}{2 n_2} \sum\limits_{i \in \brk{d}} \paren{\lambda_i - 1} + \frac{\eps^2}{16 n_2} \nonumber \\
    &= \sum\limits_{i \in \brk{d}} \brk{\paren{\lambda_i - 1}^2 - \frac{\eps}{2 n_2} \paren{\lambda_i - 1} + \frac{\eps^2}{16 n_2^2}} - \frac{d \eps^2}{16 n_2^2} + \frac{\eps^2}{16 n_2} \nonumber \\
    &= \sum\limits_{i \in \brk{d}} \brk{\lambda_i - \paren{1 + \frac{\eps}{16 n_2}}}^2 + \frac{\eps^2}{16 n_2^2} \paren{n_2 - d}. \label{eq:frob_norm}
\end{align}
Our goal now is to derive a worst-case upper bound on (\ref{eq:frob_norm}).
By assumption, we have that $\lambda_i \geq 1, \forall i \in \brk{d}$.
Due of this, Lemmas~\ref{lem:stable_cov} and~\ref{lem:dist_inverse} imply:
\begin{align}
    \mnorm{\tr}{\widehat{\Sigma}^{- \frac{1}{2}} \widehat{\Sigma}' \widehat{\Sigma}^{- \frac{1}{2}} - \id} \le \gamma \paren{1 + 2 \gamma}
    &\implies \mnorm{\tr}{\widehat{\Sigma}^{\frac{1}{2}} \widehat{\Sigma}'^{- 1} \widehat{\Sigma}^{\frac{1}{2}} - \id} \le \gamma \paren{1 + 2 \gamma} \nonumber \\
    &\iff \sum\limits_{i \in \brk{d}} \paren{\lambda_i - 1} \le \gamma \paren{1 + 2 \gamma} \le \frac{3}{2} \gamma, \label{eq:eig_constraint}
\end{align}
where the last inequality used the assumption that $n_2$ is large enough to ensure $\gamma \le \frac{\eps}{4} \le \frac{1}{4}$.

Combining (\ref{eq:frob_norm}) and (\ref{eq:eig_constraint}), along with the assumption $\lambda_i \geq 1, \forall i \in \brk{d}$, we get a constrained maximization problem that is solvable exactly via the KKT conditions.
The upper bound we obtain that way is:
\begin{align*}
    \fnorm{A}^2
    &\le d \brk{\paren{1 + \frac{3 \gamma}{2 d}} - \paren{1 + \frac{\eps}{4 n_2}}}^2 + \frac{\eps^2}{16 n_2^2} \paren{n_2 - d} \\
    &= d \paren{\frac{3}{2 d} \cdot \frac{8 e^2 \lambda_0}{n_2} - \frac{\eps}{4 n_2}}^2 + \frac{\eps^2}{16 n_2^2} \paren{n_2 - d} \\
    &= \frac{d}{n_2^2} \paren{\frac{12 e^2 \lambda_0}{d} - \frac{\eps}{4}}^2 + \frac{\eps^2}{16 n_2^2} \paren{n_2 - d} \\
    &= \frac{d \eps^2}{16 n_2^2} \brk{\paren{\frac{48 e^2 \lambda_0}{d \eps} - 1}^2 + \paren{\frac{n_2}{d} - 1}},
\end{align*}
completing the calculation.
\end{proof}

As a consequence of the previous lemma, we obtain the following corollary:

\begin{corollary}
\label{cor:matrix_bounds_supp}
In the setting of Lemmas~\ref{lem:priv_anal_multi1_dpi} and~\ref{lem:matrix_bounds}, we have $\min\brc{\frac{- \tr\paren{A}}{\llnorm{A}}, \frac{\tr\paren{A}^2}{\fnorm{A}^2}} \geq \frac{3 \useconstant{C2}}{256 e^2} \log\paren{\frac{1}{\delta}}$.
\end{corollary}

\begin{proof}
We will show separately for each term that it lower-bounded by the target bound.
For the first term, Lemma~\ref{lem:matrix_bounds}, as well as our bound on $n_2$ and the definition of $\gamma$, imply:
\begin{align}
    \frac{- \tr\paren{A}}{\llnorm{A}} \geq \frac{\frac{\eps}{4} - \frac{3}{2} \gamma}{\frac{\eps}{4 n_2} \paren{\frac{128 e^2 \lambda_0}{3}} - 1} = \frac{n_2 - \frac{6 n_2 \gamma}{\eps}}{\frac{128 e^2 \lambda_0}{3} - 1} = \frac{n_2 - \frac{48 e^2 \lambda_0}{\eps}}{\frac{128 e^2 \lambda_0}{3} - 1}
    &\geq \frac{\useconstant{C2} \frac{\lambda_0 \log\paren{\frac{1}{\delta}}}{\eps} - \frac{48 e^2 \lambda_0}{\eps}}{\frac{128 e^2 \lambda_0}{3}} \nonumber \\
    \overset{\paren{a}}&{\geq} \frac{\frac{\useconstant{C2}}{2} \cdot \frac{\lambda_0 \log\paren{\frac{1}{\delta}}}{\eps}}{\frac{128 e^2 \lambda_0}{3}} \nonumber \\
    &= \frac{3 \useconstant{C2}}{256 e^2} \cdot \frac{\log\paren{\frac{1}{\delta}}}{\eps} \nonumber \\
    \overset{\paren{b}}&{\geq} \frac{3 \useconstant{C2}}{256 e^2} \log\paren{\frac{1}{\delta}}, \label{eq:lb1}
\end{align}
where $\paren{a}$ relies on the assumption that $\useconstant{C2}$ is appropriately large, and $\paren{b}$ uses that $\eps \le 1$.

For the second term, working similarly to before yields:
\begin{align}
    \frac{\tr\paren{A}^2}{\fnorm{A}^2} \geq \frac{\paren{\frac{\eps}{4} - \frac{3}{2} \gamma}^2}{\frac{d \eps^2}{16 n_2^2} \brk{\paren{\frac{48 e^2 \lambda_0}{d \eps} - 1}^2 + \paren{\frac{n_2}{d} - 1}}}
    &= \frac{\paren{n_2 - \frac{6 n_2 \gamma}{\eps}}^2}{d \brk{\paren{\frac{48 e^2 \lambda_0}{d \eps} - 1}^2 + \paren{\frac{n_2}{d} - 1}}} \nonumber \\
    &= \frac{\paren{n_2 - \frac{48 e^2 \lambda_0}{\eps}}^2}{d \brk{\paren{\frac{48 e^2 \lambda_0}{d \eps} - 1}^2 + \paren{\frac{n_2}{d} - 1}}}. \label{eq:frob_ratio1}
\end{align}
We bound the numerator and the denominator of (\ref{eq:frob_ratio1}) separately.
For the numerator, by our bound on $n_2$, we get that:
\begin{equation}
    n_2 - \frac{48 e^2 \lambda_0}{\eps} \geq \frac{1}{2} n_2. \label{eq:num_lb}
\end{equation}
For the denominator, again by our bound on $n_2$, we have:
\begin{equation}
    \paren{\frac{48 e^2 \lambda_0}{d \eps} - 1}^2 + \paren{\frac{n_2}{d} - 1} \le \paren{\frac{48 e^2}{\useconstant{C2}} \cdot \frac{n_2}{d \log\paren{\frac{1}{\delta}}}}^2 + \frac{n_2}{d}. \label{eq:denom_ub}
\end{equation}
Bounding (\ref{eq:frob_ratio1}) using (\ref{eq:num_lb}) and (\ref{eq:denom_ub}), we get:
\begin{equation}
    \frac{\tr\paren{A}^2}{\fnorm{A}^2} \geq \frac{\paren{\frac{1}{2} n_2}^2}{d \brk{\paren{\frac{48 e^2}{\useconstant{C2}} \cdot \frac{n_2}{d \log\paren{\frac{1}{\delta}}}}^2 + \frac{n_2}{d}}} \geq \frac{\frac{1}{4}}{\paren{\frac{48 e^2}{\useconstant{C2}}}^2 \cdot \frac{1}{d \log^2\paren{\frac{1}{\delta}}} + \frac{1}{n_2 d}}. \label{eq:frob_ratio2}
\end{equation}
We have again for the denominator:
\begin{align}
    &\qquad\quad n_2 \geq \useconstant{C2} \frac{\lambda_0 \log\paren{\frac{1}{\delta}}}{\eps} \nonumber \\
    &\iff \frac{1}{n_2 d} \le \frac{\eps}{\useconstant{C2} \lambda_0 d \log\paren{\frac{1}{\delta}}} \le \frac{1}{\useconstant{C2} d^2 \log\paren{\frac{1}{\delta}}} \nonumber \\
    &\implies \paren{\frac{48 e^2}{\useconstant{C2}}}^2 \cdot \frac{1}{d \log^2\paren{\frac{1}{\delta}}} + \frac{1}{n_2 d} \le \brk{\paren{\frac{48 e^2}{\useconstant{C2}}}^2 + \frac{1}{\useconstant{C2}}} \frac{1}{d \log\paren{\frac{1}{\delta}}} \le \frac{2}{\useconstant{C2}} \cdot \frac{1}{d \log\paren{\frac{1}{\delta}}}, \label{eq:denom_ub2}
\end{align}
where the last inequality again relied on $\useconstant{C2}$ being appropriately large.

Thus, lower-bounding (\ref{eq:frob_ratio2}) using (\ref{eq:denom_ub2}) yields:
\begin{equation}
    \frac{\tr\paren{A}^2}{\fnorm{A}^2} \geq \frac{\useconstant{C2}}{8} d \log\paren{\frac{1}{\delta}}. \label{eq:lb2}
\end{equation}
Taking into account (\ref{eq:lb1}) and (\ref{eq:lb2}) yields the desired result.
\end{proof}

We now have everything we need to establish the main result of this section.

\begin{proposition}
\label{prop:priv_anal_multi1}
Let $\eps \in \brk{0, 1}$ and $\delta \in \brk{0, \frac{\eps}{10}}, \lambda_0 \geq d$.
Let us assume that $n_2 \geq \useconstant{C2} \frac{\lambda_0 \log\paren{\frac{1}{\delta}}}{\eps}$, for some appropriately large absolute constant $\useconstant{C2} \geq 1$.
Finally, for $n \coloneqq n_1 + 2 n_2$, let $X, X' \in \bR^{n \times d}$ be adjacent datasets and $R \subseteq \brk{n_1}$ be a representation set such that $\Psi\paren{X, R} = \Psi\paren{X', R} = 1$.
Then, for $z \sim \cU\paren{\bS^{n_2 - 1}}$, we have:
\[
    D_{e^{\frac{\eps}{2}}}\paren{\widehat{\mu}' + \sqrt{\paren{1 - \frac{1}{n_1}} n_2} W z \middle\| \widehat{\mu}' + \sqrt{\paren{1 - \frac{1}{n_1}} n_2} W' z} \le \delta.
\]
\end{proposition}

\begin{proof}
The proof will follow a structure similar to that of Proposition~\ref{prop:priv_anal_uni1}, i.e., establishing the desired upper bound amounts to upper-bounding the probability of the conditions identified in Lemma~\ref{lem:priv_anal_multi1_suff_cond}.
Indeed, by Lemma~\ref{lem:priv_anal_multi1_dpi}, (\ref{eq:multi_hs_div_simp1}), and Lemma~\ref{lem:priv_anal_multi1_suff_cond}, we get that:
\begin{align}
    &\quad\ D_{e^{\frac{\eps}{2}}}\paren{\widehat{\mu}' + \sqrt{\paren{1 - \frac{1}{n_1}} n_2} W z \middle\| \widehat{\mu}' + \sqrt{\paren{1 - \frac{1}{n_1}} n_2} W' z} \nonumber \\
    &\le D_{e^{\frac{\eps}{2}}}\paren{z_{\le d} \middle\| T} \nonumber\\
    &\le \pr{t \sim z_{\le d}}{f_{z_{\le d}}\paren{t} > e^{\frac{\eps}{2}} f_T\paren{t}} \nonumber\\
    &\le \pr{s \sim z_{\le d}}{\brc{s^{\top} \widehat{\Sigma}^{\frac{1}{2}} \widehat{\Sigma}'^{- 1} \widehat{\Sigma}^{\frac{1}{2}} s > \frac{1}{2}} \cup \brc{s^{\top} \paren{\widehat{\Sigma}^{\frac{1}{2}} \widehat{\Sigma}'^{- 1} \widehat{\Sigma}^{\frac{1}{2}} - \id} s > \frac{\eps}{4 n_2}}} \nonumber \\
    &\le \pr{s \sim z_{\le d}}{s^{\top} \widehat{\Sigma}^{\frac{1}{2}} \widehat{\Sigma}'^{- 1} \widehat{\Sigma}^{\frac{1}{2}} s > \frac{1}{2}} + \pr{s \sim z_{\le d}}{s^{\top} \paren{\widehat{\Sigma}^{\frac{1}{2}} \widehat{\Sigma}'^{- 1} \widehat{\Sigma}^{\frac{1}{2}} - \id} s > \frac{\eps}{4 n_2}}, \label{eq:union_bound_ub}
\end{align}
where the fact that $s \sim z_{\le d}$ follows from the definition of $s$ (see Lemma~\ref{lem:t_density}) and Fact~\ref{fact:uniform_rot_inv}.

We will analyze each of the terms of (\ref{eq:union_bound_ub}) separately.
For the first term, observe that, by Fact~\ref{fact:unit_sphere_proj_dist}, $\llnorm{s}^2 \sim \Beta\paren{\frac{d}{2}, \frac{n_2 - d}{2}}$.
Thus, by the previous and Fact~\ref{fact:beta_concentration}, we get:
\begin{align}
    &\quad\ \pr{s \sim z_{\le d}}{s^{\top} \widehat{\Sigma}^{\frac{1}{2}} \widehat{\Sigma}'^{- 1} \widehat{\Sigma}^{\frac{1}{2}} s > \frac{1}{2}} \nonumber \\
    &\le \pr{s \sim z_{\le d}}{\llnorm{s}^2 > \frac{3}{8}} \nonumber \\
    &\le 2 \exp\paren{- \useconstant{cbeta} \min\brc{\frac{\paren{n_2 - d}^2}{2 d} \paren{\frac{3}{8} - \frac{d}{n_2}}^2, \frac{n_2 - d}{2} \paren{\frac{3}{8} - \frac{d}{n_2}}}} \nonumber \\
    \overset{\paren{a}}&{=} 2 \exp\paren{- \useconstant{cbeta} \frac{n_2 - d}{2} \paren{\frac{3}{8} - \frac{d}{n_2}}} \nonumber \\
    \overset{\paren{b}}&{\le} \frac{\delta}{2}, \label{eq:term_bound1}
\end{align}
where $\paren{a}$ and $\paren{b}$ are implied by our bound on $n_2$ (assuming $\useconstant{C2}$ to be large enough).

We now turn our attention to the second term of (\ref{eq:union_bound_ub}).
To analyze this term, we will make two assumptions, without loss of generality.
The first assumption is that $z$ is generated according to the method of Fact~\ref{fact:random_unit_vector_gen}.
We note that this is a reasonable assumption to make.
Indeed, it is always possible to replace $z$ in $D_{e^{\frac{\eps}{2}}}\paren{\widehat{\mu}' + \sqrt{\paren{1 - \frac{1}{n_1}} n_2} W z \middle\| \widehat{\mu}' + \sqrt{\paren{1 - \frac{1}{n_1}} n_2} W' z}$ with a $z'$ that has been generated according to that method, and the value of the divergence term would still be the same, since $z$ and $z'$ have the same density function.
Moreover, it is always within our power as algorithm designers to assume that, when $z$ is drawn in Line~\ref{ln:post_ptr1}, it is generated by a subroutine that implements the method (especially given the fact that the only resources we are looking to minimize are sample and time complexity, but not randomness complexity).

The second assumption we will make is that $\widehat{\Sigma}^{\frac{1}{2}} \widehat{\Sigma}'^{- 1} \widehat{\Sigma}^{\frac{1}{2}} \succeq \id$.
Again, this is a reasonable assumption because, if $\widehat{\Sigma}^{\frac{1}{2}} \widehat{\Sigma}'^{- 1} \widehat{\Sigma}^{\frac{1}{2}}$ had an eigenvalue that is $< 1$, this would have a negative contribution to the LHS of $s^{\top} \paren{\widehat{\Sigma}^{\frac{1}{2}} \widehat{\Sigma}'^{- 1} \widehat{\Sigma}^{\frac{1}{2}} - \id} s > \frac{\eps}{4 n_2}$, thus making it harder to satisfy the inequality (i.e., the conditions imposed on $s$ would have to be more restrictive).
Thus, our assumption can only lead to an increase in the probability of the condition being satisfied, which justifies our choice.

Under the above two assumptions, we note that there will exist a random vector $G \sim \cN\paren{0, 1}^{\otimes n_2}$ such that $s_i = \frac{G_i}{\llnorm{G}}, \forall i \in \brk{d}$.
This follows from the fact that $s \sim z_{\le d}$.
Additionally, for $Q \coloneqq \paren{\id_{d \times d} \middle| 0_{d \times \paren{n_2 - d}}}$ (as in Lemma~\ref{lem:matrix_bounds}), we get that $s = Q \frac{G}{\llnorm{G}}$.
We substitute based on this in the second term of (\ref{eq:union_bound_ub}), and use the notation $A \coloneqq Q^{\top} \paren{\widehat{\Sigma}^{\frac{1}{2}} \widehat{\Sigma}'^{- 1} \widehat{\Sigma}^{\frac{1}{2}} - \id_{d \times d}} Q - \frac{\eps}{4 n_2} \id_{n_2 \times n_2}$ (again as in Lemma~\ref{lem:matrix_bounds} and Corollary~\ref{cor:matrix_bounds_supp}).
Then, Fact~\ref{fact:hanson_wright} and Corollary~\ref{cor:matrix_bounds_supp} yield:
\begin{align}
    &\quad\ \pr{s \sim z_{\le d}}{s^{\top} \paren{\widehat{\Sigma}^{\frac{1}{2}} \widehat{\Sigma}'^{- 1} \widehat{\Sigma}^{\frac{1}{2}} - \id} s > \frac{\eps}{4 n_2}} \nonumber \\
    &= \pr{G \sim \cN\paren{0, 1}^{\otimes n_2}}{\frac{G^{\top}}{\llnorm{G}} Q^{\top} \paren{\widehat{\Sigma}^{\frac{1}{2}} \widehat{\Sigma}'^{- 1} \widehat{\Sigma}^{\frac{1}{2}} - \id} Q \frac{G}{\llnorm{G}} > \frac{\eps}{4 n_2}} \nonumber \\
    &= \pr{G \sim \cN\paren{0, 1}^{\otimes n_2}}{G^{\top} Q^{\top} \paren{\widehat{\Sigma}^{\frac{1}{2}} \widehat{\Sigma}'^{- 1} \widehat{\Sigma}^{\frac{1}{2}} - \id} Q G > \frac{\eps}{4 n_2} \llnorm{G}^2} \nonumber \\
    &= \pr{G \sim \cN\paren{0, 1}^{\otimes n_2}}{G^{\top} A G > 0} \nonumber \\
    &= \pr{G \sim \cN\paren{0, 1}^{\otimes n_2}}{G^{\top} A G - \tr\paren{A} > - \tr\paren{A}} \nonumber \\
    &\le \pr{G \sim \cN\paren{0, 1}^{\otimes n_2}}{\abs{G^{\top} A G - \tr\paren{A}} \geq - \tr\paren{A}} \nonumber \\
    &\le 2 \exp\paren{- \useconstant{chs} \min\brc{\frac{\tr\paren{A}^2}{\fnorm{A}^2}, \frac{- \tr\paren{A}}{\llnorm{A}}}} \nonumber \\
    &\le 2 \exp\paren{- \frac{3 \useconstant{C2} \useconstant{chs}}{256 e^2} \log\paren{\frac{1}{\delta}}} \nonumber \\
    &\le \frac{\delta}{2}, \label{eq:term_bound2}
\end{align}
where, in the last inequality, we assume that $\useconstant{C2}$ is appropriately large.

Upper-bounding (\ref{eq:union_bound_ub}) using (\ref{eq:term_bound1}) and (\ref{eq:term_bound2}) completes the proof.
\end{proof}

\subsubsection{Analyzing the First Term of (\ref{eq:weak_triangle_ineq})}
\label{subsubsec:priv_anal2}

We now analyze the first term of (\ref{eq:weak_triangle_ineq}).
As in the previous section, the proof will be broken down into a number of intermediate steps, which we outline below:
\begin{enumerate}
    \item we start by showing that the Hockey-Stick divergence term of interest can be equivalently expressed as the divergence between $z_{\le d}$ and $z_{\le d} + \ell$, where $\ell \coloneqq \sqrt{\frac{n_1}{n_1 - 1} \cdot \frac{1}{n_2}} \Lambda^{- \frac{1}{2}} U^{\top} \paren{\widehat{\mu}' - \widehat{\mu}}$\footnote{$\Lambda$ and $U$ come from the spectral decomposition of $\widehat{\Sigma}$.}
    and $\llnorm{\ell}$ is ``small'' (Lemma~\ref{lem:priv_anal2_dpi}),
    \item subsequently, we identify sufficient conditions for points $t \in \supp\paren{z_{\le d}}$ to imply that log-density ratio $\ln\paren{\frac{f_{z_{\le d}}\paren{t}}{f_{z_{\le d} + \ell}\paren{t}}}$ is at most $\frac{\eps}{2}$ (Lemma~\ref{lem:priv_anal2_suff_cond}),
    \item finally, we will show that our sample complexity suffices for the first term of (\ref{eq:weak_triangle_ineq}) to be at most $\delta$ (Proposition~\ref{prop:priv_anal2}).
\end{enumerate}
The argument is structured similarly to Lemma $4.4$ from~\cite{GhaziHKM23}.
However, differences arise along the way, due to the fact that the corresponding proof in~\cite{GhaziHKM23} used that datapoints are truncated, whereas we have to appeal to Lemma~\ref{lem:stable_mean}.
Additionally, the original proof of Lemma $4.4$ given in~\cite{GhaziHKM23} contains a bug in the step where it is argued that the sample complexity bound suffices for the log-density ratio to be upper-bounded by $\eps$ (the step corresponding to Lemma~\ref{lem:priv_anal2_suff_cond}).
The bug is fixable (with a minor adjustment in the final sample complexity), and has been acknowledged in private communication with the authors of that work.
For all the above reasons, we present the proof in full detail.

We start with the lemma that implements the first step.

\begin{lemma}
\label{lem:priv_anal2_dpi}
Let $\eps \in \brk{0, 1}$ and $\delta \in \brk{0, \frac{\eps}{10}}, \lambda_0 \geq d$.
Also, assume that $n_1 \geq \useconstant{C1} \frac{\sqrt{\lambda_0} \log\paren{\frac{1}{\delta}}}{\eps}$ and $n_2 \geq \useconstant{C2} \frac{\lambda_0 \log\paren{\frac{1}{\delta}}}{\eps}$ for appropriately large absolute constants $\useconstant{C1}, \useconstant{C2} \geq 1$.
Finally, for $n \coloneqq n_1 + 2 n_2$, let $X, X' \in \bR^{n \times d}$ be adjacent datasets and $R \subseteq \brk{n_1}$ be a representation set such that $\Psi\paren{X, R} = \Psi\paren{X', R} = 1$.
Then, for $z \sim \cU\paren{\bS^{n_2 - 1}}$, we have:
\[
    D_{e^{\frac{\eps}{2}}}\paren{\widehat{\mu} + \sqrt{\paren{1 - \frac{1}{n_1}} n_2} W z \middle\| \widehat{\mu}' + \sqrt{\paren{1 - \frac{1}{n_1}} n_2} W z} = D_{e^{\frac{\eps}{2}}}\paren{z_{\le d} \middle\| z_{\le d} + \ell},
\]
for $\ell \coloneqq \sqrt{\frac{n_1}{n_1 - 1} \cdot \frac{1}{n_2}} \Lambda^{- \frac{1}{2}} U^{\top} \paren{\widehat{\mu}' - \widehat{\mu}}$ where $\Lambda \succ 0$ is a diagonal matrix and $U$ is a rotation matrix such that $\widehat{\Sigma} = U \Lambda U^{\top}$.
Furthermore, we have $\llnorm{\ell} \le \frac{\sqrt{114 \lambda_0} e}{n_1 \sqrt{n_2}}$.
\end{lemma}

\begin{proof}
We start by repeating a number of arguments from the beginning of the proof of Lemma~\ref{lem:priv_anal_multi1_dpi}, namely the DPI-based argument of (\ref{eq:multi_dpi_app1}), and the SVD-based argument of (\ref{eq:multi_hs_div_simp1}).
This yields:
\begin{align}
    &\quad\ D_{e^{\frac{\eps}{2}}}\paren{\widehat{\mu} + \sqrt{\paren{1 - \frac{1}{n_1}} n_2} W z \middle\| \widehat{\mu}' + \sqrt{\paren{1 - \frac{1}{n_1}} n_2} W z} \nonumber \\
    &= D_{e^{\frac{\eps}{2}}}\paren{z_{\le d} \middle\| \sqrt{\frac{n_1}{n_1 - 1} \cdot \frac{1}{n_2}} \Lambda^{- \frac{1}{2}} U^{\top} \paren{\widehat{\mu}' - \widehat{\mu}} + z_{\le d}}, \label{eq:multi1_dpi_app}
\end{align}
where $U \in \bR^{d \times d}$ is a rotation matrix and $\Lambda \in \bR^{d \times d}$ is a diagonal positive-definite matrix such that $U \Lambda U^{\top} = W W^{\top} = \widehat{\Sigma}$.
We set $\ell \coloneqq \sqrt{\frac{n_1}{n_1 - 1} \cdot \frac{1}{n_2}} \Lambda^{- \frac{1}{2}} U^{\top} \paren{\widehat{\mu}' - \widehat{\mu}}$.
By our bounds on $n_1$ and $n_2$, and the assumption that $\Psi\paren{X, R} = \Psi\paren{X', R} = 1$, the conditions of Lemma~\ref{lem:stable_mean} are satisfied, so the lemma yields:
\begin{align*}
    \llnorm{\ell} = \sqrt{\frac{n_1}{n_1 - 1} \cdot \frac{1}{n_2}} \llnorm{\Lambda^{- \frac{1}{2}} U^{\top} \paren{\widehat{\mu}' - \widehat{\mu}}}
    \overset{\paren{a}}&{=} \sqrt{\frac{n_1}{n_1 - 1} \cdot \frac{1}{n_2}} \mnorm{\Sigma}{\widehat{\mu}' - \widehat{\mu}} \\
    &\le \sqrt{\frac{n_1}{n_1 - 1}} \cdot \frac{\sqrt{\paren{1 + 2 \gamma} 38 \lambda_0} e}{n_1 \sqrt{n_2}} \\
    \overset{\paren{b}}&{\le} \frac{\sqrt{114 \lambda_0} e}{n_1 \sqrt{n_2}},
\end{align*}
where $\paren{a}$ used the observation that:
\[
    \llnorm{\Lambda^{- \frac{1}{2}} U^{\top} \paren{\widehat{\mu}' - \widehat{\mu}}} = \sqrt{\paren{\widehat{\mu}' - \widehat{\mu}}^{\top} U \Lambda^{- 1} U^{\top} \paren{\widehat{\mu}' - \widehat{\mu}}} = \sqrt{\paren{\widehat{\mu}' - \widehat{\mu}}^{\top} \widehat{\Sigma}^{- 1} \paren{\widehat{\mu}' - \widehat{\mu}}} = \mnorm{\widehat{\Sigma}}{\widehat{\mu}' - \widehat{\mu}},
\]
and $\paren{b}$ used our bounds on $n_1, n_2$, and $\eps$ to argue that $n_1 \geq 2 \iff \frac{n_1}{n_1 - 1} \le 2$ and $\gamma \le \frac{\eps}{4} \le \frac{1}{4}$.
\end{proof}

As with the other term of (\ref{eq:weak_triangle_ineq}), we must now reason about the quantity:
\begin{equation}
    D_{e^{\frac{\eps}{2}}}\paren{z_{\le d} \middle\| z_{\le d} + \ell} \le \pr{t \sim z_{\le d}}{f_{z_{\le d}}\paren{t} > e^{\frac{\eps}{2}} f_{z_{\le d} + \ell}\paren{t}}. \label{eq:multi_hs_div_ub2}
\end{equation}
Thus, we need to identify the density of $z_{\le d} + \ell$, and identify sufficient conditions for the log-density ratio to satisfy $\ln\paren{\frac{f_{z_{\le d}}\paren{t}}{f_{z_{\le d} + \ell}\paren{t}}} \le \frac{\eps}{2}$.
We do so in the following lemma:

\begin{lemma}
\label{lem:priv_anal2_suff_cond}
In the setting of Lemma~\ref{lem:priv_anal2_dpi}, we have that:
\[
    \llnorm{t} \le 0.9 \text{ and } \abs{\iprod{t, \ell}} \le \frac{1}{50} \cdot \frac{\eps}{n_2} \le 0.01 \implies \ln\paren{\frac{f_{z_{\le d}}\paren{t}}{f_{z_{\le d} + \ell}\paren{t}}} \le \frac{\eps}{2}.
\]
\end{lemma}

\begin{proof}
We start by identifying the density of $z_{\le d} + \ell$.
Fact~\ref{fact:transformations} yields that $f_T\paren{t} = f_{z_{\le d}}\paren{t - \ell}$.
Additionally, by the conclusion of Lemma~\ref{lem:priv_anal2_dpi} and our bounds on $n_1, n_2$, and $\delta$, we get:
\begin{equation}
    \llnorm{\ell} \le \frac{1}{5} \cdot \frac{\eps}{\sqrt{n_2}} \le 0.1. \label{eq:abs_l_ineq}
\end{equation}
We now work towards upper-bounding the log-density ratio $\ln\paren{\frac{f_{z_{\le d}}\paren{t}}{f_{z_{\le d} + \ell}\paren{t}}}$ using our assumptions.
We highlight that $z_{\le d}$ is supported on the unit ball in $d$ dimensions, whereas $z_{\le d} + \ell$ is supported on the ball of radius $1$ centered at the point $\ell$ (which (\ref{eq:abs_l_ineq}) implies has small distance from the origin).
Observe that, by our assumptions and (\ref{eq:abs_l_ineq}), we have:
\begin{equation}
    \llnorm{t - \ell}^2 = \llnorm{t}^2 - 2 \iprod{t, \ell} + \llnorm{\ell}^2 \le \llnorm{t}^2 + 2 \abs{\iprod{t, \ell}} + \llnorm{\ell}^2 \le 0.81 + 2 \cdot 0.01 + 0.01 < 0.9. \label{eq:l2_dist_ub}
\end{equation}
The above implies that our conditions ensure that $t$ will be within the support of $z_{\le d} + \ell$, so $\ln\paren{\frac{f_{z_{\le d}}\paren{t}}{f_{z_{\le d} + \ell}\paren{t}}}$ will be finite.
Building on this observation, Fact~\ref{fact:unit_sphere_proj_dist} and the standard inequality $\ln\paren{x} \le x - 1, \forall x > 0$ yield for the log-density ratio:
\begin{align*}
    \ln\paren{\frac{f_{z_{\le d}}\paren{t}}{f_{z_{\le d} + \ell}\paren{t}}} = \frac{n_2 - d - 2}{2} \ln\paren{\frac{1 - \llnorm{t}^2}{1 - \llnorm{t - \ell}^2}}
    &< \frac{n_2}{2} \paren{\frac{1 - \llnorm{t}^2}{1 - \llnorm{t - \ell}^2} - 1} \\
    &= \frac{n_2}{2} \cdot \frac{\llnorm{\ell}^2 - 2 \iprod{t, \ell}}{1 - \llnorm{t - \ell}^2} \\
    &\le \frac{n_2}{2} \cdot \frac{\llnorm{\ell}^2 + 2 \abs{\iprod{t, \ell}}}{1 - \llnorm{t - \ell}^2} \\
    \overset{\paren{a}}&{\le} \frac{n_2}{2} \cdot 10 \paren{\frac{1}{25} \cdot \frac{\eps^2}{n_2} + \frac{1}{25} \cdot \frac{\eps}{n_2}} \\
    \overset{\paren{b}}&{\le} \frac{n_2}{2} \cdot 10 \cdot \frac{2}{25} \cdot \frac{\eps}{n_2} \\
    &= \frac{10}{25} \eps \\
    &\le \frac{\eps}{2},
\end{align*}
where $\paren{a}$ used (\ref{eq:abs_l_ineq}) and (\ref{eq:l2_dist_ub}), as well as our assumption about $\abs{\iprod{t, \ell}}$, and $\paren{b}$ used that $\eps \le 1$.
\end{proof}

We conclude by proving the main result of the section.

\begin{proposition}
\label{prop:priv_anal2}
Let $\eps \in \brk{0, 1}$ and $\delta \in \brk{0, \frac{\eps}{10}}, \lambda_0 \geq d$.
Also, assume that $n_1 \geq \useconstant{C1} \frac{\sqrt{\lambda_0} \log\paren{\frac{1}{\delta}}}{\eps}$ and $n_2 \geq \useconstant{C2} \frac{\lambda_0 \log\paren{\frac{1}{\delta}}}{\eps}$ for appropriately large absolute constants $\useconstant{C1}, \useconstant{C2} \geq 1$.
Finally, let $X, X' \in \bR^{n \times d}$ be adjacent datasets and $R \subseteq \brk{n_1}$ such that $\Psi\paren{X, R} = \Psi\paren{X', R} = 1$.
Then, for $z \sim \cU\paren{\bS^{n_2 - 1}}$, we have:
\[
    D_{e^{\frac{\eps}{2}}}\paren{\widehat{\mu} + \sqrt{\paren{1 - \frac{1}{n_1}} n_2} W z \middle\| \widehat{\mu}' + \sqrt{\paren{1 - \frac{1}{n_1}} n_2} W z} \le \delta.
\]
\end{proposition}

\begin{proof}
We work as in the proofs of Propositions~\ref{prop:priv_anal_uni1} and~\ref{prop:priv_anal_multi1}.
Lemmas~\ref{lem:priv_anal2_dpi} and~\ref{lem:priv_anal2_suff_cond}, as well as (\ref{eq:multi_hs_div_ub2}), yield:
\begin{align}
    &\quad\ D_{e^{\frac{\eps}{2}}}\paren{\widehat{\mu} + \sqrt{\paren{1 - \frac{1}{n_1}} n_2} W z \middle\| \widehat{\mu}' + \sqrt{\paren{1 - \frac{1}{n_1}} n_2} W z} \nonumber \\
    &= D_{e^{\frac{\eps}{2}}}\paren{z_{\le d} \middle\| z_{\le d} + \ell} \nonumber \\
    &\le \pr{t \sim z_{\le d}}{f_{z_{\le d}}\paren{t} > e^{\frac{\eps}{2}} f_{z_{\le d} + \ell}\paren{t}} \nonumber \\
    &\le \pr{t \sim z_{\le d}}{\brc{\llnorm{t} > 0.9} \cup \brc{\abs{\iprod{\ell, t}} > \frac{1}{50} \cdot \frac{\eps}{n_2}}} \nonumber \\
    &\le \pr{t \sim z_{\le d}}{\llnorm{t} > 0.9} + \pr{t \sim z_{\le d}}{\abs{\iprod{\ell, t}} > \frac{1}{50} \cdot \frac{\eps}{n_2}} \nonumber \\
    &= \pr{t \sim z_{\le d}}{\llnorm{t} > 0.9} + \pr{t \sim z_{\le d}}{\abs{\iprod{\frac{\ell}{\llnorm{\ell}}, t}} > \frac{1}{50} \cdot \frac{\eps}{n_2} \cdot \frac{1}{\llnorm{\ell}}} \nonumber \\
    &\le \pr{t \sim z_{\le d}}{\llnorm{t} > 0.9} + \pr{t \sim z_{\le d}}{\abs{\iprod{\frac{\ell}{\llnorm{\ell}}, t}} > \frac{1}{50} \cdot \frac{\eps}{n_2} \cdot \frac{n_1 \sqrt{n_2}}{\sqrt{114 \lambda_0} e}} \nonumber \\
    &\le \pr{t \sim z_{\le d}}{\llnorm{t} > 0.9} + \pr{t \sim z_{\le d}}{\abs{\iprod{\frac{\ell}{\llnorm{\ell}}, t}} > \frac{\useconstant{C1}}{50 \sqrt{114}} \cdot \frac{\log\paren{\frac{1}{\delta}}}{\sqrt{n_2}}}, \label{eq:union_bound_ub1}
\end{align}
where the last inequality used our bound on $n_1$.

We handle each term of (\ref{eq:union_bound_ub1}) as we did in the conclusion of Proposition~\ref{prop:priv_anal_uni1}.
For the first term, Fact~\ref{fact:unit_sphere_proj_dist} yields that $\llnorm{t}^2 \sim \Beta\paren{\frac{d}{2}, \frac{n_2 - d}{2}}$.
Additionally, our bound on $n_2$ implies that $n_2 > 4 d$.
Thus, we get from Fact~\ref{fact:beta_concentration}:
\begin{align}
    \pr{t \sim z_{\le d}}{\llnorm{t} > 0.9}
    &= \pr{t \sim z_{\le d}}{\llnorm{t}^2 > 0.81} \nonumber \\
    &\le 2 \exp\paren{- \useconstant{cbeta} \min\brc{\frac{\paren{n_2 - d}^2}{2 d} \paren{0.81 - \frac{d}{n_2}}^2, \frac{n_2 - d}{2} \paren{0.81 - \frac{d}{n_2}}}} \nonumber \\
    &= 2 \exp\paren{- \useconstant{cbeta} \frac{n_2 - d}{2} \paren{0.81 - \frac{d}{n_2}}} \nonumber \\
    &\le 2 \exp\paren{- \frac{\useconstant{cbeta} n_2}{10}} \nonumber \\
    &\le \frac{\delta}{2}. \label{eq:ub_term1}
\end{align}
For the second term of (\ref{eq:union_bound_ub1}), Facts~\ref{fact:uniform_rot_inv} and~\ref{fact:unit_sphere_proj_dist} imply that $\iprod{\frac{\ell}{\llnorm{\ell}}, t} \sim z_1$, yielding that $\iprod{\frac{\ell}{\llnorm{\ell}}, t}^2 \sim \Beta\paren{\frac{1}{2}, \frac{n_2 - 1}{2}}$.
Thus, by our bound on $n_2$ and $\delta$, Fact~\ref{fact:beta_concentration} yields:
\begin{align}
    &\quad\ \pr{t \sim z_{\le d}}{\abs{\iprod{\frac{\ell}{\llnorm{\ell}}, t}} > \frac{\useconstant{C1}}{50 \sqrt{114}} \cdot \frac{\log\paren{\frac{1}{\delta}}}{\sqrt{n_2}}} \nonumber \\
    &= \pr{t \sim z_{\le d}}{\iprod{\frac{\ell}{\llnorm{\ell}}, t}^2 > \frac{\useconstant{C1}^2}{285000} \cdot \frac{\log^2\paren{\frac{1}{\delta}}}{n_2}} \nonumber \\
    &\le 2 \exp\paren{- \useconstant{cbeta} \min\brc{\frac{\paren{n_2 - 1}^2}{2} \paren{\frac{\useconstant{C1}^2}{285000} \cdot \frac{\log^2\paren{\frac{1}{\delta}}}{n_2} - \frac{1}{n_2}}^2, \frac{n_2 - 1}{2} \paren{\frac{\useconstant{C1}^2}{285000} \cdot \frac{\log^2\paren{\frac{1}{\delta}}}{n_2} - \frac{1}{n_2}}}} \nonumber \\
    &= 2 \exp\paren{- \useconstant{cbeta} \frac{n_2 - 1}{2} \paren{\frac{\useconstant{C1}^2}{285000} \cdot \frac{\log^2\paren{\frac{1}{\delta}}}{n_2} - \frac{1}{n_2}}} \nonumber \\
    &\le 2 \exp\paren{- \useconstant{cbeta} \frac{n_2 - 1}{2} \cdot \frac{\useconstant{C1}^2}{570000} \cdot \frac{\log^2\paren{\frac{1}{\delta}}}{n_2}} \nonumber \\
    &= 2 \exp\paren{- \useconstant{C1}^2 \useconstant{cbeta} \frac{n_2 - 1}{n_2} \cdot \frac{\log^2\paren{\frac{1}{\delta}}}{1140000}} \nonumber \\
    &\le 2 \exp\paren{- \frac{\useconstant{C1}^2 \useconstant{cbeta}}{2280000} \log^2\paren{\frac{1}{\delta}}} \nonumber \\
    &\le \frac{\delta}{2}. \label{eq:ub_term2}
\end{align}
Upper-bounding (\ref{eq:union_bound_ub1}) using (\ref{eq:ub_term1}) and (\ref{eq:ub_term2}) yields the desired result.
\end{proof}

\subsubsection{Completing the Proof of the Privacy Guarantee}
\label{subsubsec:priv_guarantee}

Thanks to the work undertaken in the previous sections, we can now complete the proof of the privacy guarantee.
We start with a statement that completes the calculation we initiated at the start of Section~\ref{subsec:priv_proof}.

\begin{corollary}
\label{cor:weak_triangle_ineq_full_ub}
Let $\eps \in \brk{0, 1}$ and $\delta \in \brk{0, \frac{\eps}{10}}, \lambda_0 \geq d$.
Also, assume that $n_1 \geq \useconstant{C1} \frac{\sqrt{\lambda_0} \log\paren{\frac{1}{\delta}}}{\eps}$ and $n_2 \geq \useconstant{C2} \frac{\lambda_0 \log\paren{\frac{1}{\delta}}}{\eps}$ for appropriately large absolute constants $\useconstant{C1}, \useconstant{C2} \geq 1$.
Finally, for $n \coloneqq n_1 + 2 n_2$, let $X, X' \in \bR^{n \times d}$ be adjacent datasets and $R \subseteq \brk{n_1}$ such that $\Psi\paren{X, R} = \Psi\paren{X', R} = 1$.
Then, for $z \sim \cU\paren{\bS^{n_2 - 1}}$, we have:
\[
    D_{e^{\eps}}\paren{\widehat{\mu} + \sqrt{\paren{1 - \frac{1}{n_1}} n_2} W z \middle\| \widehat{\mu}' + \sqrt{\paren{1 - \frac{1}{n_1}} n_2} W' z} \le \delta.
\]
\end{corollary}

\begin{proof}
We set $\delta \to \frac{\delta}{4}$ in Propositions~\ref{prop:priv_anal_multi1} and~\ref{prop:priv_anal2}.
Then, the result follows directly from (\ref{eq:weak_triangle_ineq}).
\end{proof}

In all of the above, we have been assuming that $\Psi\paren{X, R} = \Psi\paren{X', R} = 1$.
At this point, we will also take into account the effect that privately checking this has on the privacy guarantees.
This leads us to the main theorem of this section, which brings the proof of the privacy guarantee full circle.

\begin{theorem}
\label{thm:priv_guarantee}
Let $\eps \in \brk{0, 1}$ and $\delta \in \brk{0, \frac{\eps}{10}}, \lambda_0 \geq d$.
Also, assume that $n = \cO\paren{\frac{\lambda_0 \log\paren{\frac{1}{\delta}}}{\eps}}$.
Then, Algorithm~\ref{alg:sampling} is $\paren{\eps, \delta}$-DP.
\end{theorem}

\begin{proof}
The proof follows the same structure as the argument in the proof of Lemma $15$ in~\cite{BrownHS23}.
First, due to our choice of $M$ in Line~\ref{ln:init}, the set $R$ sampled in Line~\ref{ln:repr_set} must be degree-representative, except with probability $\frac{\delta}{6}$ (see Fact~\ref{fact:hyp_geom_conc}).
Under this condition, Lemmas~\ref{lem:bhs_ptr} and~\ref{lem:stable_conditions_ptr} imply that $\cM_{\mathrm{PTR}}^{\paren{\frac{\eps}{3}, \frac{\delta}{6}}}$ satisfies $\paren{\frac{\eps}{3}, \frac{\delta}{6}}$-DP.
Finally, under both of the above conditions, by setting $\eps \to \frac{2}{3} \eps$ and $\delta \to \frac{2}{3} \delta$ in Corollary~\ref{cor:weak_triangle_ineq_full_ub}, we get that the steps in Lines~\ref{ln:post_ptr1}-\ref{ln:post_ptr2} of Algorithm~\ref{alg:sampling} satisfy $\paren{\frac{2}{3} \eps, \frac{2}{3} \delta}$-DP.
Given the above, the desired result follows directly from Lemma~\ref{lem:composition_halting}.
\end{proof}

\subsection{Putting Everything Together}
\label{subsec:wrap}

Having established the utility and privacy guarantees in Theorems~\ref{thm:uti_guarantee} and~\ref{thm:priv_guarantee}, respectively, it remains to put the two together, which will allow us to identify the final sample complexity of our algorithm, thus completing the proof of Theorem~\ref{thm:main_thm_formal}.
We do so in this section.

\begin{proof}[Proof of Theorem~\ref{thm:main_thm_formal}]
We assume that we have $n$ samples, of which the first $n_1$ are fed to \StableCov, and the last $2 n_2$ are fed to \StableMean.
By Theorem~\ref{thm:priv_guarantee}, we get that $n_1 \geq \Theta\paren{\frac{\sqrt{\lambda_0} \log\paren{\frac{1}{\delta}}}{\eps}}$ and $n_2 \geq \Theta\paren{\frac{\lambda_0 \log\paren{\frac{1}{\delta}}}{\eps}}$, i.e., it must be the case that $n \geq \Theta\paren{\frac{\lambda_0 \log\paren{\frac{1}{\delta}}}{\eps}}$.
By Theorem~\ref{thm:uti_guarantee}, we get that $\lambda_0 \geq \Theta\paren{d + \sqrt{d \log\paren{\frac{3 n}{\alpha}}} + \log\paren{\frac{n}{\alpha}}}$ and $n_2 \geq \Theta\paren{d + \log\paren{\frac{1}{\alpha}}}$.
The latter condition is trivially satisfied by the bound on $n_2$ implied by Theorem~\ref{thm:priv_guarantee}.
Thus, it remains now to identify a bound on $n$ such that:
\[
    n \geq \Theta\paren{\frac{\lambda_0 \log\paren{\frac{1}{\delta}}}{\eps}} \geq \Theta\paren{\frac{\paren{d + \sqrt{d \ln\paren{\frac{3 n}{\alpha}}} + \ln\paren{\frac{n}{\alpha}}} \log\paren{\frac{1}{\delta}}}{\eps}}.
\]
The above is satisfied when $n = \cO\paren{\frac{\log\paren{\frac{1}{\delta}}}{\eps} \paren{d + \log\paren{\frac{\log\paren{\frac{1}{\delta}}}{\alpha \eps}}}}$.
\end{proof}

\begin{remark}
\label{rem:known_cov}
We note that it is possible to also solve the case where our samples come from a Gaussian with \emph{known covariance matrix} (e.g., $\cN\paren{\mu, \id}$) using a variant of Algorithm~\ref{alg:sampling}.
Indeed, we would not need to use \StableCov, and \StableMean\ would receive $\id$ as input in the place of $\widehat{\Sigma}$.
Because of that our bounds would only involve $n_1$, and we would get a final sample complexity of $n = \cO\paren{\frac{\log\paren{\frac{1}{\delta}}}{\eps} \sqrt{d + \log\paren{\frac{\log\paren{\frac{1}{\delta}}}{\alpha \eps}}}}$.
This roughly matches the corresponding result of~\cite{GhaziHKM23}. 
\end{remark}

\section*{Acknowledgments}
\label{sec:ack}

The authors would like to thank Gavin Brown for helpful comments and discussions.

\bibliographystyle{alpha}
\bibliography{biblio.bib}

\newcommand{\etalchar}[1]{$^{#1}$}
\begin{thebibliography}{KMS{\etalchar{+}}22b}

\bibitem[AAK21]{AdenAliAK21}
Ishaq {Aden-Ali}, Hassan Ashtiani, and Gautam Kamath.
\newblock On the sample complexity of privately learning unbounded high-dimensional gaussians.
\newblock In {\em Proceedings of the 32nd International Conference on Algorithmic Learning Theory}, ALT '21, pages 185--216. JMLR, Inc., 2021.

\bibitem[AGSV20]{AxelrodGSV20}
Brian Axelrod, Shivam Garg, Vatsal Sharan, and Gregory Valiant.
\newblock Sample amplification: Increasing dataset size even when learning is impossible.
\newblock In {\em Proceedings of the 37th International Conference on Machine Learning}, ICML '20, pages 442--451. JMLR, Inc., 2020.

\bibitem[AKM{\etalchar{+}}25]{AgarwalKMMSU25}
Sushant Agarwal, Gautam Kamath, Mahbod Majid, Argyris Mouzakis, Rose Silver, and Jonathan Ullman.
\newblock Private mean estimation with person-level differential privacy.
\newblock In {\em Proceedings of the 2025 Annual ACM-SIAM Symposium on Discrete Algorithms (SODA)}, SODA '25, pages 2819--2880. SIAM, 2025.

\bibitem[AKT{\etalchar{+}}23]{AlabiKTVZ23}
Daniel Alabi, Pravesh~K Kothari, Pranay Tankala, Prayaag Venkat, and Fred Zhang.
\newblock Privately estimating a {G}aussian: Efficient, robust and optimal.
\newblock In {\em Proceedings of the 55th Annual ACM Symposium on the Theory of Computing}, STOC '23. ACM, 2023.

\bibitem[AL22]{AshtianiL22}
Hassan Ashtiani and Christopher Liaw.
\newblock Private and polynomial time algorithms for learning {G}aussians and beyond.
\newblock In {\em Proceedings of the 35th Annual Conference on Learning Theory}, COLT '22, pages 1075--1076, 2022.

\bibitem[ALNP23]{AumullerLNP23}
Martin Aum{\"u}ller, Christian~Janos Lebeda, Boel Nelson, and Rasmus Pagh.
\newblock Plan: variance-aware private mean estimation.
\newblock {\em arXiv preprint arXiv:2306.08745}, 2023.

\bibitem[AUZ23]{AsiUZ23}
Hilal Asi, Jonathan Ullman, and Lydia Zakynthinou.
\newblock From robustness to privacy and back.
\newblock {\em arXiv preprint arXiv:2302.01855}, 2023.

\bibitem[BBC{\etalchar{+}}23]{BenDavidBCKS23}
Shai {Ben-David}, Alex Bie, Cl\'ement~L. Canonne, Gautam Kamath, and Vikrant Singhal.
\newblock Private distribution learning with public data: The view from sample compression.
\newblock In {\em Advances in Neural Information Processing Systems 36}, NeurIPS '23, pages 7184--7215. Curran Associates, Inc., 2023.

\bibitem[BDKU20]{BiswasDKU20}
Sourav Biswas, Yihe Dong, Gautam Kamath, and Jonathan Ullman.
\newblock Coinpress: Practical private mean and covariance estimation.
\newblock In {\em Advances in Neural Information Processing Systems 33}, NeurIPS '20, pages 14475--14485. Curran Associates, Inc., 2020.

\bibitem[BGS{\etalchar{+}}21]{BrownGSUZ21}
Gavin Brown, Marco Gaboardi, Adam Smith, Jonathan Ullman, and Lydia Zakynthinou.
\newblock Covariance-aware private mean estimation without private covariance estimation.
\newblock In {\em Advances in Neural Information Processing Systems 34}, NeurIPS '21. Curran Associates, Inc., 2021.

\bibitem[BHS23]{BrownHS23}
Gavin Brown, Samuel~B Hopkins, and Adam Smith.
\newblock Fast, sample-efficient, affine-invariant private mean and covariance estimation for subgaussian distributions.
\newblock In {\em Proceedings of the 36th Annual Conference on Learning Theory}, COLT '23, pages 5578--5579, 2023.

\bibitem[BKS22]{BieKS22}
Alex Bie, Gautam Kamath, and Vikrant Singhal.
\newblock Private estimation with public data.
\newblock In {\em Advances in Neural Information Processing Systems 35}, NeurIPS '22. Curran Associates, Inc., 2022.

\bibitem[BKSW19]{BunKSW19}
Mark Bun, Gautam Kamath, Thomas Steinke, and Zhiwei~Steven Wu.
\newblock Private hypothesis selection.
\newblock In {\em Advances in Neural Information Processing Systems 32}, NeurIPS '19, pages 156--167. Curran Associates, Inc., 2019.

\bibitem[BKZ23]{BieKZ23}
Alex Bie, Gautam Kamath, and Guojun Zhang.
\newblock Private {GAN}s, revisited.
\newblock {\em Transactions on Machine Learning Research}, 2023.

\bibitem[BS19]{BunS19}
Mark Bun and Thomas Steinke.
\newblock Average-case averages: Private algorithms for smooth sensitivity and mean estimation.
\newblock In {\em Advances in Neural Information Processing Systems 32}, NeurIPS '19, pages 181--191. Curran Associates, Inc., 2019.

\bibitem[BTGT18]{BassilyTT18}
Raef Bassily, Om~Thakkar, and Abhradeep Guha~Thakurta.
\newblock Model-agnostic private learning.
\newblock In {\em Advances in Neural Information Processing Systems 31}, NeurIPS '18, pages 7102--7112. Curran Associates, Inc., 2018.

\bibitem[BWW{\etalchar{+}}19]{BeaulieuJonesWWLBBG19}
Brett~K {Beaulieu-Jones}, Zhiwei~Steven Wu, Chris Williams, Ran Lee, Sanjeev~P Bhavnani, James~Brian Byrd, and Casey~S Greene.
\newblock Privacy-preserving generative deep neural networks support clinical data sharing.
\newblock {\em Circulation: Cardiovascular Quality and Outcomes}, 12(7):e005122, 2019.

\bibitem[CBV{\etalchar{+}}21]{CaoBVFK21}
Tianshi Cao, Alex Bie, Arash Vahdat, Sanja Fidler, and Karsten Kreis.
\newblock Don't generate me: Training differentially private generative models with sinkhorn divergence.
\newblock In {\em Advances in Neural Information Processing Systems 34}, NeurIPS '21, pages 12480--12492. Curran Associates, Inc., 2021.

\bibitem[CN25]{CheuN25}
Albert Cheu and Debanuj Nayak.
\newblock Differentially private multi-sampling from distributions.
\newblock In {\em Proceedings of the 36th International Conference on Algorithmic Learning Theory}, ALT '25, pages 289--314. JMLR, Inc., 2025.

\bibitem[CWZ21]{CaiWZ21}
T~Tony Cai, Yichen Wang, and Linjun Zhang.
\newblock The cost of privacy: Optimal rates of convergence for parameter estimation with differential privacy.
\newblock {\em The Annals of Statistics}, 49(5):2825--2850, 2021.

\bibitem[DCVK23]{DockhornCVK23}
Tim Dockhorn, Tianshi Cao, Arash Vahdat, and Karsten Kreis.
\newblock Differentially private diffusion models.
\newblock {\em Transactions on Machine Learning Research}, 2023.

\bibitem[DF18]{DworkF18}
Cynthia Dwork and Vitaly Feldman.
\newblock Privacy-preserving prediction.
\newblock In {\em Proceedings of the 31st Annual Conference on Learning Theory}, COLT '18, pages 1693--1702, 2018.

\bibitem[DFM{\etalchar{+}}20]{DuFMBG20}
Wenxin Du, Canyon Foot, Monica Moniot, Andrew Bray, and Adam Groce.
\newblock Differentially private confidence intervals.
\newblock {\em arXiv preprint arXiv:2001.02285}, 2020.

\bibitem[DL09]{DworkL09}
Cynthia Dwork and Jing Lei.
\newblock Differential privacy and robust statistics.
\newblock In {\em Proceedings of the 41st Annual ACM Symposium on the Theory of Computing}, STOC '09, pages 371--380. ACM, 2009.

\bibitem[DMNS06]{DworkMNS06}
Cynthia Dwork, Frank McSherry, Kobbi Nissim, and Adam Smith.
\newblock Calibrating noise to sensitivity in private data analysis.
\newblock In {\em Proceedings of the 3rd Conference on Theory of Cryptography}, TCC '06, pages 265--284, Berlin, Heidelberg, 2006. Springer.

\bibitem[GHKM23]{GhaziHKM23}
Badih Ghazi, Xiao Hu, Ravi Kumar, and Pasin Manurangsi.
\newblock On differentially private sampling from gaussian and product distributions.
\newblock In {\em Advances in Neural Information Processing Systems 36}, NeurIPS '23, pages 77783--77809. Curran Associates, Inc., 2023.

\bibitem[HJSP23]{HarderJSP23}
Frederik Harder, Milad Jalali, Danica~J Sutherland, and Mijung Park.
\newblock Pre-trained perceptual features improve differentially private image generation.
\newblock {\em Transactions on Machine Learning Research}, 2023.

\bibitem[HKM22]{HopkinsKM22}
Samuel~B Hopkins, Gautam Kamath, and Mahbod Majid.
\newblock Efficient mean estimation with pure differential privacy via a sum-of-squares exponential mechanism.
\newblock In {\em Proceedings of the 54th Annual ACM Symposium on the Theory of Computing}, STOC '22, pages 1406--1417. ACM, 2022.

\bibitem[HKMN23]{HopkinsKMN23}
Samuel~B Hopkins, Gautam Kamath, Mahbod Majid, and Shyam Narayanan.
\newblock Robustness implies privacy in statistical estimation.
\newblock In {\em Proceedings of the 55th Annual ACM Symposium on the Theory of Computing}, STOC '23. ACM, 2023.

\bibitem[HLY21]{HuangLY21}
Ziyue Huang, Yuting Liang, and Ke~Yi.
\newblock Instance-optimal mean estimation under differential privacy.
\newblock In {\em Advances in Neural Information Processing Systems 34}, NeurIPS '21. Curran Associates, Inc., 2021.

\bibitem[HT10]{HardtT10}
Moritz Hardt and Kunal Talwar.
\newblock On the geometry of differential privacy.
\newblock In {\em Proceedings of the 42nd Annual ACM Symposium on the Theory of Computing}, STOC '10, pages 705--714. ACM, 2010.

\bibitem[KDH23]{KuditipudiDH23}
Rohith Kuditipudi, John Duchi, and Saminul Haque.
\newblock A pretty fast algorithm for adaptive private mean estimation.
\newblock In {\em Proceedings of the 36th Annual Conference on Learning Theory}, COLT '23, pages 2511--2551, 2023.

\bibitem[KLSU19]{KamathLSU19}
Gautam Kamath, Jerry Li, Vikrant Singhal, and Jonathan Ullman.
\newblock Privately learning high-dimensional distributions.
\newblock In {\em Proceedings of the 32nd Annual Conference on Learning Theory}, COLT '19, pages 1853--1902, 2019.

\bibitem[KMR{\etalchar{+}}23]{KamathMRSSU23}
Gautam Kamath, Argyris Mouzakis, Matthew Regehr, Vikrant Singhal, Thomas Steinke, and Jonathan Ullman.
\newblock A bias-variance-privacy trilemma for statistical estimation.
\newblock {\em arXiv preprint arXiv:2301.13334}, 2023.

\bibitem[KMS22a]{KamathMS22}
Gautam Kamath, Argyris Mouzakis, and Vikrant Singhal.
\newblock New lower bounds for private estimation and a generalized fingerprinting lemma.
\newblock In {\em Advances in Neural Information Processing Systems 35}, NeurIPS '22, pages 24405--24418. Curran Associates, Inc., 2022.

\bibitem[KMS{\etalchar{+}}22b]{KamathMSSU22}
Gautam Kamath, Argyris Mouzakis, Vikrant Singhal, Thomas Steinke, and Jonathan Ullman.
\newblock A private and computationally-efficient estimator for unbounded gaussians.
\newblock In {\em Proceedings of the 35th Annual Conference on Learning Theory}, COLT '22, pages 544--572, 2022.

\bibitem[KMV22]{KothariMV22}
Pravesh~K Kothari, Pasin Manurangsi, and Ameya Velingker.
\newblock Private robust estimation by stabilizing convex relaxations.
\newblock In {\em Proceedings of the 35th Annual Conference on Learning Theory}, COLT '22, pages 723--777, 2022.

\bibitem[KV18]{KarwaV18}
Vishesh Karwa and Salil Vadhan.
\newblock Finite sample differentially private confidence intervals.
\newblock In {\em Proceedings of the 9th Conference on Innovations in Theoretical Computer Science}, ITCS '18, pages 44:1--44:9, Dagstuhl, Germany, 2018. Schloss Dagstuhl--Leibniz-Zentrum fuer Informatik.

\bibitem[LGK{\etalchar{+}}24]{LinGKNY24}
Zinan Lin, Sivakanth Gopi, Janardhan Kulkarni, Harsha Nori, and Sergey Yekhanin.
\newblock Differentially private synthetic data via foundation model apis 1: Images.
\newblock In {\em Proceedings of the 12th International Conference on Learning Representations}, ICLR '24, 2024.

\bibitem[LKKO21]{LiuKKO21}
Xiyang Liu, Weihao Kong, Sham Kakade, and Sewoong Oh.
\newblock Robust and differentially private mean estimation.
\newblock In {\em Advances in Neural Information Processing Systems 34}, NeurIPS '21. Curran Associates, Inc., 2021.

\bibitem[LKO22]{LiuKO22}
Xiyang Liu, Weihao Kong, and Sewoong Oh.
\newblock Differential privacy and robust statistics in high dimensions.
\newblock In {\em Proceedings of the 35th Annual Conference on Learning Theory}, COLT '22, pages 1167--1246, 2022.

\bibitem[LM00]{LaurentM00}
Beatrice Laurent and Pascal Massart.
\newblock Adaptive estimation of a quadratic functional by model selection.
\newblock {\em The Annals of Statistics}, 28(5):1302--1338, 2000.

\bibitem[Nar23]{Narayanan23}
Shyam Narayanan.
\newblock Better and simpler lower bounds for differentially private statistical estimation.
\newblock {\em arXiv preprint arXiv:2310.06289}, 2023.

\bibitem[PAE{\etalchar{+}}17]{PapernotAEGT17}
Nicolas Papernot, Mart{\'\i}n Abadi, Ulfar Erlingsson, Ian Goodfellow, and Kunal Talwar.
\newblock Semi-supervised knowledge transfer for deep learning from private training data.
\newblock In {\em Proceedings of the 5th International Conference on Learning Representations}, ICLR '17, 2017.

\bibitem[PH24]{PortellaH24}
Victor~S Portella and Nick Harvey.
\newblock Lower bounds for private estimation of gaussian covariance matrices under all reasonable parameter regimes.
\newblock {\em arXiv preprint arXiv:2404.17714}, 2024.

\bibitem[Ran21]{Ranosova21}
Hedvika Ranosova.
\newblock Spherically {S}ymmetric {M}easures.
\newblock {\em Bachelor’s thesis, Department of Probability and Mathematical Statistics, Charles University}, 2021.

\bibitem[RSSS21]{RaskhodnikovaSSS21}
Sofya Raskhodnikova, Satchit Sivakumar, Adam Smith, and Marika Swanberg.
\newblock Differentially private sampling from distributions.
\newblock In {\em Advances in Neural Information Processing Systems 34}, NeurIPS '21, pages 28983--28994. Curran Associates, Inc., 2021.

\bibitem[Ska13]{Skala13}
Matthew Skala.
\newblock Hypergeometric tail inequalities: ending the insanity.
\newblock {\em arXiv preprint arXiv:1311.5939}, 2013.

\bibitem[SV16]{SasonV2016}
Igal Sason and Sergio Verdu.
\newblock $f$-divergence inequalities.
\newblock {\em IEEE Transactions on Information Theory}, 62(11):5973–6006, 2016.

\bibitem[XLB{\etalchar{+}}24]{XieLBGYINJZLLY24}
Chulin Xie, Zinan Lin, Arturs Backurs, Sivakanth Gopi, Da~Yu, Huseyin~A Inan, Harsha Nori, Haotian Jiang, Huishuai Zhang, Yin~Tat Lee, Bo~Li, and Sergey Yekhanin.
\newblock Differentially private synthetic data via foundation model apis 2: Text.
\newblock {\em arXiv preprint arXiv:2403.01749}, 2024.

\bibitem[XLW{\etalchar{+}}18]{XieLWWZ18}
Liyang Xie, Kaixiang Lin, Shu Wang, Fei Wang, and Jiayu Zhou.
\newblock Differentially private generative adversarial network.
\newblock {\em arXiv preprint arXiv:1802.06739}, 2018.

\bibitem[ZZ20]{ZhangZ20}
Anru Zhang and Yuchen Zhou.
\newblock On the non-asymptotic and sharp lower tail bounds of random variables.
\newblock {\em Stat}, 9(1):e314, 2020.

\end{thebibliography}

\appendix
\renewcommand{\thesection}{\Alph{section}}

\section{Facts from Linear Algebra}
\label{sec:lin_alg_facts}

In this appendix, we state a number of useful facts from linear algebra.
We start with a fact about the semi-definite ordering.

\begin{fact}
\label{fact:sd_ord}
Let $A, B \in \bR^{n \times n}$ be symmetric positive-definite matrices.
Then, $A \preceq B \iff \id \preceq A^{- \frac{1}{2}} B A^{- \frac{1}{2}} \iff B^{- \frac{1}{2}} A B^{- \frac{1}{2}} \preceq \id$.
\end{fact}

Next, we have a fact about matrix norms.
Specifically, we recall the property known as \emph{rotational invariance}.

\begin{fact}
\label{fact:rotational_invariance}
Let $A \in \bR^{n \times m}$.
For any rotation matrix $U \in \bR^{n \times n}$, we have $\llnorm{U A} = \llnorm{A}, \fnorm{U A} = \fnorm{A}$, and $\mnorm{\tr}{U A} = \mnorm{\tr}{A}$.
The above also holds for $A U$, for any rotation matrix $U \in \bR^{m \times m}$.
\end{fact}

We now introduce basic matrix factorizations, namely the \emph{spectral decomposition} for symmetric matrices, and the \emph{singular value decomposition (SVD)} for general matrices.

\begin{fact}
\label{fact:matrix_fact}
Let $A \in \bR^{n \times m}$.
Then, $A$ can be factorized as follows:
\begin{itemize}
    \item If $n = m$ and $A = A^{\top}$, we can write $A = S \Lambda S^{\top}$, where $\Lambda \coloneqq \diag\brc{\lambda_1, \dots, \lambda_n}$ and $\brc{\lambda_i}_{i \in \brk{n}}$ are the eigenvalues of $A$, and $S S^{\top} = \id$.
    For the special case of PSD matrices, $S$ can be assumed to be a rotation matrix.
    The previous can equivalently be written as $A = \sum\limits_{i \in \brk{n}} \lambda_i s_i s_i^{\top}$, where $s_i \in \bR^n$ are the vectors representing the columns of $S$.
    
    \item If $n \geq m$, we can write $A = U D V^{\top}$, where $D \coloneqq \paren{\diag\brc{\sigma_1, \dots, \sigma_m} \middle| 0_{m \times \paren{n - m}}}^{\top} \in \bR^{n \times m}$ and $\brc{\sigma_i}_{i \in \brk{m}}$ are the singular values of $A$, and $U \in \bR^{n \times n}, V \in \bR^{m \times m}$ are rotation matrices.
    The previous can equivalently be written as $A = \sum\limits_{i \in \brk{m}} \sigma_i u_i v_i^{\top}$, where $u_i \in \bR^n$ and $v_i \in \bR^m$ are the vectors representing the columns of $U$ and $V$, respectively.
\end{itemize}
\end{fact}

\begin{remark}
\label{rem:svd_comment}
The statement given above for the SVD is phrased in terms of matrices whose number of rows is lower-bounded by the number of columns.
However, the decomposition is valid even without this assumption and, given a matrix that does not satisfy this, we can obtain its SVD by taking its transpose, and then applying the previous statement.
\end{remark}

We include one standard inequality here that upper-bounds the determinant of a positive-definite matrix in terms of its trace.
The inequality follows directly from AM-GM.

\begin{fact}
\label{fact:pd_am_gam}
Let $A \in \bR^{n \times n}$ be a positive-definite matrix.
Then, we have $\det\paren{A} \le \paren{\frac{\tr\paren{A}}{n}}^n$.
\end{fact}

Finally, we have a lemma which establishes that, if a symmetric positive-definite matrix $A$ has small distance from the identity matrix, the same holds for $A^{- 1}$ as well.

\begin{lemma}
\label{lem:dist_inverse}
Let $A \in \bR^{n \times n}$ be a symmetric positive-definite matrix with $A \succeq \id$.
Then, it holds that $\mnorm{\tr}{A - \id} \le \alpha \implies \mnorm{\tr}{A^{- 1} - \id} \le \alpha$,
\end{lemma}

\begin{proof}
We only prove the first case, because the proof for the second is analogous.
Since the matrix$A$ is symmetric and positive-definite, by Fact~\ref{fact:matrix_fact}, it can be written in the form $A = U \Lambda U^{\top}$, where $\Lambda$ is the eigenvalue matrix of $A$, and $U$ is a rotation matrix.
This implies that we have $A^{- 1} = U \Lambda^{- 1} U^{\top}$.
Thus, by Fact~\ref{fact:rotational_invariance}, we get $\mnorm{\tr}{A - \id} = \mnorm{\tr}{\Lambda - \id}$ and $\mnorm{\tr}{A^{- 1} - \id} = \mnorm{\tr}{\Lambda^{- 1} - \id}$.
Let $\brc{\lambda_i}_{i \in \brk{n}}$ be the spectrum of $A$.
By the previous and our assumptions that $A \succeq \id$ and $\mnorm{\tr}{A - \id} \le \alpha$, we can reduce the problem of upper-bounding $\mnorm{\tr}{A^{- 1} - \id}$ to the following constrained maximization problem:
\begin{align*}
    \max_{\lambda_i} \quad & \sum\limits_{i \in \brk{n}} \paren{1 - \frac{1}{\lambda_i}} \\
       \textrm{s.t.} \quad & \sum\limits_{i \in \brk{n}} \paren{\lambda_i - 1} \le a < \frac{1}{2} \\
                     \quad & \lambda_i \geq 1, \forall i \in \brk{n},
\end{align*}
The above can be solved exactly via an application of the KKT conditions, yielding the bound $\fnorm{A^{- 1} - \id} \le \frac{\alpha n}{n + \alpha} \le \alpha$.
\end{proof}

\section{Facts from Probability \& Statistics}
\label{sec:prob_facts}

In this appendix, we expand upon the background from probability and statistics that was given in Section~\ref{subsec:math_prelim}.
We start with a number of results about the densities and concentration properties of distributions of interest.
These results are also used at various points in~\cite{BrownHS23} and~\cite{GhaziHKM23}.

\begin{fact}[See~\cite{Skala13}]
\label{fact:hyp_geom_conc}
Suppose an urn contains $N$ balls and exactly $k$ of them are black.
Let random variable $y$ be the number of black balls selected when drawing $n$ uniformly at random from the urn without replacement.
Then, for all $t \geq 0$, we have:
\[
    \pr{}{\abs{\frac{y}{n} - \frac{k}{N}} \geq t} \le 2 e^{- 2 t^2 n}.
\]
\end{fact}

\begin{fact}[Lemma 1 of~\cite{LaurentM00}]
\label{fact:hd_gaussian_tail}
If $X \sim \chi^2\paren{k}$, and $\beta \in \brk{0, 1}$, then:
\[
    \pr{}{X - k \geq 2 \sqrt{k \ln\paren{\frac{1}{\beta}}} + 2 \ln\paren{\frac{1}{\beta}}} \le \beta.
\]
Equivalently, the above can be written as:
\[
    \pr{}{X \geq t} \le e^{- \frac{\paren{\sqrt{2 t - k} - \sqrt{k}}^2}{4}}, \forall t \geq k.
\]
Thus, if $Y \in \bR^d$ and $Y \sim \cN\paren{0, \id}$, then $\pr{}{\llnorm{Y}^2 \geq t} \le e^{- \frac{\paren{\sqrt{2 t - d} - \sqrt{d}}^2}{4}}, \forall t \geq d$.
\end{fact}

\begin{fact}[Hanson-Wright Inequality]
\label{fact:hanson_wright}
Let $X \in \bR^d$ be a random vector with $X \sim \cN\paren{0, \id}$.
There exists an absolute constant $\useconstant{chs} > 0$ such that, for any non-zero symmetric matrix $A \in \bR^{d \times d}$:\newconstant{chs}
\[
    \pr{}{\abs{X^{\top} A X - \ex{}{X^{\top} A X}} \geq t} \le 2 \exp\paren{- \useconstant{chs} \min\brc{\frac{t^2}{\fnorm{A}^2}, \frac{t}{\llnorm{A}}}}, \forall t \geq 0,
\]
where $\ex{}{X^{\top} A X} = \tr\paren{A}$.
\end{fact}

\begin{fact}[Folklore - see Fact $3.4$ from~\cite{KamathLSU19}]
\label{fact:gaussian_spec_conc}
Let $X \coloneqq \paren{X_1, \dots, X_n} \sim \cN\paren{0, \id}^{\otimes n}$ and $\widehat{\Sigma} \coloneqq \frac{1}{n} \sum\limits_{i \in \brk{n}} X_i X_i^{\top}$.
Then, except with probability $\beta$, we have:
\[
    \paren{1 - \cO\paren{\sqrt{\frac{d + \log\paren{\frac{1}{\beta}}}{n}}}} \id \preceq \widehat{\Sigma} \preceq \paren{1 + \cO\paren{\sqrt{\frac{d + \log\paren{\frac{1}{\beta}}}{n}}}} \id.
\]
\end{fact}

\begin{fact}[Theorem $2$ from~\cite{Ranosova21}]
\label{fact:unit_sphere_proj_dist}
Let $v \sim \cU\paren{\bS^{n - 1}}$.
Then, the density of $v_{\le i} \sim \cU\paren{\bS^{n - 1}}_{\le i}$ is:
\[
    f\paren{x} \ \propto \ 
    \left\{
    \begin{array}{ll}
         \paren{1 - \llnorm{x}^2}^{\frac{n - i}{2} - 1}, & \llnorm{x} \le 1 \\
         0,                                              & \llnorm{x} > 1
    \end{array}
    \right..
\]
Additionally, $\llnorm{v_{\le i}}^2 \sim \Beta\paren{\frac{i}{2}, \frac{n - i}{2}}$.
\end{fact}

\begin{fact}[Theorem $8$ from~\cite{ZhangZ20}]
\label{fact:beta_concentration}
Let $X \sim \Beta\paren{\alpha, \beta}$ with $0 < \alpha < \beta$.
Then, there exists an absolute constant $\useconstant{cbeta} \in \paren{0, 1}$ such that:$\newconstant{cbeta}$
\[
    \pr{}{X \geq \frac{\alpha}{\alpha + \beta} + x} \le 2 e^{- \useconstant{cbeta} \min\brc{\frac{\beta^2 x^2}{\alpha}, \beta x}}, \forall x \geq 0.
\]
\end{fact}

We will also need a standard result about transformations of random vectors.

\begin{fact}
\label{fact:transformations}
Let $X \in \bR^d$ be a random vector with density $f_X$.
Also, for an invertible function $g \colon \bR^d \to \bR^d$, let $Y \coloneqq g\paren{X}$.
Then, we have for the density of $Y$:
\[
    f_Y\paren{y} = f_X\paren{g^{- 1}\paren{y}} \abs{\det\paren{J_{g^{- 1}}}},
\]
where $J_{g^{- 1}}$ denotes the Jacobian matrix of the inverse transform $g^{- 1}$.
\end{fact}

As a direct consequence of the above and Fact~\ref{fact:rotational_invariance}, we get the following result, which establishes that the uniform distribution over the unit sphere is \emph{rotationally invariant}.

\begin{fact}
\label{fact:uniform_rot_inv}
Let $v \sim \cU\paren{\bS^{n - 1}}$.
Then, for any rotation matrix $U \in \bR^{n \times n}$, we have $U v \overset{d}{=} v$.
\end{fact}

We now present a result that describes one method for generating random vectors that are uniformly distributed over the unit sphere.

\begin{fact}
\label{fact:random_unit_vector_gen}
Let $X \coloneqq \paren{X_1, \dots, X_d} \sim \cN\paren{0, 1}^{\otimes d}$.
We have $\frac{X}{\llnorm{X}} \sim \cU\paren{\bS^{d - 1}}$.
\end{fact}

We conclude our facts about distributions with a lemma that generalizes the stability property of the Gaussian distribution.
Namely, stability says that any sum of independent Gaussians or product of a Gaussian with a number also follows a Gaussian distribution.
We extend this to random linear combinations when the weights come from a distribution over the unit sphere.

\begin{lemma} 
\label{lem:stability_gaussian}
Let $\cD \in \Delta\paren{\bS^{n - 1}}$, and $a \sim \cD$.
Then, given $X \coloneqq \paren{X_1, \dots, X_n} \sim \cN\paren{0, \id}^{\otimes n}$ that is independent of $a$, we have $\sum\limits_{i \in \brk{n}} a_i X_i \overset{d}{=} \cN\paren{0, \id}$.
\end{lemma}

\begin{proof}
Let $X \coloneqq \sum\limits_{i \in \brk{n}} a_i \cN\paren{0, \id}$.
Let $a_0$ be a fixed realization of $a$.
Observe that:
\[
    \paren{X \middle| a = a_0} = \sum\limits_{i \in \brk{n}} a_{0, i} \cN\paren{0, \id} \overset{d}{=} \sum\limits_{i \in \brk{n}} \cN\paren{0, a_{0, i}^2 \id} \overset{d}{=} \cN\paren{0, \paren{\sum\limits_{i \in \brk{n}} a_{0, i}^2} \id} \overset{d}{=} \cN\paren{0, \id}.
\]
Now, the density of $X$ is a weighted average of the densities of $\paren{X \middle| \alpha}$ for all fixed realizations $a_0$ of $a$, where the weighting is performed according to the density of $\cD$.
Thus, we have:
\[
    f_X\paren{x} = \rmint\limits_{\supp\paren{\cD}} f_{\cD}\paren{a_0} f_{\paren{X \middle| a = a_0}}\paren{x} \, da_0 = \rmint\limits_{\supp\paren{\cD}} f_{\cD}\paren{a_0} \frac{1}{\sqrt{2 \pi}} e^{- \frac{x^2}{2}} \, da_0 = \frac{1}{\sqrt{2 \pi}} e^{- \frac{x^2}{2}}, \forall x \in \bR,
\]
which yields $X \overset{d}{=} \cN\paren{0, 1}$, as desired.
\end{proof}

Having presented all the necessary statements about distributions, we continue with a number of properties of $f$-divergences, namely the \emph{Data-Processing Inequality (DPI)} and \emph{joint convexity}.

\begin{fact}
\label{fact:dpi}
Let $P, Q \in \Delta\paren{\cX}$ and $X \sim P, Y \sim Q$, and let $f$ be a function satisfying the conditions of Definition~\ref{def:f_div}.
For any function $g \colon \cX \to \cY$ (deterministic or randomized), we get:
\[
    D_f\paren{g\paren{X} \middle\| g\paren{Y}} \le D_f\paren{X \middle\| Y}.
\]
If $g$ is invertible, the above holds as an equality.
\end{fact}

\begin{fact}
\label{fact:joint_convexity}
Let $P, Q \in \Delta\paren{\cX}$ and $X \sim P, Y \sim Q$, and let $f$ be a function satisfying the conditions of Definition~\ref{def:f_div}.
We assume that there exists a $\lambda \in \paren{0, 1}$ and distributions $P_1, P_2 \in \Delta\paren{\cX}$ and $Q_1, Q_2 \in \Delta\paren{\cX}$ such that $P = \lambda P_1 + \paren{1 - \lambda} P_2$ and $Q = \lambda Q_1 + \paren{1 - \lambda} Q_2$, i.e., $P$ and $Q$ are mixtures of $\paren{P_1, P_2}$ and $\paren{Q_1, Q_2}$, respectively, with shared mixing weights $\paren{\lambda, 1 - \lambda}$.
Then, we have:
\[
    D_f\paren{P \middle\| Q} \le \lambda D_f\paren{P_1 \middle\| Q_1} + \paren{1 - \lambda} D_f\paren{P_2 \middle\| Q_2}.
\]
\end{fact}

We now introduce a property which is specific to the Hockey-Stick divergence.
As implied by Definition~\ref{def:hs_div}, $D_{e^{\eps}}\paren{\cdot \middle\| \cdot}$ does not satisfy a triangle inequality in the general case (and is thus not a distance metric).
However, the following property that we will state can be interpreted as a weak form of the triangle inequality.
We note that the property has appeared in the literature before (see Inequality ($404$) in~\cite{SasonV2016}), but it is stated under the assumption that we have distributions $P, R, Q \in \Delta\paren{\cX}$ with $\supp\paren{P} \subseteq \supp\paren{R} \subseteq \supp\paren{Q}$ (presumably because~\cite{SasonV2016} only mentions $\max\brc{x - e^{\eps}, 0}$ as a generator of $D_{e^{\eps}}$).
This assumption is not actually necessary, which is why we include a proof of the claim.
Our proof is elementary and relies on the equivalent form of the Hockey-Stick divergence that uses $\frac{1}{2} \abs{x - e^{\eps}} - \frac{1}{2} \paren{e^{\eps} - 1}$ as a generator.

\begin{lemma}
\label{lem:hs_div_triangle_ineq}
Let $P, R, Q \in \Delta\paren{\cX}$ with respective probability density functions $p, r, q$.
Then, for any $\eps_1, \eps_2 \geq 0$, we have:
\[
    D_{e^{\eps_1 + \eps_2}}\paren{P \middle\| Q} \le D_{e^{\eps_1}}\paren{P \middle\| R} + e^{\eps_1} D_{e^{\eps_2}}\paren{R \middle\| Q}.
\]
\end{lemma}

\begin{proof}
By Definition~\ref{def:hs_div} and the triangle inequality we get:
\begin{align*}
    D_{e^{\eps_1 + \eps_2}}\paren{P \middle\| Q}
    &= \frac{1}{2} \rmint\limits_{\cX} \abs{p\paren{x} - e^{\eps_1 + \eps_2} q\paren{x}} \, dx - \frac{1}{2} \paren{e^{\eps_1 + \eps_2} - 1} \\
    &\le \frac{1}{2} \rmint\limits_{\cX} \paren{\abs{p\paren{x} - e^{\eps_1} r\paren{x}} + \abs{e^{\eps_1} r\paren{x} -  e^{\eps_1 + \eps_2} q\paren{x}}} \, dx - \frac{1}{2} \paren{e^{\eps_1 + \eps_2} - 1} \\
    &= \brk{\frac{1}{2} \rmint\limits_{\cX} \abs{p\paren{x} - e^{\eps_1} r\paren{x}} \, dx - \frac{1}{2} \paren{e^{\eps_1} - 1}} + e^{\eps_1} \rmint\limits_{\cX} \abs{r\paren{x} -  e^{\eps_2} q\paren{x}} \, dx \\
    &\quad\ + \frac{1}{2} \paren{e^{\eps_1} - 1} - \frac{1}{2} \paren{e^{\eps_1 + \eps_2} - 1} \\
    &= D_{e^{\eps_1}}\paren{P \middle\| R} + e^{\eps_1} \brk{\rmint\limits_{\cX} \abs{r\paren{x} - e^{\eps_2} q\paren{x}} \, dx - \frac{1}{2} \paren{e^{\eps_2} - 1}} \\
    &= D_{e^{\eps_1}}\paren{P \middle\| R} + e^{\eps_1} D_{e^{\eps_2}}\paren{R \middle\| Q}.
\end{align*}
\end{proof}

\section{Details of Algorithm~\ref{alg:bhs23}}
\label{sec:bhs23_details}

In this appendix, we include all the details related to the subroutines of Algorithm~\ref{alg:bhs23} which we omitted from the main body.

We start by giving the pseudocode for \StableCov\ (Algorithm~\ref{alg:stablecov}).
As we have discussed in Section~\ref{sec:bhs23}, the algorithm's function is to perform a soft outlier removal process, which will result in a sequence of weights.
These can be used to construct a weighted version of the empirical covariance which, in addition to being accurate, also has stability guarantees which prove to be useful in guaranteeing privacy.
The algorithm makes crucial of use of the subroutine \LargestGS\ (Algorithm~\ref{alg:largest_gs}), which is given separately.

\begin{algorithm}[H]
    \caption{Stable Covariance}\label{alg:stablecov}
    \hspace*{\algorithmicindent} \textbf{Input:} Dataset $X \coloneqq \paren{X_1, \dots, X_n}^{\top} \in \bR^{n \times d}$; outlier threshold $\lambda_0 \geq 1$; discretization parameter $k \in \N$. \\
    \hspace*{\algorithmicindent} \textbf{Output:} $W \in \bR^{d \times n}, \Score \in \N$.
    \begin{algorithmic}[1]
    \Procedure{\StableCov$_{\lambda_0, k}$}{$X$}
        \State Let $m \gets \fl{\frac{n}{2}}$.
        \For{$i \in \brk{m}$}\Comment{Pair and rescale.}
            \State Let $Y_i \gets \frac{1}{\sqrt{2}} \paren{X_i - X_{i + m}}$.
        \EndFor
        \For{$\ell \in \brc{0, \dots, 2 k}$}
            \State Let $S_{\ell} \gets \LargestGS_{e^{\frac{\ell}{k}} \lambda_0}\paren{Y}$.\Comment{Algorithm~\ref{alg:largest_gs}}
        \EndFor
        \State Let $\Score \gets \min\brc{k, \min\limits_{0 \le \ell \le k}\brc{m - \abs{S_{\ell}} + \ell}}$.
        \For{$i \in \brk{m}$}
            \State Let $w_i \gets \frac{1}{k m} \sum\limits_{\ell = k + 1}^{2 k} \mathds{1}\brc{i \in S_{\ell}}$.
        \EndFor
        \State Let $W \gets \paren{\sqrt{w_1} Y_1, \dots, \sqrt{w_m} Y_m}$.
        \State \Return $\paren{W, \Score}$.
    \EndProcedure
    \end{algorithmic}
\end{algorithm}

We now give the pseudocode for \LargestGS.
The algorithm starts by constructing the standard empirical covariance of the whole input dataset $Y$.
Then, the algorithm calculates the Mahalanobis norm of all the points $Y_i$ measured with respect to the empirical covariance.
Due to the fact that Algorithm~\ref{alg:stablecov} starts by performing a centering operation to the dataset, the points that have large Mahalanobis norm are the outliers.
These are removed, the empirical covariance of the remaining points is calculated, and the previous check is repeated with those points.
The process is repeated iteratively until a fixpoint is reached, i.e., all outliers have been removed.

\begin{algorithm}[H]
    \caption{Largest Good Subset}\label{alg:largest_gs}
    \hspace*{\algorithmicindent} \textbf{Input:} Dataset $Y \coloneqq \paren{Y_1, \dots, Y_m}^{\top} \in \bR^{m \times d}$; outlier threshold $\lambda$. \\
    \hspace*{\algorithmicindent} \textbf{Output:} $S \subseteq \brk{m}$.
    \begin{algorithmic}[1]
    \Procedure{\LargestGS$_{\lambda}$}{$Y$}
        \State Let $S \gets \brk{m}$.
        \Repeat
            \State Let $A \gets \frac{1}{m} \sum\limits_{j \in S} Y_i Y_i^{\top}$.
            \State Let $\Out \gets \brc{i \in S \colon \llnorm{A^{- \frac{1}{2}} Y_i}^2 > \lambda}$.\Comment{For $A$: singular, define $\llnorm{A^{- \frac{1}{2}} v} = \infty, \forall v \in \bR^d$.}
            \State $S \gets S \setminus \Out$
        \Until
    \EndProcedure
    \end{algorithmic}
\end{algorithm}

We now give the pseudocode for \StableMean\ (Algorithm~\ref{alg:stablemean}).
The algorithm works similarly to \StableCov, but has some minor differences which we highlight here.
First, the algorithm receives the output of \StableCov\ as input, which it uses to rescale the points and measure their pair-wise distances accordingly.
The goal of this is to result in a dataset that comes from a distribution that is approximately isotropic.
Second, the algorithm also receives a reference set $R$.
This is a subset of the dataset that the algorithm uses to identify outliers.
Instead of computing all pair-wise distances between points in the dataset, the algorithm will identify outliers by only considering distances of elements of the dataset from points in the set reference set.
The purpose of this aspect of the algorithm is to ensure computational efficiency.
Similarly to \StableCov, this process leverages a subroutine which is given separately (\LargestCore\ - Algorithm~\ref{alg:largest_core}).

\begin{algorithm}[H]
    \caption{Stable Mean}\label{alg:stablemean}
    \hspace*{\algorithmicindent} \textbf{Input:} Dataset $X \coloneqq \paren{X_1, \dots, X_n}^{\top} \in \bR^{n \times d}$; covariance $\widehat{\Sigma}$; outlier threshold $\lambda_0 \geq 1$; discretization parameter $k \in \N$; reference set $R \subseteq \brk{n}$. \\
    \hspace*{\algorithmicindent} \textbf{Output:} $v \in \bR^d, \Score \in \N$.
    \begin{algorithmic}[1]
    \Procedure{\StableMean$_{\widehat{\Sigma}, \lambda_0, k, R}$}{$X$}
        \For{$\ell \in \brc{0, \dots, 2 k}$}
            \State Let $S_{\ell} \gets \LargestCore_{\widehat{\Sigma}, \lambda, \abs{R} - \ell, R}\paren{X}$.
        \EndFor
        \State Let $\Score \gets \min\brc{k, \min\limits_{0 \le \ell \le k}\brc{n - \abs{S_{\ell}} + \ell}}$.
        \For{$i \in \brk{n}$}
            \State Let $c_i \gets \sum\limits_{\ell = k + 1}^{2 k} \mathds{1}\brc{i \in S_{\ell}}$.
        \EndFor
        \State Let $Z \gets \sum\limits_{i \in \brk{n}} c_i$.\Comment{Normalizing constant.}
        \For{$i \in \brk{n}$}
            \State Let $v_i \gets \frac{c_i}{Z}$.\Comment{Set $v_i \gets 0$ if $Z = 0$.}
        \EndFor
        \State \Return $\paren{v, \Score}$.
    \EndProcedure
    \end{algorithmic}
\end{algorithm}

We conclude by giving the pseudocode for \LargestCore.
The algorithm identifies the points in the dataset that have ``small'' Mahalanobis distance (measured with respect to the output $\widehat{\Sigma}$ of \StableCov) from at least $\tau$ points in the reference set $R$.

\begin{algorithm}[H]
    \caption{Largest Core}\label{alg:largest_core}
    \hspace*{\algorithmicindent} \textbf{Input:} Dataset $X \coloneqq \paren{X_1, \dots, X_m}^{\top} \in \bR^{n \times d}$; outlier threshold $\lambda$. \\
    \hspace*{\algorithmicindent} \textbf{Output:} $S \subseteq \brk{n}$.
    \begin{algorithmic}[1]
    \Procedure{\LargestCore$_{\widehat{\Sigma}, \lambda, \tau, R}$}{$X$}
        \For{$i \in \brk{n}$}
            \State $N_i \gets \brc{j \in R \colon \mnorm{\widehat{\Sigma}}{X_i - X_j}^2 \le \lambda}$.
        \EndFor
        \State Let $S \gets \brc{i \in \brk{n} \colon \abs{N_i} \geq \tau}$.
        \State \Return $S$.
    \EndProcedure
    \end{algorithmic}
\end{algorithm}

\section{Alternative Version of Proposition~\ref{prop:priv_anal_multi1}}
\label{sec:priv_proof_alt}

In this appendix, we present our original version and proof of Proposition~\ref{prop:priv_anal_multi1}.
We start by sketching a na\"ive approach, and explaining why it fails, thus motivating the proof that we will present here.

The steps up to Lemma~\ref{lem:priv_anal_multi1_suff_cond} are the same as in Section~\ref{subsubsec:priv_anal_multi1}.
Our goal from this point on is to identify a sufficient condition that implies:
\begin{equation}
    s^{\top} \paren{\widehat{\Sigma}^{\frac{1}{2}} \widehat{\Sigma}'^{- 1} \widehat{\Sigma}^{\frac{1}{2}} - \id} s \le \frac{\eps}{4 n_2}. \label{eq:goal}
\end{equation}
Observe that, by (\ref{eq:sd_order_constraint}), we have:
\begin{align}
    \paren{1 - \gamma} \id \preceq \widehat{\Sigma}^{\frac{1}{2}} \widehat{\Sigma}'^{- 1} \widehat{\Sigma}^{\frac{1}{2}} \preceq \frac{1}{1 - \gamma} \id
    &\iff - \gamma \id \preceq \widehat{\Sigma}^{\frac{1}{2}} \widehat{\Sigma}'^{- 1} \widehat{\Sigma}^{\frac{1}{2}} - \id \preceq \frac{\gamma}{1 - \gamma} \id \nonumber \\
    &\implies \llnorm{\widehat{\Sigma}^{\frac{1}{2}} \widehat{\Sigma}'^{- 1} \widehat{\Sigma}^{\frac{1}{2}} - \id} \le \frac{\gamma}{1 - \gamma}. \label{eq:stable_cov_ineq1}
\end{align}
Based on the above, a na\"ive approach involves observing that, by the definition of the spectral norm for symmetric matrices and the fact that $\gamma \coloneqq \frac{8 e^2 \lambda_0}{n_2} \le \frac{1}{4}$, we have:
\begin{equation}
    s^{\top} \paren{\widehat{\Sigma}^{\frac{1}{2}} \widehat{\Sigma}'^{- 1} \widehat{\Sigma}^{\frac{1}{2}} - \id} s \le \llnorm{\widehat{\Sigma}^{\frac{1}{2}} \widehat{\Sigma}'^{- 1} \widehat{\Sigma}^{\frac{1}{2}} - \id} \llnorm{s}^2 \le \frac{\gamma}{1 - \gamma} \llnorm{s}^2 \le \frac{4}{3} \gamma \llnorm{s}^2 = \frac{32 e^2 \lambda_0}{3 n_2} \llnorm{s}^2. \label{eq:quad_form_ub}
\end{equation}
Thus, a sufficient condition that would lead to (\ref{eq:goal}) being satisfied is $\llnorm{s}^2 \le \frac{3 \eps}{128 e^2 \lambda_0}$.
However, issues arise when upper-bounding the probability of this condition failing.
Specifically, Fact~\ref{fact:unit_sphere_proj_dist} yields $\llnorm{s}^2 \sim \Beta\paren{\frac{d}{2}, \frac{n_2 - d}{2}}$, so the natural next step here is to apply Fact~\ref{fact:beta_concentration}.
However, that requires $\frac{3 \eps}{128 e^2 \lambda_0} > \frac{d}{n_2} \iff n_2 \geq \frac{128 e^2 \lambda_0 d}{3 \eps}$.
Since $\lambda_0 \geq d$, this leads to a sample complexity bound that is sub-optimal in the dimension dependence by a quadratic factor, which is clearly prohibitive.

The failure of the na\"ive approach described above prompts us to re-examine (\ref{eq:quad_form_ub}).
Observe that the way we leveraged the definition of the spectral norm in the upper bound is not tight in general.
Indeed, for the inequality to be tight, it must be the case that $s$ is parallel to the eigenvector of $\widehat{\Sigma}^{\frac{1}{2}} \widehat{\Sigma}'^{- 1} \widehat{\Sigma}^{\frac{1}{2}} - \id$ that corresponds to the largest in absolute value eigenvalue.
This significantly restricts the realizations of $s$, which is a factor that was not taken into account in our analysis.

At this point, we recall the second consequence of Lemma~\ref{lem:stable_cov}.
We denote the spectrum of $\widehat{\Sigma}^{\frac{1}{2}} \widehat{\Sigma}'^{- 1} \widehat{\Sigma}^{\frac{1}{2}}$ by $\brc{\lambda_i}_{i \in \brk{d}}$.
Without loss of generality, we can assume that $\lambda_i \geq 1, \forall i \in \brk{d}$.
Indeed, having an eigenvalue that is $< 1$ would make it easier to satisfy (\ref{eq:goal}), since the corresponding term would contribute negatively to the LHS.
Then, (\ref{eq:eig_constraint}) yields:
\begin{equation}
    \sum\limits_{i \in \brk{d}} \paren{\lambda_i - 1} \le \gamma \paren{1 + 2 \gamma}. \label{eq:eig_constraint_restated}
\end{equation}
By our bound on $n_2$, we know that $1 + 2 \gamma = \Theta\paren{1}$, yielding $\sum\limits_{i \in \brk{d}} \paren{\lambda_i - 1} \le \Theta\paren{\gamma}$.
Thus, the above implies that not all terms $\lambda_i - 1$ can be $\approx \gamma$.
Two extreme cases would involve having either one large term $\lambda_i - 1 = \Theta\paren{\gamma}$ and all the other terms being negligible, or having most (if not all) terms being $\Theta\paren{\frac{\gamma}{d}}$.
We stress that using the definition of the spectral norm as we did in (\ref{eq:quad_form_ub}) is tight only in the latter case, the former requiring us to reason about the angle between $s$ and the eigenvector that corresponds to the eigenvalue for which $\lambda_i - 1 = \Theta\paren{\gamma}$.\footnote{We note that this observation is standard, and can be interpreted as a strong form of the pigeonhole principle.
That is, consider an example where we have numbers $x_1, \dots, x_n > 0$ such that $\sum\limits_{i \in \brk{n}} x_i = m$.
Then, the two extreme cases are that there must either be multiple numbers which are $\Omega\paren{\frac{m}{n}}$ or one number which is $\Omega\paren{m}$).
Recently, this trick has appeared in the context of the privacy literature in the proof of Lemma $4.9$ in~\cite{AgarwalKMMSU25}.}

Based on the above intuition, our proof will try to interpolate between the two previous regimes.
We ignore the smallest eigenvalues (these would contribute little to $s^{\top} \paren{\widehat{\Sigma}^{\frac{1}{2}} \widehat{\Sigma}'^{- 1} \widehat{\Sigma}^{\frac{1}{2}} - \id} s$ anyway), and partition the rest into sets that consist of eigenvalues whose magnitude is within a constant factor of each other.
Then, we get different conditions on different components of $s$, depending what the magnitude of the corresponding eigenvalue is.\footnote{By components of $s$ here we refer to its projections along the eigenbasis of $\widehat{\Sigma}^{\frac{1}{2}} \widehat{\Sigma}'^{- 1} \widehat{\Sigma}^{\frac{1}{2}}$.
Implicit here is a change of basis argument.}
This leads to a sample complexity of $\widetilde{\cO}\paren{\lambda_0}$ and no explicit dependence on $d$, thus resolving the previous issue.

We break the proof into two statements.
The first statement can be considered to be analogous to Lemma~\ref{lem:priv_anal_multi1_suff_cond}.
Indeed, we showed in Lemma~\ref{lem:priv_anal_multi1_suff_cond} that:
\[
    s^{\top} \widehat{\Sigma}^{\frac{1}{2}} \widehat{\Sigma}'^{- 1} \widehat{\Sigma}^{\frac{1}{2}} s \le \frac{1}{2} \text{ and } s^{\top} \paren{\widehat{\Sigma}^{\frac{1}{2}} \widehat{\Sigma}'^{- 1} \widehat{\Sigma}^{\frac{1}{2}} - \id} s \le \frac{\eps}{4 n_2} \implies \ln\paren{\frac{f_{z_{\le d}}\paren{t}}{f_T\paren{t}}} \le \frac{\eps}{2}.
\]
Now, we will identify a set of sufficient conditions on $s$ that imply $s^{\top} \paren{\widehat{\Sigma}^{\frac{1}{2}} \widehat{\Sigma}'^{- 1} \widehat{\Sigma}^{\frac{1}{2}} - \id} s \le \frac{\eps}{4 n_2}$.
In accordance with the sketch given above, we set $\ell_{\max} \coloneqq \cl{\log\paren{\frac{8 n_2 \gamma \paren{1 + 2 \gamma}}{\eps}}} = \cl{\log\paren{\frac{64 e^2 \lambda_0 \paren{1 + 2 \gamma}}{\eps}}}$ (recall that $\gamma \coloneqq \frac{8 e^2 \lambda_0}{n_2} \le \frac{\eps}{4} \le\frac{1}{4}$), and define the following partition of the indices $i \in \brk{d}$:
\begin{align*}
    \cB_{\ell}
    &\coloneqq \brc{i \in \brk{d} \colon \frac{\gamma \paren{1 + 2 \gamma}}{2^{\ell}} < \lambda_i - 1 \le \frac{\gamma \paren{1 + 2 \gamma}}{2^{\ell - 1}}}, \forall \ell \in \brk{\ell_{\max}}, \\
    \cB_0
    &\coloneqq \brk{d} \setminus \bigcup\limits_{\ell \in \brk{\ell_{\max}}} \cB_{\ell}.
\end{align*}
We now give the formal statement and proof of our lemma.

\begin{lemma}
\label{lem:priv_anal_multi1_suff_cond_alt}
Let $\brc{v_i}_{i \in \brk{d}}$ be an orthonormal set of eigenvectors that correspond to the spectrum $\brc{\lambda_i}_{i \in \brk{d}}$ of $\widehat{\Sigma}^{\frac{1}{2}} \widehat{\Sigma}'^{- 1} \widehat{\Sigma}^{\frac{1}{2}}$.
Then, for vectors $s \in \bR^d$ with $\llnorm{s} \le 1$, we have:
\[
    \sum\limits_{i \in \cB_{\ell}} \iprod{s, v_i}^2 \le \frac{2^{\ell - 1} \eps}{96 e^2 \lambda_0 \ell_{\max}}, \forall \ell \in \brk{\ell_{\max}} \implies s^{\top} \paren{\widehat{\Sigma}^{\frac{1}{2}} \widehat{\Sigma}'^{- 1} \widehat{\Sigma}^{\frac{1}{2}} - \id} s \le \frac{\eps}{4 n_2}.
\]
\end{lemma}

\begin{proof}
By Fact~\ref{fact:matrix_fact}, the target inequality can be written as:
\begin{equation}
    \sum\limits_{i \in \brk{d}} \paren{\lambda_i - 1} \iprod{v_i, s}^2 \le \frac{\eps}{4 n_2}. \label{eq:result_to_prove}
\end{equation}
We will leverage the assumption to upper-bound the LHS, and show that the resulting bound does not exceed the RHS.
To obtain the upper bound, we note that, for every $i \in \cB_0$, we must have $\lambda_i - 1 \le \frac{\eps}{8 n_2}$.
Using this observation and the definition of the sets $\cB_{\ell}, \forall \ell \in \brk{\ell_{\max}}$, we get:
\begin{align}
    \qquad\quad \sum\limits_{i \in \brk{d}} \paren{\lambda_i - 1} \iprod{v_i, s}^2
    &= \sum\limits_{\ell \in \brk{\ell_{\max}}} \sum\limits_{i \in \cB_{\ell}} \paren{\lambda_i - 1} \iprod{s, v_i}^2 + \sum\limits_{i \in \cB_0} \paren{\lambda_i - 1} \iprod{s, v_i}^2 \nonumber \\
    &\le \sum\limits_{\ell \in \brk{\ell_{\max}}} \sum\limits_{i \in \cB_{\ell}} \paren{\lambda_i - 1} \iprod{s, v_i}^2 + \frac{\eps}{8 n_2} \sum\limits_{i \in \cB_0} \iprod{s, v_i}^2 \nonumber \\
    &\le \gamma \paren{1 + 2 \gamma} \sum\limits_{\ell \in \brk{\ell_{\max}}} \sum\limits_{i \in \cB_{\ell}} \frac{\iprod{s, v_i}^2}{2^{\ell - 1}} + \frac{\eps}{8 n_2}, \label{eq:result_to_prove_help}
\end{align}
where we used the fact that $\brc{v_i}_{i \in \brk{d}}$ is an orthonormal basis of $\bR^d$ to argue that:
\[
    \sum\limits_{i \in \cB_0} \iprod{s, v_i}^2 \le \sum\limits_{i \in \brk{d}} \iprod{s, v_i}^2 = \llnorm{s}^2 \le 1.
\]
Thus, based on (\ref{eq:result_to_prove_help}), to satisfy (\ref{eq:result_to_prove}), it suffices to have:
\[
    \sum\limits_{\ell \in \brk{\ell_{\max}}} \sum\limits_{i \in \cB_{\ell}} \frac{\iprod{s, v_i}^2}{2^{\ell - 1}} \le \frac{\eps}{8 \gamma \paren{1 + 2 \gamma} n_2} = \frac{\eps}{64 e^2 \lambda_0 \paren{1 + 2 \gamma}},
\]
where the equality used the fact that $\gamma \coloneqq \frac{8 e^2 \lambda_0}{n_2}$.

Our assumption suffices for this to hold.
Indeed, we have:
\[
    \sum\limits_{\ell \in \brk{\ell_{\max}}} \sum\limits_{i \in \cB_{\ell}} \iprod{s, v_i}^2 \le \frac{2^{\ell - 1} \eps}{96 e^2 \lambda_0} \le \frac{\eps}{64 e^2 \lambda_0 \paren{1 + 2 \gamma}},
\]
where the last inequality follows from the assumption that $\gamma \le \frac{\eps}{4} \le \frac{1}{4}$.
\end{proof}

We are now ready to prove our original version of Proposition~\ref{prop:priv_anal_multi1}.
The proof has virtually the same structure as for that proposition.
However, instead of upper-bounding the probability of $s^{\top} \paren{\widehat{\Sigma}^{\frac{1}{2}} \widehat{\Sigma}'^{- 1} \widehat{\Sigma}^{\frac{1}{2}} - \id} s > \frac{\eps}{4 n_2}$ occurring (which would require the argument of Section~\ref{subsubsec:priv_anal_multi1}), we will instead use Fact~\ref{fact:beta_concentration} to upper-bound the probability of the event that corresponds to the assumption in Lemma~\ref{lem:priv_anal_multi1_suff_cond_alt} failing.

\begin{proposition}
\label{prop:priv_anal_multi1_alt}
Let $\eps \in \brk{0, 1}$ and $\delta \in \brk{0, \frac{\eps}{10}}, \lambda_0 \geq d$.
Let us assume that:
\[
    n_2 \geq \useconstant{C2} \frac{\lambda_0 \log\paren{\lambda_0} \paren{\log\paren{\log\paren{\lambda_0}} + \log\paren{\frac{1}{\delta}}}}{\eps},
\]
for some appropriately large absolute constant $\useconstant{C2} \geq 1$.
Finally, let $X, X' \in \bR^{n \times d}$ be adjacent datasets such that $\Psi\paren{X} = \Psi\paren{X'} = \Pass$.
Then, for $z \sim \cU\paren{\bS^{n_2 - 1}}$, we have:
\[
    D_{e^{\frac{\eps}{2}}}\paren{\widehat{\mu}' + \sqrt{\paren{1 - \frac{1}{n_1}} n_2} W z \middle\| \widehat{\mu}' + \sqrt{\paren{1 - \frac{1}{n_1}} n_2} W' z} \le \delta.
\]
\end{proposition}

\begin{proof}
As in the proof of Proposition~\ref{prop:priv_anal_multi1}, we have:
\begin{align}
    &\quad\ D_{e^{\frac{\eps}{2}}}\paren{\widehat{\mu}' + \sqrt{\paren{1 - \frac{1}{n_1}} n_2} W z \middle\| \widehat{\mu}' + \sqrt{\paren{1 - \frac{1}{n_1}} n_2} W' z} \nonumber \\
    &\le D_{e^{\frac{\eps}{2}}}\paren{z_{\le d} \middle\| T} \nonumber\\
    &\le \pr{t \sim z_{\le d}}{f_{z_{\le d}}\paren{t} > e^{\frac{\eps}{2}} f_T\paren{t}} \nonumber\\
    &\le \pr{s \sim z_{\le d}}{\brc{s^{\top} \widehat{\Sigma}^{\frac{1}{2}} \widehat{\Sigma}'^{- 1} \widehat{\Sigma}^{\frac{1}{2}} s > \frac{1}{2}} \cup \brc{\exists \ell \in \brk{\ell_{\max}} \colon \sum\limits_{i \in \cB_{\ell}} \iprod{s, v_i}^2 > \frac{2^{\ell - 1} \eps}{96 e^2 \lambda_0 \ell_{\max}}}} \nonumber \\
    &\le \pr{s \sim z_{\le d}}{s^{\top} \widehat{\Sigma}^{\frac{1}{2}} \widehat{\Sigma}'^{- 1} \widehat{\Sigma}^{\frac{1}{2}} s > \frac{1}{2}} + \sum\limits_{\ell \in \brk{\ell_{\max}}} \pr{s \sim z_{\le d}}{\sum\limits_{i \in \cB_{\ell}} \iprod{s, v_i}^2 > \frac{2^{\ell - 1} \eps}{96 e^2 \lambda_0 \ell_{\max}}}. \label{eq:union_bound_ub_alt}
\end{align}
The process to upper-bound the first term is exactly the same as the one we followed in the proof of Proposition~\ref{prop:priv_anal_multi1}, so we point readers to (\ref{eq:term_bound1}), and do not repeat it here.
It remains to show that the sum over $\ell$ in (\ref{eq:union_bound_ub_alt}) is upper-bounded by $\frac{\delta}{2}$.
To establish that, we will show that individual terms of the sum are upper-bounded by $\frac{\delta}{2 \ell_{\max}}$.

We note that, for any $\ell \in \brk{\ell_{\max}}$, the sum $\sum\limits_{i \in \cB_{\ell}} \iprod{s, v_i}^2$ is distributed as $\Beta\paren{\frac{\abs{\cB_{\ell}}}{2}, \frac{n_2 - \abs{\cB_{\ell}}}{2}}$.
This follows from Facts~\ref{fact:unit_sphere_proj_dist} and~\ref{fact:uniform_rot_inv}.
Additionally, by the definition of the sets $\cB_{\ell}$, it must be the case that $\abs{\cB_{\ell}} \le 2^{\ell}, \forall \ell \in \brk{\ell_{\max}}$, since otherwise (\ref{eq:eig_constraint_restated}) would be violated.
Observe that the mean of the distribution $\Beta\paren{\frac{\abs{\cB_{\ell}}}{2}, \frac{n_2 - \abs{\cB_{\ell}}}{2}}$ is $\frac{\abs{\cB_{\ell}}}{n_2}$.
Based on the previous remark, the mean is upper-bounded by $\frac{2^{\ell}}{n_2}$ for all $\ell \in \brk{\ell_{\max}}$.
By our sample complexity bound, this is at most $\frac{2^{\ell - 1} \eps}{192 e^2 \lambda_0 \ell_{\max}}, \forall \ell \in \brk{\ell_{\max}}$.
Thus, we can appeal to Fact~\ref{fact:beta_concentration}, and obtain the bound:
\begin{align*}
    &\quad\ \pr{s \sim z_{\le d}}{\sum\limits_{i \in \cB_{\ell}} \iprod{s, v_i}^2 > \frac{2^{\ell - 1} \eps}{96 e^2 \lambda_0 \ell_{\max}}} \\
    &\le 2 \exp\paren{- \useconstant{cbeta} \min\brc{\frac{\paren{n_2 - 2^{\ell}}^2}{2 \cdot 2^{\ell}} \paren{\frac{2^{\ell - 1} \eps}{96 e^2 \lambda_0 \ell_{\max}} - \frac{2^{\ell}}{n_2}}^2, \frac{n_2 - 2^{\ell}}{2} \paren{\frac{2^{\ell - 1} \eps}{96 e^2 \lambda_0 \ell_{\max}} - \frac{2^{\ell}}{n_2}}}} \\
    &\le 2 \exp\paren{- \useconstant{cbeta} \frac{n_2 - 2^{\ell}}{2} \cdot \frac{2^{\ell - 1} \eps}{192 e^2 \lambda_0 \ell_{\max}}} \\
    &\le \frac{\delta}{2 \ell_{\max}},
\end{align*}
where the last inequality follows from our bound on $n_2$.
\end{proof}

\end{document}